\documentclass[11pt]{article}

\newlength{\actualtopmargin}
\newlength{\actualsidemargin}
\usepackage[tbtags]{amsmath}
\usepackage{amsthm}
\usepackage{amssymb}
\usepackage{amsfonts}
\usepackage[frame, line, arrow, matrix, arc]{xy}
\usepackage{hyperref}

\theoremstyle{plain}
  \newtheorem{theorem}{Theorem}
  \newtheorem{lemma}[theorem]{Lemma}
  \newtheorem{corollary}[theorem]{Corollary}
  \newtheorem{proposition}[theorem]{Proposition}
  
\theoremstyle{definition}
  \newtheorem{definition}[theorem]{Definition}

\theoremstyle{remark}
  \newtheorem*{remark}{Remark}
\theoremstyle{plain}
  \newtheorem*{theorem*}{Theorem}
  \newtheorem*{lemma*}{Lemma}
  \newtheorem*{corollary*}{Corollary}
  \newtheorem*{proposition*}{Proposition}
  \newtheorem*{claim*}{Claim}

\newenvironment{step}
  {
    \begin{enumerate}

  }
  {\end{enumerate}}

\newenvironment{algorithm*}[1]
  {
    \begin{center}
      \hrulefill\\
      \textbf{#1}\\
  }
  {
    \vspace{-\baselineskip}
    \hrulefill
    \end{center}
  }

\newenvironment{protocol*}[1]
  {
    \begin{center}
      \hrulefill\\
      \textbf{#1}\\
  }
  {
    \vspace{-\baselineskip}
    \hrulefill
    \end{center}
  }

\newlength{\itemwidth}
\newlength{\descriptionwidth}
\newenvironment{promiseproblem*}[4]
  {
    \begin{center}
      \hrulefill\\
      \textbf{\textsc{#1}}\\
      \settowidth{\itemwidth}{\textbf{Yes Instances:}}
      \setlength{\descriptionwidth}{\textwidth}
      \addtolength{\descriptionwidth}{-\itemwidth}
      \addtolength{\descriptionwidth}{-\labelsep}
      \begin{description}
        \item[\parbox{\itemwidth}{Input:}]
          \parbox[t]{\descriptionwidth}{#2}
        \item[\parbox{\itemwidth}{Yes Instances:}]
          \parbox[t]{\descriptionwidth}{#3}
        \item[\parbox{\itemwidth}{No Instances:}]
          \parbox[t]{\descriptionwidth}{#4\\}
      \end{description}
  }
  {
    \vspace{-1.1\baselineskip}
    \hrulefill
    \end{center}
  }

\newcommand{\bbC}{\mathbb{C}}

\newcommand{\bbN}{\mathbb{N}}

\newcommand{\bbR}{\mathbb{R}}

\newcommand{\bbZ}{\mathbb{Z}}

\newcommand{\bfU}{\mathbf{U}}

\newcommand{\calC}{\mathcal{C}}

\newcommand{\classfont}{\mathrm}
\newcommand{\problemfont}{\textsc}

\newcommand{\spacefont}{\mathcal}
\newcommand{\registerfont}{\mathsf}

\newcommand{\spaceH}{\spacefont{H}}

\newcommand{\regC}{\registerfont{C}}

\newcommand{\regO}{\registerfont{O}}

\newcommand{\regQ}{\registerfont{Q}}
\newcommand{\regR}{\registerfont{R}}

\newcommand{\regX}{\registerfont{X}}
\newcommand{\regY}{\registerfont{Y}}

\newcommand{\classP}{\classfont{P}}
\newcommand{\co}{\classfont{co}\textrm{-}}
\newcommand{\NP}{\classfont{NP}}

\newcommand{\BPP}{\classfont{BPP}}
\newcommand{\BQP}{\classfont{BQP}}
\newcommand{\PP}{\classfont{PP}}
\newcommand{\SBP}{\classfont{SBP}}
\newcommand{\SBQP}{\classfont{SBQP}}
\newcommand{\postBPP}{\classfont{PostBPP}}
\newcommand{\postBQP}{\classfont{PostBQP}}
\newcommand{\MA}{\classfont{MA}}
\newcommand{\AM}{\classfont{AM}}
\newcommand{\NQP}{\classfont{NQP}}
\newcommand{\CequalP}{{\classfont{C}_=\classfont{P}}}
\newcommand{\coCequalP}{{\co\CequalP}}
\newcommand{\PH}{\classfont{PH}}
\newcommand{\AC}{\classfont{AC}}

\newcommand{\QfP}[1]{{\classfont{Q}_{[#1]}\classfont{P}}}

\newcommand{\BQfP}[1]{{\classfont{BQ}_{[#1]}\classfont{P}}}
\newcommand{\NQfP}[1]{{\classfont{NQ}_{[#1]}\classfont{P}}}
\newcommand{\SBQfP}[1]{{\classfont{SBQ}_{[#1]}\classfont{P}}}
\newcommand{\QOfP}[1]{{\classfont{Q}_{#1}\classfont{P}}}

\newcommand{\BQOfP}[1]{{\classfont{BQ}_{#1}\classfont{P}}}

\newcommand{\QlogP}{\QOfP{\textlog}}

\newcommand{\BQlogP}{\BQOfP{\textlog}}

\newcommand{\QoneP}{\QfP{1}}

\newcommand{\BQoneP}{\BQfP{1}}
\newcommand{\NQoneP}{\NQfP{1}}
\newcommand{\SBQoneP}{\SBQfP{1}}
\newcommand{\QtwoP}{\QfP{2}}

\newcommand{\BP}{\classfont{BP}}

\newcommand{\NOT}{\mathrm{NOT}}
\newcommand{\CNOT}{\mathrm{CNOT}}

\newcommand{\INCR}[1]{{\op{U}_{+1}(\Integers_{#1})}}
\newcommand{\THRESHOLD}[2]{{\op{U}_{\geq {#1}}(\Integers_{#2})}}

\newcommand{\threshold}[1]{{\op{f}_{\geq {#1}}}}

\newcommand{\TrEst}{\problemfont{TrEst}}

\newcommand{\op}[1]{\operatorname{\mathnormal{#1}}}

\newcommand{\tensor}{\otimes}

\newcommand{\intersection}{\cap}

\newcommand{\Expect}{\operatorname{E}}

\newcommand{\probability}[1]{\Pr [ #1 ]}
\newcommand{\prob}{\probability}

\newcommand{\biggprobability}[1]{\Pr \biggl[ #1 \biggr]}
\newcommand{\biggprob}{\biggprobability}

\newcommand{\expectation}[1]{\Expect [ #1 ]}
\newcommand{\expect}{\expectation}

\newcommand{\Biggexpectation}[1]{\Expect \Biggl[ #1 \Biggr]}
\newcommand{\Biggexpect}{\Biggexpectation}

\newcommand{\bra}[1]{\langle #1 \rvert}

\newcommand{\ket}[1]{\lvert #1 \rangle}

\newcommand{\bigket}[1]{\bigl\lvert #1 \bigr\rangle}

\newcommand{\ketbra}[1]{\lvert #1 \rangle \langle #1 \rvert}

\newcommand{\conjugate}[1]{{#1^\dagger}}

\newcommand{\tr}{\operatorname{tr}}

\newcommand{\bignorm}[1]{\bigl\lVert #1 \bigr\rVert}

\newcommand{\abs}[1]{\lvert #1 \rvert}

\newcommand{\bigabs}[1]{\bigl\lvert #1 \bigr\rvert}

\newcommand{\ceil}[1]{\lceil #1 \rceil}
\newcommand{\ceilL}[1]{\left\lceil #1 \right\rceil}

\newcommand{\lrangle}[1]{\langle #1 \rangle}

\newcommand{\function}[3]{{#1 \colon #2 \to #3}}
\newcommand{\set}[2]{{\{ #1 \colon #2 \}}}

\newcommand{\Complex}{\bbC}
\newcommand{\Real}{\bbR}

\newcommand{\Natural}{\bbN}
\newcommand{\Integers}{\bbZ}
\newcommand{\Nonnegative}{{\Integers^+}}

\newcommand{\Binary}{{\{ 0, 1 \}}}

\newcommand{\Unitary}{\bfU}

\newcommand{\controlled}{\operatorname{\Lambda}}

\newcommand{\textlog}{\mathrm{log}}

\newcommand{\acc}{\mathrm{acc}}

\newcommand{\init}{\mathrm{init}}
\newcommand{\final}{\mathrm{final}}
\newcommand{\yes}{\mathrm{yes}}
\newcommand{\no}{\mathrm{no}}

\newcommand{\DQC}[1]{{\textrm{DQC}#1}}

\newcommand{\metermark}
{%
  \setlength{\unitlength}{1pt}
  \begin{picture}(12.44, 6.58)(  0,   -0.62)
    \hspace{-2.575pt}%
    \begin{xy}
      {\ar @{-} <6.22pt, -0.62pt>;p+<4pt, 7.2pt>},
      {<6.22pt, -6.22pt> *\cir <8.8pt>{ur_dr}}
    \end{xy}
  \end{picture}
}

\newcommand{\smallmetermark}
{%
  \setlength{\unitlength}{1pt}
  \begin{picture}(11.08, 5.86)(  0,   -0.56)
    \hspace{-2.292pt}%
    \begin{xy}
      {\ar @{-} <5.54pt, -0.55pt>;p+<3.56pt, 6.41pt>},
      {<5.54pt, -5.54pt> *\cir <7.83pt>{ur_dr}}
    \end{xy}
  \end{picture}
}

\newcommand{\ignore}[1]{}

\begin{document}

\sloppy


\title{
  \Large{
    \textbf{
         Power of Quantum Computation with Few Clean Qubits
    }
  }
}

\author{
  Keisuke Fujii\footnotemark[1]\\
  \and
  Hirotada Kobayashi\footnotemark[2]\\
  \and
  Tomoyuki Morimae\footnotemark[3]\\
  \and
  Harumichi Nishimura\footnotemark[4]\\
  \and
  Shuhei Tamate\footnotemark[2]\\
  \and
  Seiichiro Tani\footnotemark[5]
}

\date{}

\maketitle
\thispagestyle{empty}
\pagestyle{plain}
\setcounter{page}{0}

\renewcommand{\thefootnote}{\fnsymbol{footnote}}

\vspace{-5mm}

\begin{center}
  \large{
    \footnotemark[1]%
    The Hakubi Center for Advanced Research,\\
    Kyoto University,
    Kyoto, Japan\\
    [2.5mm]
    \footnotemark[2]%
    Principles of Informatics Research Division,\\
    National Institute of Informatics,
    Tokyo, Japan\\
    [2.5mm]
    \footnotemark[3]%
    Advanced Scientific Research Leaders Development Unit,\\
    Gunma University,
    Kiryu, Gunma, Japan\\
    [2.5mm]
    \footnotemark[4]%
    Department of Computer Science and Mathematical Informatics,\\
    Graduate School of Information Science,\\
    Nagoya University,
    Nagoya, Aichi, Japan\\
    [2.5mm]
    \footnotemark[5]%
    NTT Communication Science Laboratories,\\
    NTT Corporation,
    Atsugi, Kanagawa, Japan
  }\\
  [5mm]
  \large{23 September 2015}\\
  [8mm]
\end{center}

\renewcommand{\thefootnote}{\arabic{footnote}}




\begin{abstract}
  This paper investigates the power of polynomial-time quantum computation
  in which only a very limited number of qubits are initially clean in the $\ket{0}$~state,
  and all the remaining qubits are initially in the totally mixed state.
  No initializations of qubits are allowed during the computation, nor intermediate measurements.
  The main results of this paper are unexpectedly strong error-reducible properties of such quantum computations.
  It is proved that any problem solvable by a polynomial-time quantum computation 
  with one-sided bounded error that uses logarithmically many clean qubits
  can also be solvable
  with exponentially small one-sided error using just two clean qubits,
  and with polynomially small one-sided error using just one clean qubit
  (which in particular implies the solvability with any small constant one-sided error). 
  It is further proved in the case of two-sided bounded error
  that any problem solvable by such a computation
  with a constant gap between completeness and soundness using logarithmically many clean qubits
  can also be solvable with exponentially small two-sided error using just two clean qubits.
  If only one clean qubit is available,
  the problem is again still solvable with exponentially small error 
  in one of the completeness and soundness and polynomially small error in the other. 
  As an immediate consequence of the above result for the two-sided-error case,
  it follows that the \problemfont{Trace Estimation} problem
  defined with \emph{fixed} constant threshold parameters
  is complete for $\BQlogP$ and $\BQoneP$,
  the classes of problems solvable by polynomial-time quantum computations
  with completeness~$2/3$ and soundness~$1/3$
  using logarithmically many clean qubits and just one clean qubit, respectively.  
  The techniques used for proving the error-reduction results
  may be of independent interest in themselves,
  and one of the technical tools can also be used to show
  the hardness of weak classical simulations of one-clean-qubit computations
  (i.e., $\DQC{1}$ computations).
\end{abstract}

\clearpage


\section{Introduction}
\label{Section: introduction}


\subsection{Background}
\label{Subsection: background}

An inherent nature of randomized and quantum computing
is that the outcome of a computation is probabilistic and may not always be correct.
Error reduction, or success-probability amplification,
is thus one of the most fundamental issues in randomized and quantum computing.
Computation error can be efficiently reduced to be negligibly small
in many standard computation models
via a simple repetition followed by an OR-type, or an AND-type, or a threshold-value decision,
depending on whether the error can happen in the original computation
only for yes-instances, or only for no-instances, or for both.
Typical examples are polynomial-time randomized and quantum computations with bounded error,
and in particular,
the error can be made exponentially small in $\BPP$ and $\BQP$ both in completeness and in soundness,
which provides a reasonable ground for the well-used definitions of $\BPP$ and $\BQP$
that employ bounds~$2/3$~and~$1/3$ for completeness and soundness, respectively.
In many other computation models, however, it is unclear whether the error can be reduced efficiently
by the simple repetition-based method mentioned above,
and more generally, whether error reduction itself is possible.
Such a situation often occurs
when a computation model involves communications with some untrusted parties,
like interactive proof systems.
For instance, the simple repetition-based method does work for quantum interactive proofs
in the one-sided error case of perfect completeness,
but its proof is highly nontrivial~\cite{KitWat00STOC, Gut09PhD}.
Moreover, a negative evidence is known in the two-sided error case
that the error may not be reduced efficiently via the above-mentioned simple method
of repetition with threshold-value decision~\cite{MolWat12RSPA},
although error reduction itself is anyway possible in this case,
as any two-sided-error quantum interactive proof system can be made to have perfect completeness
by adding more communication turns~\cite{KitWat00STOC, KobLeGNis15SIComp}
(and the number of communication turns then can be reduced to three).
Another situation where error reduction becomes nontrivial (and sometimes impossible)
appears when a computation model can use only very limited computational resources,
like space-bounded computations.
If the resources are too limited,
it is simply impossible to repeat the original computation sufficiently many times,
which becomes a quite enormous obstacle to error reduction
in the case of space-bounded quantum computations
when initializations of qubits are disallowed after the computation starts.
Indeed, it is known impossible in the case of one-way quantum finite state automata
to reduce computation error smaller than some constant~\cite{AmbFre98FOCS}.
Also, it is unclear whether error reduction is possible or not
in various logarithmic-space quantum computations.
For computations of one-sided bounded error performed by logarithmic-space quantum Turing machines,
Watrous~\cite{Wat01JCSS} presented a nontrivial method that reduces the error to be exponentially small.
Other than this result,
error-reduction techniques have not been developed much for space-bounded quantum computations.

Another well-studied model of quantum computing with limited computational resources is 
the \emph{deterministic quantum computation with one quantum bit~($\DQC{\mathit{1}}$)},
often mentioned as the \emph{one-clean-qubit model},
which was introduced by Knill~and~Laflamme~\cite{KniLaf98PRL},
originally motivated by nuclear magnetic resonance (NMR) quantum information processing.
A $\DQC{1}$ computation over $w$~qubits starts with the initial state of
the totally mixed state except for a single clean qubit,
namely,
${\ketbra{0} \tensor \bigl( \frac{I}{2} \bigr)^{\tensor (w - 1)}}$.
After applying a polynomial-size unitary quantum circuit to this state,
only a single output qubit is measured in the computational basis
at the end of the computation in order to read out the computation result.
The $\DQC{1}$ model may be viewed as a variant of space-bounded quantum computations in a sense,
and is believed not to have full computational power of the standard polynomial-time quantum computation.
Indeed, it is known strictly less powerful than the standard polynomial-time quantum computation
under some reasonable assumptions~\cite{AmbSchVaz06JACM}.
Moreover, since any quantum computation
over the totally mixed state~${\bigl( \frac{I}{2} \bigr)^{\tensor w}}$
is trivially simulatable by a classical computation,
the $\DQC{1}$ model looks easy to classically simulate at first glance.
Surprisingly, however,
the model turned out to be able to efficiently solve several problems
for which no efficient classical algorithms are known,
such as calculations of
the spectral density~\cite{KniLaf98PRL},
an integrability tester~\cite{PouLafMilPaz03PRA},
the fidelity decay~\cite{PouBluLafOll04PRL},
Jones and HOMFLY polynomials~\cite{ShoJor08QIC, JorWoc09QIC},
and an invariant of 3-manifolds~\cite{JorAla11TQC}.
More precisely, 
$\DQC{1}$ computations can solve the decisional versions of these problems
with two-sided bounded error.
Computation error can be quite large in such computations,
and in fact, the gap between completeness and soundness is allowed to be polynomially small.
The only method known for amplifying success probability of these computations
is to sequentially repeat an attempt of the computation polynomially many times,
but this requires the clean qubit to be initialized every time after finishing one attempt,
and moreover, the result of each attempt must be recorded to classical work space
prepared outside of the $\DQC{1}$ model.
It is definitely more desirable if computation error can be reduced without such initializations.
The situation is similar
even when the number of clean qubits is allowed to be logarithmically many with respect to the input length.
It is also known that any quantum computation of two-sided bounded error
that uses logarithmically many clean qubits
can be simulated by a quantum computation still of two-sided bounded error
that uses just one clean qubit,
but the known method for this simulation considerably increases the computational error,
and the gap between completeness and soundness becomes polynomially small.


\subsection{Main Results}

This paper develops methods of reducing
computation error in quantum computations with few clean qubits,
including the $\DQC{1}$ model.
As will be presented below, the methods proposed in this paper are unexpectedly powerful
and provide an almost fully satisfying solution in the cases of one-sided bounded error,
both to the reducibility of computation error and to the reducibility of the number of clean qubits. 
In the case of two-sided bounded error,
the methods in this paper are applicable
only when there is a constant gap between completeness and soundness in the original computation,
but still significantly improve the situation of quantum computations with few clean qubits.

Let ${\QlogP(c,s)}$, ${\QoneP(c,s)}$, and ${\QtwoP(c,s)}$ denote the class of problems
solvable by a polynomial-time quantum computation with completeness~$c$ and soundness~$s$
that uses logarithmically many clean qubits,
one clean qubit, and two clean qubits, respectively.
The rigorous definitions of these complexity classes will be found in Subsection~\ref{Subsection: DQCk computation and complexity classes}.

First, in the case of one-sided bounded error,
it is proved that
any problem solvable by a polynomial-time quantum computation with one-sided bounded error
that uses logarithmically many clean qubits
can also be solvable by that with exponentially small one-sided error using just two clean qubits. 

\begin{theorem}
  For any polynomially bounded function~$\function{p}{\Nonnegative}{\Natural}$
  and any polynomial-time computable function~$\function{s}{\Nonnegative}{[0,1]}$
  satisfying ${1 - s \geq \frac{1}{q}}$ for some polynomially bounded function~$\function{q}{\Nonnegative}{\Natural}$,
  \[
    \QlogP(1, s) \subseteq \QtwoP(1, 2^{-p}).
  \]
  \vspace{-0.341\baselineskip}
  \label{Theorem: error reduction for perfect-completeness case}
\end{theorem}

If only one clean qubit is available,
the problem is still solvable with polynomially small one-sided error
(which in particular implies the solvability with any small constant one-sided error).

\begin{theorem}
  For any polynomially bounded function~$\function{p}{\Nonnegative}{\Natural}$
  and any polynomial-time computable function~$\function{s}{\Nonnegative}{[0,1]}$
  satisfying ${1 - s \geq \frac{1}{q}}$ for some polynomially bounded function~$\function{q}{\Nonnegative}{\Natural}$,
  \[
    \QlogP(1, s) \subseteq \QoneP \Bigl( 1, \, \frac{1}{p} \Bigr).
  \]
  \vspace{-0.341\baselineskip}
  \label{Theorem: co-RQlogP is in co-RQ1P}
\end{theorem}

The above two theorems are for the one-sided error case of perfect completeness,
and similar statements hold even for the case of perfect soundness,
by considering the complement of the problem.

\begin{corollary}
  For any polynomially bounded function~$\function{p}{\Nonnegative}{\Natural}$
  and any polynomial-time computable function~$\function{c}{\Nonnegative}{[0,1]}$
  satisfying ${c \geq \frac{1}{q}}$ for some polynomially bounded function~$\function{q}{\Nonnegative}{\Natural}$,
  the following two properties hold:
  \begin{itemize}
  \item[(i)]
    ${
      \QlogP(c, 0) \subseteq \QtwoP(1 - 2^{-p}, 0)
    }$,
  \item[(ii)]
    ${
      \displaystyle{\QlogP(c, 0) \subseteq \QoneP \Bigl( 1 - \frac{1}{p}, \, 0 \Bigr)}
    }$.
  \end{itemize}
  \label{Corollary: error reduction for perfect-soundness cases}
\end{corollary}

\clearpage 

In the case of two-sided bounded error,
similar statements are proved on the condition that
there is a constant gap between completeness and soundness in the original computation.
Namely, it is proved that any problem solvable by a polynomial-time quantum computation
that uses logarithmically many clean qubits and has a constant gap between completeness and soundness
can also be solvable by that with exponentially small two-sided error using just two clean qubits.

\begin{theorem}
  For any polynomially bounded function~$\function{p}{\Nonnegative}{\Natural}$
  and any constants~$c$~and~$s$ in $\Real$
  satisfying ${0 < s < c < 1}$,
  \[
    \QlogP(c, s) \subseteq \QtwoP(1 - 2^{-p}, 2^{-p}).
  \]
  \vspace{-0.341\baselineskip}
  \label{Theorem: error reduction for two-sided-error case}
\end{theorem}

If only one clean qubit is available,
the problem is again still solvable
with exponentially small error in one of the completeness and soundness
and polynomially small error in the other.

\begin{theorem}
  For any polynomially bounded function~$\function{p}{\Nonnegative}{\Natural}$
  and any constants~$c$~and~$s$ in $\Real$
  satisfying ${0 < s < c < 1}$,
  \[
    \QlogP(c, s)
    \subseteq
    \QoneP \Bigl( 1 - 2^{-p}, \, \frac{1}{p} \Bigr)
    \!\: \intersection \!\:
    \QoneP \Bigl( 1 - \frac{1}{p}, \, 2^{-p} \Bigr).
  \]
  \vspace{-0.341\baselineskip}
  \label{Theorem: BQlogP with constant gap is in BQ1P}
\end{theorem}

The ideas for the proofs of these statements will be overviewed in Section~\ref{Section: overview of error reduction results}.
The techniques developed in the proofs may be of independent interest in themselves,
and one of the technical tools can be used to show
the hardness of weak classical simulations of $\DQC{1}$ computations,
as will be summarized below.


\subsection{Further Results}

\paragraph{Completeness results for \problemfont{Trace Estimation} problem} 

Define the complexity classes $\BQlogP$ and $\BQoneP$
by ${\BQlogP =\QlogP \bigl(\frac{2}{3}, \, \frac{1}{3} \bigr)}$
and ${\BQoneP =\QoneP \bigl(\frac{2}{3}, \, \frac{1}{3} \bigr)}$,
respectively.
An immediate consequence of Theorem~\ref{Theorem: BQlogP with constant gap is in BQ1P}
is that the \problemfont{Trace Estimation} problem is complete for $\BQlogP$ and $\BQoneP$
under polynomial-time many-one reduction,
even when the problem is defined with \emph{fixed} constant parameters
that specify the bounds on normalized traces in the yes-instance and no-instance cases.

Given a description of a quantum circuit that specifies a unitary transformation~$U$,
the \problemfont{Trace Estimation} problem specified with
two parameters~$a$~and~$b$ satisfying ${-1 \leq b < a \leq 1}$
is the problem of deciding
whether the real part of the normalized trace of $U$ is at least $a$ or it is at most $b$.

\begin{promiseproblem*}
  {Trace Estimation Problem: $\boldsymbol{\TrEst(a,b)}$}
  {
    A description of a quantum circuit~$Q$ that implements a unitary transformation~$U$ over $n$~qubits.
  }
  {
    ${\frac{1}{2^n} \Re(\tr U) \geq a}$.
  }
  {
    ${\frac{1}{2^n} \Re(\tr U) \leq b}$.
  }
\end{promiseproblem*}

The paper by Knill~and~Laflamme~\cite{KniLaf98PRL} that introduced the $\DQC{1}$ model
already pointed out that this problem is closely related to the $\DQC{1}$ computation.
This point was further clarified in the succeeding literature
(see Refs.~\cite{She06arXiv, She09PhD, ShoJor08QIC}, for instance).
More precisely,
consider a variant of the \problemfont{Trace Estimation} problem
where the two parameters~$a$~and~$b$ may depend on the input length
(i.e., the length of the description of $Q$).
It is known that
this version of the \problemfont{Trace Estimation} problem,
for any $a$ and $b$ such that the gap~${a - b}$ is bounded from below by an inverse-polynomial
with respect to the input length,
can be solved by a $\DQC{1}$ computation with \emph{some} two-sided bounded error
where the completeness and soundness parameters~$c$~and~$s$ depend on $a$ and $b$.
It is also known that,
for any two nonnegative parameters~$a$~and~$b$
such that the gap~${a - b}$ is bounded from below by an inverse-polynomial
with respect to the input length,
the corresponding version of the \problemfont{Trace Estimation} problem
is hard for the complexity class~${\QoneP(c,s)}$
for \emph{some} completeness parameter~$c$ and soundness parameter~$s$ that depend on $a$ and $b$.
Hence,
the \problemfont{Trace Estimation} problem essentially characterizes
the power of the $\DQC{1}$ computation,
except the following subtle point.
One thing to be pointed out in the existing arguments above is that,
when the parameters~$a$~and~$b$ are fixed for the \problemfont{Trace Estimation} problem,
the completeness~$c$ and soundness~$s$ with which the problem is in ${\QoneP(c,s)}$
are \emph{different} from
the completeness~$c'$ and soundness~$s'$ with which the problem is hard for ${\QoneP(c',s')}$.
Namely,
given two nonnegative parameters~$a$~and~$b$ of the problem,
the computation solves the problem
with completeness~${c = (1 + a) / 2}$ and soundness~${s = (1 + b) / 2}$,
while the problem is hard for the class
with completeness~${c' = a/4}$ and soundness~${s' = b/4}$.
Therefore, the existing arguments are slightly short
for proving $\BQoneP$-completeness of the \problemfont{Trace Estimation} problem
with fixed parameters~$a$~and~$b$
(and ${\QoneP(c, s)}$-completeness of that
for fixed completeness and soundness parameters~$c$~and~$s$, in general).

In contrast, with Theorem~\ref{Theorem: BQlogP with constant gap is in BQ1P} in hand,
it is immediate to show that
the \problemfont{Trace Estimation} problem is complete for $\BQlogP$ and for $\BQoneP$
for any constants~$a$~and~$b$ satisfying ${0 < b < a < 1}$.

\begin{theorem}
  For any constants~$a$~and~$b$ in $\Real$ satisfying ${0 < b < a < 1}$,
  ${\TrEst(a, b)}$ is complete for $\BQlogP$ and for $\BQoneP$
  under polynomial-time many-one reduction.
  \label{Theorem: completeness of the trace estimation problem}
\end{theorem}

\paragraph{Hardness of weak classical simulations of $\boldsymbol{\DQC{1}}$ computation}

Recently, quite a few number of studies
focused on the hardness of \emph{weak} classical simulations
of restricted models of quantum computing under some reasonable assumptions~\cite{TerDiV04QIC, BreJozShe11RSPA, AarArk13ToC, NiVan13QIC, JozVan14QIC, MorFujFit14PRL, TakYamTan14QIC, Bro15PRA, TakTanYamTan15COCOON}.
Namely, a plausible assumption in complexity theory
leads to the impossibility of efficient sampling by a classical computer
according to an output probability distribution generatable with a quantum computing model.
Among them are the IQP model~\cite{BreJozShe11RSPA} and the Boson sampling~\cite{AarArk13ToC},
both of which are proved hard for classical computers to simulate within multiplicative error,
unless the polynomial-time hierarchy collapses to the third level
(in fact, the main result of Ref.~\cite{AarArk13ToC} is a much more meaningful hardness result
on the weak simulatability of the Boson sampling within \emph{polynomially small additive error},
but which needs a much stronger complexity assumption than the collapse of polynomial-time hierarchy).

An interesting question to ask is whether a similar result holds even for the $\DQC{1}$ model.
Very recently, Morimae, Fujii, and Fitzsimons~\cite{MorFujFit14PRL} approached to answering the question.
They focused on the $\DQC{1}_m$-type computation, the generalization of the $\DQC{1}$ model
that allows $m$~output~qubits to be measured at the end of the computation,
and proved that a $\DQC{1}_m$-type computation with ${m \geq 3}$
cannot be simulated within multiplicative error
unless the polynomial-time hierarchy collapses to the third level.
Their proof essentially shows that any $\postBQP$ circuit
can be simulated by a $\DQC{1}_3$-type computation,
where $\postBQP$ is the complexity class
corresponding to bounded-error quantum polynomial-time computations with postselection,
which is known equivalent to $\PP$~\cite{Aar05RSPA}.
By an argument similar to that in Ref.~\cite{BreJozShe11RSPA},
it follows that $\PP$ is in $\postBPP$ (the version of $\BPP$ with postselection),
if the $\DQC{1}_3$-type computation is classically simulatable within multiplicative error.
Together with Toda's theorem~\cite{Tod91SIComp},
this implies the collapse of the polynomial-time hierarchy to the third level.

One obvious drawback of the existing argument above
is an inevitable postselection measurement inherent to the definition of $\postBQP$.
This becomes a quite essential obstacle when trying to extend this argument to the $\DQC{1}$ model,
where only one qubit is allowed to be measured.

To deal with the $\DQC{1}$ model,
this paper takes a different approach
by considering the complexity class~$\NQP$ introduced in Ref.~\cite{AdlDeMHua97SIComp}
or the class~$\SBQP$ introduced in Ref.~\cite{Kup15ToC}.
Let $\NQoneP$ and $\SBQoneP$ be the variants of $\NQP$ and $\SBQP$, respectively,
in which the quantum computation performed is restricted to the $\DQC{1}$ computation
(the precise definitions of $\NQoneP$ and $\SBQoneP$ will be found in Subsection~\ref{Subsection: DQCk computation and complexity classes}).
From one of the technical tools used for proving the main results of this paper
(the \textsc{One-Clean-Qubit Simulation Procedure} in Subsection~\ref{Subsection: One-Clean-Qubit Simulation Procedure}),
it is almost immediate to show the following theorem that states that
the restriction to the $\DQC{1}$ computation
does not change the complexity classes~$\NQP$ and $\SBQP$.

\begin{theorem}
  ${\NQP = \NQoneP}$ and ${\SBQP = \SBQoneP}$.
  \label{Theorem: NQP = NQ[1]P and SBQP = SBQ[1]P}
\end{theorem}

If any $\DQC{1}$ computation were classically simulatable within multiplicative error,
however,
the class~$\NQoneP$ would be included in $\NP$ and the class~$\SBQoneP$ would be included in $\SBP$,
where $\SBP$ is a classical version of $\SBQP$ in short, introduced in Ref.~\cite{BohGlaMei06JCSS}.
Similarly,
if any $\DQC{1}$ computation were classically simulatable within exponentially small additive error,
both $\NQoneP$ and $\SBQoneP$ would be included in $\SBP$.
Combined with Theorem~\ref{Theorem: NQP = NQ[1]P and SBQP = SBQ[1]P},
any of the inclusions~%
${\NQoneP \subseteq \NP}$,
${\SBQoneP \subseteq \SBP}$,
and ${\NQoneP \subseteq \SBP}$
further implies an implausible consequence that ${\PH = \AM}$,
which in particular implies the collapse of the polynomial-time hierarchy to the second level.
Accordingly, the following theorem holds.

\begin{theorem}
  The $\DQC{\mathit{1}}$ model is not classically simulatable
  either within multiplicative error or exponentially small additive error,
  unless ${\PH = \AM}$.
  \label{Theorem: hardness of classically simulating DQC1, informal}
\end{theorem}


\subsection{Organization of the Paper}

Section~\ref{Section: overview of error reduction results}
overviews the proofs of the error reduction results,
the main results of this paper.
Section~\ref{Section: preliminaries}
summarizes the notions and properties that are used throughout this paper.
Section~\ref{Section: building blocks}
provides detailed descriptions and rigorous analyses of three key technical tools used in this paper.
Section~\ref{Section: error reduction for one-sided-error cases}
then rigorously proves Theorems~\ref{Theorem: error reduction for perfect-completeness case}~and~\ref{Theorem: co-RQlogP is in co-RQ1P}
and Corollary~\ref{Corollary: error reduction for perfect-soundness cases},
the error reduction results in the cases of one-sided bounded error.
Section~\ref{Section: error reduction for two-sided-error cases}
treats the two-sided-bounded-error cases,
by first providing one more technical tool,
and then rigorously proving Theorems~\ref{Theorem: error reduction for two-sided-error case}~and~\ref{Theorem: BQlogP with constant gap is in BQ1P}.
The completeness results
(Theorem~\ref{Theorem: completeness of the trace estimation problem})
and the results related to the hardness of weak classical simulations of the $\DQC{1}$ model
(Theorems~\ref{Theorem: NQP = NQ[1]P and SBQP = SBQ[1]P}~and~\ref{Theorem: hardness of classically simulating DQC1, informal})
are proved in
Sections~\ref{Section: completeness results}~and~\ref{Section: hardness of weak simulation}, respectively.
Finally, Section~\ref{Section: conclusion}
concludes the paper with some open problems.


\section{Overview of Error Reduction Results}
\label{Section: overview of error reduction results}

This section presents an overview of the proofs for the error reduction results.
First, Subsection~\ref{Subsection: overview of error reduction for one-sided-error case}
provides high-level descriptions of the proofs of
Theorems~\ref{Theorem: error reduction for perfect-completeness case}~and~\ref{Theorem: co-RQlogP is in co-RQ1P},
the theorems for the one-sided error case of perfect completeness.
Compared with the two-sided-error case,
the proof construction is relatively simpler in the perfect-completeness case,
but already involves most of key technical ingredients of this paper.
Subsection~\ref{Subsection: overview of error reduction for two-sided-error case}
then explains the further idea that proves
Theorems~\ref{Theorem: error reduction for two-sided-error case}~and~\ref{Theorem: BQlogP with constant gap is in BQ1P},
the theorems for the two-sided-error case.


\subsection{Proof Ideas of Theorems~\ref{Theorem: error reduction for perfect-completeness case}~and~\ref{Theorem: co-RQlogP is in co-RQ1P}}
\label{Subsection: overview of error reduction for one-sided-error case}

Let ${A = (A_\yes, A_\no)}$ be any problem in ${\QlogP(1, s)}$,
where the function~$s$ defining the soundness is bounded away from one by an inverse-polynomial,
and consider a polynomial-time uniformly generated family of quantum circuits
that puts $A$ in ${\QlogP(1, s)}$.
Let $Q_x$ denote the quantum circuit from this family when the input is $x$,
where $Q_x$ acts over ${\op{w}(\abs{x})}$~qubits for some polynomially bounded function~$w$,
and is supposed to be applied to the initial state~%
${
  (\ketbra{0})^{\tensor \op{k}(\abs{x})}
  \tensor
  \bigl( \frac{I}{2} \bigr)^{\tensor (\op{w}(\abs{x}) - \op{k}(\abs{x}))}
}$
that contains exactly ${\op{k}(\abs{x})}$ clean qubits,
for some logarithmically bounded function~$k$.

Theorems~\ref{Theorem: error reduction for perfect-completeness case}~and~\ref{Theorem: co-RQlogP is in co-RQ1P}
are proved by constructing circuits with desirable properties from the original circuit~$Q_x$.
The construction is essentially the same for both of the two theorems
and consists of three stages of transformations of circuits:
The first stage reduces the number of necessary clean qubits to just one,
while keeping perfect completeness
and soundness still bounded away from one by an inverse-polynomial.
The second stage then makes the acceptance probability of no-instances arbitrarily close to $1/2$,
still using just one clean qubit and keeping perfect completeness.
Here, it not only makes the soundness
(i.e., the upper bound of the acceptance probability of no-instances)
close to $1/2$,
but also makes the acceptance probability of no-instances \emph{at least} $1/2$.
Finally, in the case of Theorem~\ref{Theorem: co-RQlogP is in co-RQ1P},
the third stage further reduces soundness error to be polynomially small
with the use of just one clean qubit,
while preserving the perfect completeness property.
If one more clean qubit is available,
the third stage can make soundness exponentially small
with keeping perfect completeness,
which leads to Theorem~\ref{Theorem: error reduction for perfect-completeness case}.
The analyses of the third stage effectively use the fact
that the acceptance probability of no-instances is close to $1/2$
after the transformation of the second stage.
The rest of this subsection sketches the ideas that realize each of these three stages.

\paragraph{\textsc{One-Clean-Qubit Simulation Procedure}}

The first stage uses a procedure called the \textsc{One-Clean-Qubit Simulation Procedure}.
Given the quantum circuit~$Q_x$ with a specification of the number~${\op{k}(\abs{x})}$ of clean qubits,
this procedure results in a quantum circuit~$R_x$ such that
the input state to $R_x$ is supposed to contain just one clean qubit,
and when applied to the one-clean-qubit initial state,
the acceptance probability of $R_x$ is still one if $x$ is in $A_\yes$,
while it is at most ${1 - \op{\delta}(\abs{x})}$ if $x$ is in $A_\no$,
where $\delta$ is an inverse-polynomial function determined by ${\delta = 2^{- k} (1 - s)}$.
It is stressed that the \textsc{One-Clean-Qubit Simulation Procedure} preserves perfect completeness,
which is in stark contrast to the straightforward method of one-clean-qubit simulation.

Consider the ${\op{k}(\abs{x})}$-clean-qubit computation performed with $Q_x$.
Let $\regQ$ denote the quantum register
consisting of the ${\op{k}(\abs{x})}$~initially clean qubits,
and let $\regR$ denote the quantum register
consisting of the remaining ${\op{w}(\abs{x}) - \op{k}(\abs{x})}$~qubits
that are initially in the totally mixed state.
Further let $\regQ^{(1)}$ denote the single-qubit quantum register consisting of the first qubit of $\regQ$,
which corresponds to the output qubit of $Q_x$.
In the one-clean-qubit simulation of $Q_x$ by $R_x$,
the ${\op{k}(\abs{x})}$~qubits in $\regQ$ are supposed to be in the totally mixed state initially
and $R_x$ tries to simulate $Q_x$ only when $\regQ$ initially contains the clean all-zero state.
To do so, $R_x$ uses another quantum register~$\regO$ consisting of just a single qubit,
and this qubit in $\regO$ is the only qubit that is supposed to be initially clean.

For ease of explanations,
assume for a while that all the qubits in $\regQ$ are also initially clean even in the case of $R_x$.
The key idea in the construction of $R_x$
is the following simulation of $Q_x$ that makes use of the phase-flip transformation:
The simulation first applies the Hadamard transformation~$H$ to the qubit in $\regO$
and then flips the phase if and only if the content of $\regO$ is $1$
\emph{and} the simulation of $Q_x$ results in rejection
(which is realized by performing $Q_x$ to ${(\regQ, \regR)}$
and then applying the controlled-$Z$ transformation to ${(\regO, \regQ^{(1)})}$,
where the content~$1$ in $\regQ^{(1)}$
is assumed to correspond to the rejection in the original computation by $Q_x$).
The simulation further performs the inverse of $Q_x$ to ${(\regQ, \regR)}$
and again applies the Hadamard transformation~$H$ to $\regO$.
At the end of the simulation,
the qubit in $\regO$ is measured in the computational basis,
with the measurement result $0$ corresponding to acceptance.
The point is that this phase-flip-based construction provides a quite ``faithful'' simulation of $Q_x$,
meaning that the rejection probability of the simulation
is polynomially related to the rejection probability of the original computation of $Q_x$
(and in particular, the simulation never rejects when the original computation never rejects,
i.e., the simulation preserves the perfect completeness property).

As mentioned before,
all the qubits in $\regQ$ are supposed to be in the totally mixed state initially
in the one-clean-qubit simulation of $Q_x$ by $R_x$,
and $R_x$ tries to simulate $Q_x$ only when $\regQ$ initially contains the clean all-zero state.
To achieve this,
each of the applications of the Hadamard transformation
is replaced by an application of the controlled-Hadamard transformation
so that the Hadamard transformation is applied only when all the qubits in $\regQ$ are in state~$\ket{0}$.
By considering the one-clean-qubit computations with the circuit family induced by $R_x$,
the perfect completeness property is preserved
and soundness is still bounded away from one by an inverse-polynomial
(although the rejection probability becomes smaller for no-instances
by a multiplicative factor of $2^{- k}$,
where notice that $2^{- k}$ is an inverse-polynomial as $k$ is a logarithmically bounded function).
The construction of $R_x$ is summarized in Figure~\ref{Figure: simplified description of One-Clean-Qubit Simulation Procedure}.
A precise description and a detailed analysis of the \textsc{One-Clean-Qubit Simulation Procedure}
will be presented in Subsection~\ref{Subsection: One-Clean-Qubit Simulation Procedure}.

\begin{figure}[t!]
  \begin{algorithm*}{\textsc{One-Clean-Qubit Simulation Procedure} --- Simplified Description}
    \begin{step}
    \item
      Prepare a single-qubit register~$\regO$,
      a ${\op{k}(\abs{x})}$-qubit register~$\regQ$,
      and a ${(\op{w}(\abs{x}) - \op{k}(\abs{x}))}$-qubit register~$\regR$,
      where the qubit in $\regO$ is supposed to be initially in state~$\ket{0}$,
      while all the qubits in $\regQ$ and $\regR$ are supposed to be initially in the totally mixed state~${I/2}$.\\
      Apply $H$ to $\regO$
      if all the qubits in $\regQ$ are in state $\ket{0}$.
    \item
      Apply $Q_x$ to ${(\regQ, \regR)}$.
    \item
      Apply the phase-flip (i.e., multiply the phase by $-1$)
      if the content of ${(\regO, \regQ^{(1)})}$ is $11$,
      where $\regQ^{(1)}$ denotes the single-qubit register consisting of the first qubit of $\regQ$.
    \item
      Apply $\conjugate{Q_x}$ to ${(\regQ, \regR)}$.
    \item
      Apply $H$ to $\regO$
      if all the qubits in $\regQ$ are in state $\ket{0}$.
      Measure the qubit in $\regO$ in the computational basis.
      Accept if this results in $\ket{0}$, and reject otherwise.
    \end{step}
  \end{algorithm*}
  \caption{
    The \textsc{One-Clean-Qubit Simulation Procedure}
    induced by a quantum circuit~$Q_x$
    with the specification of the number~${\op{k}(\abs{x})}$ of clean qubits
    used in the computation of $Q_x$ to be simulated
    (a slightly simplified description).
  }
  \label{Figure: simplified description of One-Clean-Qubit Simulation Procedure}
\end{figure}

\paragraph{\textsc{Randomness Amplification Procedure}}

The second stage uses the procedure called the \textsc{Randomness Amplification Procedure}.
Given the circuit~$R_x$ constructed in the first stage,
this procedure results in a quantum circuit~$R'_x$ such that
the input state to $R'_x$ is still supposed to contain just one clean qubit,
and when applied to the one-clean-qubit initial state,
the acceptance probability of $R'_x$ is still one if $x$ is in $A_\yes$,
while it is in the interval~${\bigl[\frac{1}{2}, \frac{1}{2} + \op{\varepsilon}(\abs{x}) \bigr]}$
if $x$ is in $A_\no$
for some sufficiently small function~$\varepsilon$.

Consider the one-clean-qubit computation performed with $R_x$.
Let $\regO$ denote the single-qubit quantum register consisting of the initially clean qubit,
which is also the output qubit of $R_x$.
Let $\regR$ denote the quantum register
consisting of all the remaining qubits that are initially in the totally mixed state
(by the construction of $R_x$, $\regR$ consists of ${\op{w}(\abs{x})}$~qubits).

Suppose that the qubit in $\regO$ is measured in the computational basis
after $R_x$ is applied to the one-clean-qubit initial state~%
${\ketbra{0} \tensor \bigl( \frac{I}{2} \bigr)^{\tensor \op{w}(\abs{x})}}$
in ${(\regO, \regR)}$.
Obviously from the property of $R_x$,
the measurement results in $0$ with probability exactly equal to
the acceptance probability~$p_\acc$ of the one-clean-qubit computation with $R_x$.
Now suppose that $R_x$ is applied to a slightly different initial state~%
${\ketbra{1} \tensor \bigl( \frac{I}{2} \bigr)^{\tensor \op{w}(\abs{x})}}$
in ${(\regO, \regR)}$,
where $\regO$ initially contains $\ket{1}$ instead of $\ket{0}$
and all the qubits in $\regR$ are again initially in the totally mixed state.
The key property here to be proved is that,
in this case,
the measurement over the qubit in $\regO$ in the computational basis
results in $1$ again with probability exactly~$p_\acc$,
the acceptance probability of the one-clean-qubit computation with $R_x$.
This implies that,
after the application of $R_x$ to ${(\regO, \regR)}$
with all the qubits in $\regR$ being in the totally mixed state,
the content of $\regO$ remains the same with probability exactly~$p_\acc$,
and is flipped with probability exactly~${1 - p_\acc}$,
the rejection probability of the original one-clean-qubit computation with $R_x$,
regardless of the initial content of $\regO$.

The above observation leads to the following construction of the circuit~$R'_x$.
The construction of $R'_x$ is basically a sequential repetition of the original circuit~$R_x$.
The number~$N$ of repetitions is polynomially many with respect to the input length~$\abs{x}$,
and the point is that the register~$\regO$ is reused for each repetition,
and only the qubits in $\regR$ are refreshed after each repetition
(by preparing $N$~registers~${\regR_1, \dotsc, \regR_N}$,
each of which consists of ${\op{w}(\abs{x})}$~qubits, the same number of qubits as $\regR$,
all of which are initially in the totally mixed state).
After each repetition the qubit in $\regO$ is measured in the computational basis
(in the actual construction, this step is exactly simulated without any measurement
--- a single-qubit totally mixed state is prepared as a fresh ancilla qubit for each repetition
so that the content of $\regO$ is copied to this ancilla qubit using the $\CNOT$ transformation,
and this ancilla qubit is never touched after this $\CNOT$ application).
Now, no matter which measurement result is obtained at the $j$th repetition
for every~$j$ in ${\{1, \dotsc, N\}}$,
the register~$\regO$ is reused as it is,
and the circuit~$R_x$ is simply applied to ${(\regO, \regR_{j + 1})}$ at the ${(j + 1)}$st repetition.
After the $N$~repetitions,
the qubit in $\regO$ is measured in the computational basis,
which is the output of $R'_x$
(the output~$0$ corresponds to acceptance).

The point is that
at each repetition, the content of $\regO$ is flipped with probability
exactly equal to the rejection probability of the original one-clean-qubit computation of $R_x$.
Taking into account that $\regO$ is initially in state~$\ket{0}$,
the computation of $R'_x$ results in acceptance
if and only if the content of $\regO$ is flipped even number of times during the $N$~repetitions.
An analysis on Bernoulli trials then shows that,
when the acceptance probability of the original one-clean-qubit computation of $R_x$
was in the interval~${\bigl[ \frac{1}{2}, 1 \bigr)}$,
the acceptance probability of the one-clean-qubit computation of $R'_x$ is at least $1/2$
and converges linearly to $1/2$ with respect to the repetition number.
On the other hand, when the acceptance probability of the original $R_x$ was one,
the content of $\regO$ is never flipped during the computation of $R'_x$,
and thus the acceptance probability of $R'_x$ remains one.
Figure~\ref{Figure: simplified description of Randomness Amplification Procedure}
summarizes the construction of $R'_x$.
A precise description and a detailed analysis of the \textsc{Randomness Amplification Procedure}
will be presented in Subsection~\ref{Subsection: Randomness Amplification Procedure}.

\begin{figure}[t!]
  \begin{algorithm*}{\textsc{Randomness Amplification Procedure} --- Simplified Description}
    \begin{step}
    \item
      Prepare a single-qubit register~$\regO$,
      where the qubit in $\regO$ is supposed to be initially in state~$\ket{0}$.
      Prepare a ${\op{w}(\abs{x})}$-qubit register~$\regR_j$
      for each $j$ in ${\{1, \dotsc, N\}}$,
      where all the qubits in $\regR_j$ are supposed to be initially in the totally mixed state~${I/2}$.
    \item
      For ${j = 1}$ to $N$, perform the following:
      \begin{step}
      \item
        Apply $R_x$ to ${(\regO, \regR_j)}$.
      \item
        Measure the qubit in $\regO$ in the computational basis.
      \end{step}
    \item
      Accept if the qubit in $\regO$ is in state~$\ket{0}$, and reject otherwise.
    \end{step}
  \end{algorithm*}
  \caption{The \textsc{Randomness Amplification Procedure} (a slightly simplified description).}
  \label{Figure: simplified description of Randomness Amplification Procedure}
\end{figure}

\paragraph{\textsc{Stability Checking Procedures}}

In the case of Theorem~\ref{Theorem: co-RQlogP is in co-RQ1P},
the third stage uses the procedure called the \textsc{One-Clean-Qubit Stability Checking Procedure}.
Given the circuit~$R'_x$ constructed in the second stage,
this procedure results in a quantum circuit~$R''_x$ such that
the input state to $R''_x$ is still supposed to contain just one clean qubit,
and when applied to the one-clean-qubit initial state,
the acceptance probability of $R''_x$ is still one if $x$ is in $A_\yes$,
while it is ${1/\op{p}(\abs{x})}$ if $x$ is in $A_\no$
for a polynomially bounded function~$p$ predetermined arbitrarily.

Consider the one-clean-qubit computation performed with $R'_x$.
Let $\regQ$ denote the single-qubit quantum register consisting of the initially clean qubit,
which is also the output qubit of $R'_x$.
Let $\regR$ denote the quantum register
consisting of all the remaining qubits that are initially in the totally mixed state,
and let ${\op{w}'(\abs{x})}$ denote the number of qubits in $\regR$.

Again the key observation is that,
after the application of $R'_x$ to ${(\regQ, \regR)}$
with all the qubits in $\regR$ being in the totally mixed state
(followed by the measurement over the qubit in $\regQ$ in the computational basis),
the content of $\regQ$ is flipped with probability
exactly equal to the rejection probability of the original one-clean-qubit computation with $R'_x$,
regardless of the initial content of $\regQ$.

This leads to the following construction of the circuit~$R''_x$.
The construction of $R''_x$ is again basically a sequential repetition of the original circuit~$R'_x$,
but this time the qubit in $\regQ$ is also supposed to be initially in the totally mixed state.
The circuit~$R'_x$ is repeatedly applied ${2N}$~times,
where $N$ is a power of two and is polynomially many with respect to the input length~$\abs{x}$,
and again the register~$\regQ$ is reused for each repetition,
and only the qubits in $\regR$ are refreshed after each repetition
(by preparing ${2N}$~registers~${\regR_1, \dotsc, \regR_{2N}}$,
each of which consists of ${\op{w}'(\abs{x})}$~qubits,
all of which are initially in the totally mixed state).
The key idea for the construction of $R''_x$ is
to use a counter that counts the number of attempts
such that the measurement over the qubit in $\regQ$ results in $\ket{1}$ after the application of $R'_x$
(again each measurement is simulated by a $\CNOT$ application
using an ancilla qubit of a totally mixed state).
Notice that the content of $\regQ$ is never flipped regardless of the initial content of $\regQ$,
if the original acceptance probability is one in the one-clean-qubit computation with $R'_x$.
Hence, in this case the counter value either stationarily remains its initial value
or is increased exactly by ${2N}$, the number of repetitions.
On the other hand,
if the original acceptance probability is close to $1/2$ in the one-clean-qubit computation with $R'_x$,
the content of $\regQ$ is flipped with probability close to $1/2$ after each application of $R'_x$
regardless of the initial content of $\regQ$.
This means that, after each application of $R'_x$,
the measurement over the qubit in $\regQ$ results in $\ket{1}$
with probability close to $1/2$ regardless of the initial content of $\regQ$,
and thus, the increment of the counter value must be distributed around ${\frac{1}{2} \cdot 2N = N}$
with very high probability.
Now, if the counter value is taken modulo ${2N}$
and if the unique initially clean qubit is prepared for the most significant bit of the counter
(which picks the initial counter value from the set~${\{0, \dotsc, N - 1\}}$ uniformly at random),
the computational-basis measurement over this most significant qubit of the counter
always results in $\ket{0}$ if $x$ is in $A_\yes$,
while it results in $\ket{1}$ with very high probability if $x$ is in $A_\no$
(which can be made at least ${1 - \frac{1}{\op{p}(\abs{x})}}$
for an arbitrarily chosen polynomially bounded function~$p$,
by taking an appropriately large number~${2N}$ of the repetition).
Figure~\ref{Figure: simplified description of One-Clean-Qubit Stability Checking Procedure}
summarizes the construction of $R''_x$.

\begin{figure}[t!]
  \begin{algorithm*}{\textsc{One-Clean-Qubit Stability Checking Procedure} --- Simplified Description}
    \begin{step}
    \item
      Given a positive integer~$N$ that is a power of two,
      prepare a counter~$C$ whose value is taken modulo ${2N}$.
      Choose an integer~$r$ from ${\{0, \dotsc, N - 1\}}$ uniformly at random,
      and initialize a counter~$C$ to $r$
      (which sets the most significant bit of $C$ to $0$).
      Prepare a single-qubit register~$\regQ$
      and a ${\op{w}'(\abs{x})}$-qubit register~$\regR_j$ for each $j$ in ${\{1, \dotsc, 2N\}}$,
      where all the qubits in $\regQ$ and $\regR_j$
      are supposed to be initially in the totally mixed state~${I/2}$.
    \item
      For ${j = 1}$ to ${2N}$, perform the following:
      \begin{step}
      \item
        Apply $R'_x$ to ${(\regQ, \regR_j)}$.
      \item
        Measure the qubit in $\regQ$ in the computational basis.
        If this results in $\ket{1}$, increase the counter~$C$ by one.
      \end{step}
    \item
      Reject if ${N \leq C \leq 2N - 1}$, and accept otherwise
      (i.e., accept iff the most significant bit of $C$ is $0$).
    \end{step}
  \end{algorithm*}
  \caption{The \textsc{One-Clean-Qubit Stability Checking Procedure} (a simplified description).}
  \label{Figure: simplified description of One-Clean-Qubit Stability Checking Procedure}
\end{figure}

One drawback of the construction of $R''_x$ above
via the \textsc{One-Clean-Qubit Stability Checking Procedure}
is that, in the case of no-instances,
there inevitably exist some ``bad'' initial counter values in ${\{0, \dotsc, N - 1\}}$
with which $R''_x$ is forced to accept with unallowably high probability.
For instance, if the initial counter value is~$0$,
$R''_x$ is forced to accept when the increment of the counter is less than $N$,
which happens with probability at least a constant.
This is the essential reason why the current approach achieves
only a polynomially small soundness
in the one-clean-qubit case in Theorem~\ref{Theorem: co-RQlogP is in co-RQ1P},
as the number of possible initial counter values can be at most polynomially many
(otherwise the number of repetitions must be super-polynomially many) 
and even just one ``bad'' initial value is problematic
to go beyond polynomially small soundness.
In contrast, if not just one but two clean qubits are available,
one can remove the possibility of ``bad'' initial counter values,
which results in the \textsc{Two-Clean-Qubit Stability Checking Procedure}.
This time, the circuit~$R'_x$ is repeatedly applied ${8N}$~times,
and the counter value is taken modulo ${8N}$.
The two initially clean qubits are prepared for the most and second-most significant bits of the counter,
which results in picking the initial counter value
from the set~${\{0, \dotsc, 2N - 1\}}$ uniformly at random.
Now the point is that the counter value can be increased by $N$ before the repetition
so that the actual initial value of the counter is in ${\{N, \dotsc, 3N - 1\}}$,
which discards the tail sets~${\{0, \dotsc, N - 1\}}$ and ${\{3N, \dotsc, 4N - 1\}}$
of the set~${\{0, \dotsc, 4N - 1\}}$.
As the size of the tail sets discarded is sufficiently large,
there no longer exists any ``bad'' initial counter value,
which leads to the exponentially small soundness
in the two-clean-qubit case in Theorem~\ref{Theorem: error reduction for perfect-completeness case}.

Precise descriptions and detailed analyses of the \textsc{One-Clean-Qubit Stability Checking Procedure}
and \textsc{Two-Clean-Qubit Stability Checking Procedure}
will be presented in Subsection~\ref{Subsection: Stability Checking Procedures}.


\subsection{Proof Ideas of Theorems~\ref{Theorem: error reduction for two-sided-error case}~and~\ref{Theorem: BQlogP with constant gap is in BQ1P}}
\label{Subsection: overview of error reduction for two-sided-error case}

The results for the two-sided error case need more complicated arguments
and is proved in eight stages of transformations in total,
which are split into three parts.

The first part consists of three stages,
and proves that any problem solvable with constant completeness and soundness
using logarithmically many clean qubits
is also solvable with constant completeness and soundness
using just one clean qubit.
At the first stage of the first part,
by a standard repetition with a threshold-value decision,
one first reduces errors to be sufficiently small constants,
say, completeness~$15/16$ and soundness~$1/16$.
For this, if the starting computation has a constant gap between completeness and soundness,
one requires only a constant number of repetitions,
and thus, the resulting computation still requires only logarithmically many clean qubits.
The second stage of the first part then reduces the number of clean qubits to just one.
The procedure in this stage is exactly the \textsc{One-Clean-Qubit Simulation Procedure}
developed in the first stage of the one-sided error case.
The gap between completeness and soundness becomes only an inverse-polynomial by this transformation,
but the point is that the gap is still sufficiently larger
(i.e. a constant times larger)
than the completeness error.
Now the third stage of the first part transforms the computation resulting from the second stage
to the computation that still uses only one clean qubit and has constant completeness and soundness.
The procedure in this stage is exactly the \textsc{Randomness Amplification Procedure},
developed in the second stage of the one-sided error case,
and it makes use of the difference of the rates of convergence to $1/2$
of the acceptance probability between the yes- and no-instance cases.
The precise statement corresponding to the first part is found
as Lemma~\ref{Lemma: BQlogP with constant gap is in BQ1P with constant gap}
in Subsection~\ref{Subsection: Proofs of Theorems for Error-Reduction for Two-Sided-Error Cases}.

The second part consists of two stages,
and proves that
any problem solvable with constant completeness and soundness
using just one clean qubit
is also solvable with almost-perfect (i.e., exponentially close to one) completeness
and soundness below $1/2$
using just logarithmically many clean qubits.
At the first stage of the second part,
one reduces both of the completeness and soundness errors to be polynomially small,
again by a standard repetition with a threshold-value decision.
Note that the computation resulting from the first part requires only one clean qubit.
Thus, even when repeated logarithmically many times,
the resulting computation uses just logarithmically many clean qubits,
and achieves polynomially small errors.
The second stage of the second part then
repeatedly attempts the computation resulting from the first stage polynomially many times,
and accepts if at least one of the attempts results in acceptance
(i.e., takes OR of the attempts).
A straightforward repetition requires polynomially many clean qubits,
and to avoid this problem,
after each repetition one tries to recover the clean qubits for reuse
by applying the inverse of the computation
(the failure of this recovery step is counted as an ``acceptance'' when taking the OR).
This results in a computation that still requires only logarithmically many clean qubits,
and has completeness exponentially close to one,
while soundness is still below $1/2$.
The precise statement corresponding to the second part is found
as Lemma~\ref{Lemma: BQ1P with constant gap is in BQlogP with almost perfect completeness}
in Subsection~\ref{Subsection: Proofs of Theorems for Error-Reduction for Two-Sided-Error Cases}.

Now the third part is essentially the same as the three-stage transformation of the one-sided error case.
From the computation resulting from the second part,
the first stage of the third part decreases the number of clean qubits to just one,
via the \textsc{One-Clean-Qubit Simulation Procedure}.
The completeness of the resulting computation is still exponentially close to one
and its soundness is bounded away from one by an inverse-polynomial.
The second stage of the third part
then applies the \textsc{Randomness Amplification Procedure}
to make the acceptance probability of no-instances arbitrarily close to $1/2$,
while keeping completeness exponentially close to one.
Finally, the third stage of the third part proves that
one can further decrease soundness error to be polynomially small using just one qubit
via the \textsc{One-Clean-Qubit Stability Checking Procedure},
or to be exponentially small using just two qubits
via the \textsc{Two-Clean-Qubit Stability Checking Procedure},
while keeping completeness exponentially close to one.

By considering the complement problem,
the above argument can also prove
the case of exponentially small soundness error in Theorem~\ref{Theorem: BQlogP with constant gap is in BQ1P}.


\section{Preliminaries}
\label{Section: preliminaries}

Throughout this paper,
let $\Natural$ and $\Nonnegative$ denote
the sets of positive and nonnegative integers, respectively,
and let ${\Sigma = \Binary}$ denote the binary alphabet set. 
A function~$\function{f}{\Nonnegative}{\Natural}$ is \emph{polynomially bounded}
if there exists a polynomial-time deterministic Turing machine
that outputs ${1^{f(n)}}$ on input~$1^n$.
A function~$\function{f}{\Nonnegative}{\Natural}$ is \emph{logarithmically bounded}
if $f$ is polynomial-time computable and ${f(n)}$ is in ${O(\log n)}$.
A function~$\function{f}{\Nonnegative}{[0,1]}$ is \emph{negligible}
if, for every polynomially bounded function~$\function{g}{\Nonnegative}{\Natural}$,
it holds that ${f(n) < 1/g(n)}$ for all but finitely many values of~$n$.


\subsection{Useful Properties on Random Variables}
\label{Subsection: useful properties of random variables}

This subsection presents two lemmas
on properties of the random variable that follows the binomial distribution,
which are used in this paper.

The first lemma is a special case of the Hoeffding inequality.

\begin{lemma}
  For any $n$ in $\Natural$ and $p$ in ${[0,1]}$,
  let $X$ be a random variable over ${\{0, \dotsc, n\}}$
  that follows the binomial distribution~${B(n,p)}$.
  Then, for any $\delta$ in ${(0, 1)}$,
  \[
    \biggprob{ \frac{X}{n} \geq p + \delta }
    <
    e^{- 2 \delta^2 n}
    \quad
    \text{and}
    \quad
    \biggprob{ \frac{X}{n} \leq p - \delta }
    <
    e^{- 2 \delta^2 n}.
  \]
  \label{Lemma: Hoeffding bounds}
\end{lemma}

The second lemma is on the probability that a random variable takes an even number
when it follows the binomial distribution.

\begin{lemma}
  For any $n$ in $\Natural$ and $p$ in ${[0,1]}$,
  let $X$ be a random variable over ${\{0, \dotsc, n\}}$
  that follows the binomial distribution~${B(n,p)}$.
  Then,
  \[
    \prob{ \textnormal{$X$ is even} }
    =
    \frac{1}{2} + \frac{1}{2} (1 - 2p)^n.
  \]
  \label{Lemma: coin-flipping lemma}
\end{lemma}

\begin{proof}
  For each $j$ in ${\{1, \dotsc, n\}}$,
  let $Y_j$ be an independent random variable over ${\{-1, 1\}}$
  that takes $-1$ with probability~$p$,
  and let $Z$ be a random variable over $\Binary$ defined by
  \[
    Z = \frac{1}{2} + \frac{1}{2} \prod_{j=1}^n Y_j.
  \]
  Notice that $Z$ is $1$ if and only if there are even number of indices~$j$
  such that $Y_j$ is $-1$.
  Hence,
  \[
    \prob{ \textnormal{$X$ is even} }
    =
    \expect{Z}
    =
    \frac{1}{2} + \frac{1}{2} \Biggexpect{\prod_{j=1}^n Y_j}
    =
    \frac{1}{2} + \frac{1}{2} \prod_{j=1}^n \expect{Y_j}
    =
    \frac{1}{2} + \frac{1}{2} (1 - 2p)^n,
  \]
  where the third equality uses the fact that each $Y_j$ is an independent random variable,
  and the claim follows.
\end{proof}


\subsection{Quantum Fundamentals}
\label{Subsection: quantum fundamentals}

We assume the reader is familiar with the quantum formalism,
including pure and mixed quantum states, density operators,
and measurements, as well as the quantum circuit model
(see Refs.~\cite{NieChu00Book, KitSheVya02Book, Wil13Book}, for instance).
This subsection summarizes some notations and properties that are used in this paper.

For every positive integer~$n$,
let ${\Complex(\Sigma^n)}$ denote the $2^n$-dimensional complex Hilbert space
whose standard basis vectors are indexed by the elements in $\Sigma^n$.
In this paper, all Hilbert spaces are complex and have dimension a power of two.
For a Hilbert space~$\spaceH$,
let $I_\spaceH$ denote the identity operator over $\spaceH$,
and let ${\Unitary(\spaceH)}$ denote the set of unitary operators over $\spaceH$.
As usual, let
\[
  X
  =
  \begin{pmatrix}
    0 & 1\\
    1 & 0
  \end{pmatrix},
  \quad
  Z
  =
  \begin{pmatrix}
    1 & 0\\
    0 & -1
  \end{pmatrix}
\]
denote the Pauli operators,
and let
\[
  H
  =
  \frac{1}{\sqrt{2}}
  \begin{pmatrix}
    1 & 1\\
    1 & -1
  \end{pmatrix},
  \quad
  S
  =
  \begin{pmatrix}
    1 & 0\\
    0 & i
  \end{pmatrix},
  \quad
  T
  =
  \begin{pmatrix}
    1 & 0\\
    0 & e^{i \frac{\pi}{4}}
  \end{pmatrix}
\]
denote the Hadamard, $i$-phase-shift, and $T$~operators, respectively
(the $T$~operator corresponds to the ${\pi/8}$-gate).
Notice that ${S = T^2}$ and ${Z = S^2 = T^4}$.
For convenience, we may identify a unitary operator with the unitary transformation it induces.
In particular, for a unitary operator~$U$,
the induced unitary transformation is also denoted by $U$.

For a Hilbert space~$\spaceH$ and a unitary operator~$U$ in ${\Unitary(\spaceH)}$,
let ${\controlled(U)}$ denote
the \emph{controlled-$U$} operator in ${\Unitary \bigl( \Complex(\Sigma) \tensor \spaceH \bigr)}$
defined by
\[
  \controlled(U) = \ketbra{0} \tensor I_\spaceH + \ketbra{1} \tensor U.
\]
For any positive integer~${n \geq 2}$,
the \emph{$n$-controlled-$U$} operator in ${\Unitary \bigl( \Complex(\Sigma^n) \tensor \spaceH \bigr)}$
is recursively defined by
\[
  \controlled^n(U)
  =
  \controlled \bigl( \controlled^{n - 1}(U) \bigr)
  =
  \ketbra{0} \tensor I_{\Complex(\Sigma^{n - 1}) \tensor \spaceH} + \ketbra{1} \tensor \controlled^{n - 1}(U),
\]
where ${\controlled^1(U)}$ may be interpreted as the controlled-$U$ operator~${\controlled(U)}$.
In the case where $U$ is a unitary transformation over a single qubit,
notice that the last qubit is the target qubit
for each of the unitary transformations~${\controlled(U)}$ and ${\controlled^n(U)}$
in the notation above.
For notational convenience,
for each positive integer~$n$
and for each integer~$j$ in ${\{1, \dotsc, n + 1\}}$,
let ${\controlled^n_j(U)}$ denote the case of the $n$-controlled-$U$ operator
in which the corresponding transformation uses the $j$th qubit as the target qubit.
The operator~${\controlled^n(U)}$ corresponds to ${\controlled^n_{n + 1}(U)}$ in this notation.
The operator~${\controlled^1_1(U)}$ may be simply denoted by ${\controlled_1(U)}$.


\subsection{Quantum Circuits}
\label{Subsection: quantum circuits}

A quantum circuit is specified by a series of quantum gates
with designation of qubits to which each quantum gate is applied.
It is assumed that any quantum circuit is composed of gates
in some reasonable, universal, finite set of quantum gates.
A \emph{description} of a quantum circuit is a string in $\Sigma^\ast$
that encodes the specification of the quantum circuit.
The encoding must be a ``reasonable'' one,
i.e., the number of gates in a circuit encoded
is not more than the length of the description of that circuit,
and each gate of the circuit is specifiable
by a deterministic procedure in time polynomial
with respect to the length of the description.

A family~${\{ Q_x \}_{x \in \Sigma^\ast}}$ of quantum circuits is
\emph{polynomial-time uniformly generated}
if there exists a polynomial-time deterministic Turing machine
that, on every input~$x$ in $\Sigma^\ast$,
outputs a description of $Q_x$.
For convenience,
we may identify a circuit~$Q_x$ with the unitary operator it induces.

\begin{remark}
  Notice that the input~$x$ is ``hard-coded'' in the generated circuit~$Q_x$
  in the above definition of polynomial-time uniformly generated family of quantum circuits,
  i.e., each circuit~$Q_x$ depends on the input~$x$ itself.
  The choice of this ``hard-coded'' definition is just for ease of explanations,
  and all the results in this paper do remain valid even with a more standard definition
  of the polynomial-time uniformity
  where each circuit generated depends only on the input length~$\abs{x}$,
  and the input~$x$ is given to the circuit as a read-only input.
  In fact, all the results still remain valid
  even when using the \emph{logarithmic-space} uniformly generated family of quantum circuits
  to define the complexity classes in Subsection~\ref{Subsection: DQCk computation and complexity classes}
  (by suitably replacing polynomial-time computable functions by logarithmic-space computable functions
  in some statements
  and by changing the definitions of polynomially and logarithmically bounded functions
  by using logarithmic-space deterministic Turing machines and logarithmic-space computable functions).
  See Subsection~\ref{Subsection: remarks on uniformity}
  for further discussions on the uniformity of quantum circuits.
\end{remark}

This paper assumes a gate set
that includes the Hadamard, $T$, and $\CNOT$ gates.
Note that this assumption is satisfied by many standard gate sets,
and is very reasonable and not restrictive.
In particular, the gate set proposed in Ref.~\cite{BoyMorPulRoyVat00IPL}
exactly consists of these three gates.
Some useful transformations are in order
that are exactly implementable with such a gate set using or not using ancilla qubits:

\paragraph{Transformations corresponding to Clifford-group operators}

First note that, with a gate set satisfying the assumption above,
any transformation corresponding to a Clifford-group operator is exactly implementable
without using any ancilla qubits.

In particular,
the phase-flip transformation~$Z$ is nothing but $T^4$, and is easily realized.
Thus, the $\NOT$ transformation (the operator~$X$) is also easily realizable,
for ${X = HZH}$.

As ${Z = HXH}$,
the controlled-$Z$ transformation~${\controlled(Z)}$ is also realizable
by using the decomposition
\[
  \controlled(Z) = (I \tensor H) \controlled(X) (I \tensor H),
\]
where ${\controlled(X)}$ is nothing but the $\CNOT$ transformation.

\paragraph{Generalized Toffoli transformations}

Using some ancilla qubits,
any generalized Toffoli transformation
(i.e., the $n$-controlled-NOT transformation~$\controlled^n(X)$ for any positive integer~$n$),
is also exactly implementable with the gate set discussed,
according to the constructions in Ref.~\cite{Bar+95PRA}.
In the construction in Lemma~7.2 of Ref.~\cite{Bar+95PRA},
the number of necessary ancilla qubits
grows linearly with respect to the number of control qubits,
and the construction in Corollary~7.4 of Ref.~\cite{Bar+95PRA}
uses only two ancilla qubits when ${n \geq 5}$.
One very helpful property is
that no initializations are required for all these ancilla qubits
(and thus, all of them can actually be re-used when applying other generalized Toffoli transformations).
In particular, even totally mixed states may be used for these ancilla qubits,
and hence,
with the above-mentioned gate set,
generalized Toffoli transformations may be assumed available freely
when constructing quantum circuits
that are used for $k$-clean-qubit computations
defined formally in Subsection~\ref{Subsection: DQCk computation and complexity classes}.
See Lemma~7.2 and Corollary~7.4 of Ref.~\cite{Bar+95PRA} for details.

\paragraph{Controlled-Hadamard transformations}

Recall that the Hadamard operator~$H$ is decomposed as
\[
  H = \conjugate{S} H \conjugate{T} X T H S.
\]
This implies that
the $n$-controlled-Hadamard transformation~${\controlled^n(H)}$ for any positive integer~$n$
is decomposed as
\[
  \controlled^n(H)
  =
  \bigl( I^{\tensor n} \tensor \conjugate{S} \bigr)
  \bigl( I^{\tensor n} \tensor H \bigr)
  \bigl( I^{\tensor n} \tensor \conjugate{T} \bigr)
  \controlled^n(X)
  \bigl( I^{\tensor n} \tensor T \bigr)
  \bigl( I^{\tensor n} \tensor H \bigr)
  \bigl( I^{\tensor n} \tensor S \bigr).
\]
As ${S = T^2}$, ${\conjugate{T} = T^7}$, and ${\conjugate{S} = T^6}$,
the ${\controlled^n(H)}$~transformation for each $n$
is exactly implementable by using
two Hadamard gates,
sixteen $T$ gates,
and one generalized Toffoli transformation
(i.e., one $n$-controlled-NOT transformation~${\controlled^n(X)}$).
Clearly, the only necessary ancilla qubits
are those used for realizing the ${\controlled^n(X)}$~transformation
in this implementation.
Hence, provided that generalized Toffoli transformations may be assumed available freely,
for all $n$, $n$-controlled-Hadamard transformations may also be assumed available freely.

\paragraph{Increment transformations}

For any positive integer~$n$,
let $\INCR{2^n}$ denote the increment transformation over~$\Integers_{2^n}$,
which is the unitary transformation acting over $n$~qubits defined by
\[
  \INCR{2^n} \colon \ket{j} \mapsto \bigket{(j+1) \bmod 2^n},
  \quad
  \forall j \in \Integers_{2^n}
\]
Note that 
\[
  \INCR{2} = X,
\]
and for each positive integer~${n \geq 2}$,
\[
  \INCR{2^n} = \bigl( I \tensor \INCR{2^{n - 1}} \bigr) \controlled^{n - 1}_1(X),
\]
where recall that ${\controlled^{n - 1}_1(X)}$
corresponds to the ${(n - 1)}$-controlled-NOT transformation
with the first qubit being the target.
Hence, for each positive integer~${n \geq 2}$,
\[
  \INCR{2^n}
  =
  \bigl( I^{\tensor (n - 1)} \tensor X \bigr)
  \bigl( I^{\tensor (n - 2)} \tensor \controlled_1(X) \bigr)
  \dotsm
  \bigl( I \tensor \controlled^{n - 2}_1(X) \bigr)
  \controlled^{n - 1}_1(X),
\]
and thus,
each $\INCR{2^n}$~transformation
is exactly implementable by combining $\NOT$, $\CNOT$, and generalized Toffoli transformations only.
Note that
ancilla qubits are required only for realizing the generalized Toffoli transformations
in this implementation.
Accordingly, provided that generalized Toffoli transformations may be assumed available freely,
the increment transformation over $\Integers_{2^n}$ may also be assumed available freely,
for each positive integer~$n$,
and so may the controlled-$\INCR{2^n}$~transformation~${\controlled \bigl( \INCR{2^n} \bigr)}$,
by its construction.

\paragraph{Threshold-check transformations}

For any integers~$t$~and~$z$,
let $\function{\threshold{t}}{\Integers}{\Binary}$ denote the function defined by
\[
  \threshold{t}(z)
  =
  \begin{cases}
    0 & \text{if ${z < t}$},\\
    1 & \text{if ${z \geq t}$}.
  \end{cases}
\]
For any positive integer~$n$ and for an integer~$t$ in $\Integers_{2^n}$,
The threshold-check transformation~$\THRESHOLD{t}{2^n}$ over~$\Integers_{2^n}$
is the unitary transformation acting over ${n + 1}$~qubits defined by
\[
  \THRESHOLD{t}{2^n}
  \colon
  \ket{b} \tensor \ket{j} \mapsto \bigket{b \oplus \threshold{t}(j)} \tensor \ket{j},
  \quad
  \forall b \in \Binary,
  \
  \forall j \in \Integers_{2^n}.
\]
The threshold-check transformation~$\THRESHOLD{t}{2^n}$ over~$\Integers_{2^n}$
is easily implemented as follows,
combining the increment transformations~$\INCR{2^{n + 1}}$~and~$\INCR{2^n}$
over $\Integers_{2^{n + 1}}$~and~$\Integers_{2^n}$, respectively:
\[
  \THRESHOLD{t}{2^n}
  =
  \bigl( I \tensor \INCR{2^n} \bigr)^t \bigl( \INCR{2^{n + 1}} \bigr)^{2^n - t}.
\]
Accordingly, provided that generalized Toffoli transformations may be assumed available freely,
the threshold-check transformation~$\THRESHOLD{t}{2^n}$ over $\Integers_{2^n}$
may also be assumed available freely,
for each positive integer~$n$.


\subsection{$\boldsymbol{k}$-Clean-Qubit Computation and Complexity Classes}
\label{Subsection: DQCk computation and complexity classes}

For any positive integer~$k$,
a \emph{quantum computation with $k$~clean~qubits},
or simply a \emph{$k$-clean-qubit computation},
is a computation performed by a unitary quantum circuit~$Q$ acting over $w$~qubits,
where $w$ is a positive integer satisfying ${w \geq k}$.
It is assumed that one of the qubits to which the circuit~$Q$ is applied
is designated as the output qubit.
The $k$-clean-qubit computation specified by the circuit~$Q$ proceeds as follows.
For simplicity,
we identify the quantum circuit~$Q$ with the unitary operator it induces.
The initial state of the computation is the $w$-qubit state
\[
  \rho_\init^{(w,k)} = (\ketbra{0})^{\tensor k} \tensor \biggl( \frac{I}{2} \biggr)^{\tensor (w - k)}.
\]
The circuit~$Q$ is applied to this initial state,
which generates the $w$-qubit state
\[
  \rho_\final
  =
  Q \rho_\init^{(w,k)} \conjugate{Q}.
\]
Now the designated output qubit is measured in the computational basis,
where the outcome~$\ket{0}$ is interpreted as ``accept'',
and the outcome~$\ket{1}$ is interpreted as ``reject''.
Such a computation may also be called a \emph{quantum computation of $\DQC{k}$ type},
or simply a \emph{$\DQC{k}$ computation},
in analogy to the $\DQC{1}$ computation.

The complexity classes~${\QfP{f}(c,s)}$~and~${\QOfP{f}(c,s)}$ are defined as follows.

\begin{definition}
  Given a function~$\function{f}{\Nonnegative}{\Natural}$
  and functions~$\function{c,s}{\Nonnegative}{[0,1]}$ satisfying ${c > s}$,
  a promise problem~${A = (A_\yes, A_\no)}$ is in ${\QfP{f}(c,s)}$
  iff there exists a polynomial-time uniformly generated family~${\{Q_x\}_{x \in \Sigma^\ast}}$ of quantum circuits
  such that,
  for every $x$ in $\Sigma^\ast$,
  $Q_x$ acts over ${w(\abs{x})}$~qubits
  for some polynomially bounded function~$\function{w}{\Nonnegative}{\Natural}$
  satisfying ${w \geq f}$
  and has the following properties:
  \begin{description}
  \item[(Completeness)]
    if $x$ is in $A_\yes$,
    the ${f(\abs{x})}$-clean-qubit computation induced by $Q_x$
    results in acceptance with probability at least ${c(\abs{x})}$,
  \item[(Soundness)]
    if $x$ is in $A_\no$,
    the ${f(\abs{x})}$-clean-qubit computation induced by $Q_x$
    results in acceptance with probability at most ${s(\abs{x})}$.
  \end{description}
  \label{Definition: QfP(c,s)}
\end{definition}

\begin{definition}
  Given a function~$\function{f}{\Nonnegative}{\Natural}$
  and functions~$\function{c,s}{\Nonnegative}{[0,1]}$ satisfying ${c > s}$,
  a promise problem~${A = (A_\yes, A_\no)}$ is in ${\QOfP{f}(c,s)}$
  iff $A$ is in ${\QfP{g}(c,s)}$ for some function~$\function{g}{\Nonnegative}{\Natural}$ satisfying ${g \in O(f)}$.
  \label{Definition: QOfP(c,s)}
\end{definition}

Using these definitions,
the complexity classes~$\BQfP{f}$~and~$\BQOfP{f}$ are defined as follows.

\begin{definition}
  Given a function~$\function{f}{\Nonnegative}{\Natural}$,
  a promise problem~${A = (A_\yes, A_\no)}$ is in $\BQfP{f}$
  iff $A$ is in ${\QfP{f} \bigl(\frac{2}{3}, \, \frac{1}{3} \bigr)}$.
  \label{Definition: BQfP}
\end{definition}

\begin{definition}
  Given a function~$\function{f}{\Nonnegative}{\Natural}$,
  a promise problem~${A = (A_\yes, A_\no)}$ is in $\BQOfP{f}$
  iff $A$ is in ${\QOfP{f} \bigl(\frac{2}{3}, \, \frac{1}{3} \bigr)}$.
  \label{Definition: BQOfP}
\end{definition}

Some remarks are in order on the definitions of $\BQfP{f}$ and $\BQOfP{f}$.

By Theorem~\ref{Theorem: error reduction for two-sided-error case} to be proved,
for any logarithmically bounded function~$f$,
the above definition of $\BQOfP{f}$ is equivalent to
a more conservative definition where the class consists of problems that are in
${\QOfP{f}(1 - \varepsilon, \varepsilon)}$
for some negligible function~$\function{\varepsilon}{\Nonnegative}{[0,1]}$.
Similarly, for any logarithmically bounded function~${f \geq 2}$,
the class~$\BQfP{f}$ above is equivalent to the class of problems that are in
${\QfP{f}(1 - \varepsilon, \varepsilon)}$
for some negligible function~$\function{\varepsilon}{\Nonnegative}{[0,1]}$.
Another candidate of definitions is to use $\BQfP{f}'$~and~$\BQOfP{f}'$,
which are the unions of ${\QfP{f}(c, s)}$ and ${\QOfP{f}(c, s)}$, respectively,
over all functions~$\function{c,s}{\Nonnegative}{[0,1]}$ satisfying ${c - s \geq \frac{1}{p}}$
for some polynomially bounded function~$\function{p}{\Nonnegative}{\Natural}$.
As the computation error can be reduced by Theorem~\ref{Theorem: error reduction for two-sided-error case}
only when there is a constant gap between completeness and soundness in the starting computation,
it remains open whether these $\BQfP{f}'$~and~$\BQOfP{f}'$
are equal to $\BQfP{f}$~and~$\BQOfP{f}$ defined above.

Finally, the complexity classes~$\NQoneP$~and~$\SBQoneP$ are defined as follows,
which are the one-clean-qubit-computation analogues of $\NQP$ and $\SBQP$, respectively.

\begin{definition}
  A promise problem~${A = (A_\yes, A_\no)}$ is in $\NQoneP$
  iff $A$ is in ${\QoneP(c, 0)}$
  for some positive-valued function~$\function{c}{\Nonnegative}{(0,1]}$.
  \label{Definition: NQoneP}
\end{definition}

\begin{definition}
  Given a function~$\function{f}{\Nonnegative}{\Natural}$,
  a promise problem~${A = (A_\yes, A_\no)}$ is in $\SBQoneP$
  iff $A$ is in ${\QoneP(2^{-p}, 2^{-p-1})}$
  for some polynomially bounded function~$\function{p}{\Nonnegative}{\Natural}$.
  \label{Definition: SBQoneP}
\end{definition}

\begin{remark}
  As will be proved in Subsection~\ref{Subsection: proofs of hardness of weak simulatability},
  the class~$\SBQoneP$ has the following amplification property similar to $\SBQP$ and $\SBP$:
  a problem~${A = (A_\yes, A_\no)}$ is in $\SBQoneP$
  iff
  for any polynomially bounded function~$\function{p}{\Nonnegative}{\Natural}$,
  there exists a polynomially bounded function~$\function{q}{\Nonnegative}{\Natural}$
  such that $A$ is in ${\QoneP \bigl( 2^{-q} \cdot (1 - 2^{-p}), \, 2^{-q} \cdot 2^{-p} \bigr)}$.
\end{remark}


\subsection{Remarks on Uniformity of Quantum Circuits}
\label{Subsection: remarks on uniformity}

The definition of polynomial-time uniformly generated family of quantum circuits
in Subsection~\ref{Subsection: quantum circuits}
allows the input~$x$ to be ``hard-coded'' in the generated circuit~$Q_x$,
i.e., each circuit~$Q_x$ is allowed to depend on the input~$x$ itself.
As mentioned before, a more standard definition of the polynomial-time uniformity
is the one in which each circuit generated depends only on the input length~$\abs{x}$,
and the input~$x$ is given to the circuit as a read-only input.
It is stressed that
all the results in this paper remain valid
even when this more standard ``non-hard-coded'' definition is adopted
for the polynomial-time uniformity of quantum circuits.
It is further stressed that
all the results in this paper still remain valid
even when the complexity classes in Subsection~\ref{Subsection: DQCk computation and complexity classes}
are defined with
the \emph{logarithmic-space} uniformly generated family of quantum circuits,
by suitably replacing polynomial-time computable functions by logarithmic-space computable functions
in the statements of
Theorems~\ref{Theorem: error reduction for perfect-completeness case}~and~\ref{Theorem: co-RQlogP is in co-RQ1P}
and Corollary~\ref{Corollary: error reduction for perfect-soundness cases},
and by changing the definitions of polynomially and logarithmically bounded functions
by using logarithmic-space deterministic Turing machines and logarithmic-space computable functions.
In this case,
the completeness results of Theorem~\ref{Theorem: completeness of the trace estimation problem}
hold even under logarithmic-space many-one reduction.

One thing to be mentioned is, however, that
some complexity classes
defined in Subsection~\ref{Subsection: DQCk computation and complexity classes},
such as ${\QlogP(c,s)}$, ${\QoneP(c,s)}$, $\BQlogP$, and $\BQoneP$,
may depend on the definition of the uniformity of quantum circuits.
Indeed, with the ``hard-coded'' definition of polynomial-time uniformity,
the class~$\classP$ is trivially contained in each class
defined in Subsection~\ref{Subsection: DQCk computation and complexity classes},
whereas it becomes unclear whether $\classP$ is included in the bounded-error classes
such as $\BQlogP$ and $\BQoneP$
when the classes are defined
with the standard ``non-hard-coded'' definition of polynomial-time uniformity
and with logarithmic-space uniformity.
There even exists an oracle relative to which
$\classP$ is not included in $\BQoneP$,
when $\BQoneP$ is defined with logarithmic-space uniformity~\cite{She06arXiv, She09PhD}.
In this regard,
when using the standard ``non-hard-coded'' definition of polynomial-time uniformity,
polynomial-time many-one reduction
may be too powerful for the models discussed
in that the reduction itself already has computational power enough to solve any problem in $\classP$,
while it is unclear whether the models for which the completeness results are discussed
have such computational power,
although Theorem~\ref{Theorem: completeness of the trace estimation problem} itself
is mathematically valid even in this case.

In short,
which uniformity of either ``hard-coded'' polynomial-time,
or standard ``non-hard-coded'' polynomial-time,
or logarithmic-space
is used to define these complexity classes
\emph{does not} affect the properties proved in the present paper,
but may affect how large these complexity classes themselves are.


\section{Building Blocks}
\label{Section: building blocks}

First, some notations are summarized that are used throughout this section.

Consider any quantum circuit~$Q$ acting over $w$~qubits.
For any positive integer~${k \leq w}$,
let ${p_\acc(Q, k)}$ denote
the acceptance probability of the $k$-clean-qubit computation induced by $Q$.
More precisely,
for any positive integer~${k \leq w}$,
let $\rho_\init^{(w,k)}$ be the $w$-qubit initial state defined by
\[
  \rho_\init^{(w,k)} = (\ketbra{0})^{\tensor k} \tensor \biggl( \frac{I}{2} \biggr)^{\tensor (w - k)},
\]
and let $\Pi_\acc$ be the projection operator defined by
\[
  \Pi_\acc = \ketbra{0} \tensor I^{\tensor (w - 1)}.
\]
Now, for any positive integer~${k \leq w}$,
the acceptance probability~${p_\acc(Q, k)}$ is defined by
\[
  p_\acc(Q, k) = \tr \Pi_\acc Q \rho_\init^{(w,k)} \conjugate{Q}.
\]


\subsection{Simulating $\boldsymbol{k}$-Clean-Qubit Computation by One-Clean-Qubit Computation}
\label{Subsection: One-Clean-Qubit Simulation Procedure}

This subsection presents a procedure,
called the \textsc{One-Clean-Qubit Simulation Procedure},
that constructs another quantum circuit~$R$ from $Q$
such that the $k$-clean-qubit computation induced by $Q$
can be simulated by the one-clean-qubit computation induced by $R$.
More formally,
given a description of a quantum circuit~$Q$
(which also specifies the number~$w$ of qubits $Q$ acts over)
and an integer~$k$ that specifies the number of clean qubits
used in the original computation to be simulated,
the \textsc{One-Clean-Qubit Simulation Procedure}
corresponds to a quantum circuit~$R$
such that the original $k$-clean-qubit computation induced by $Q$
is simulated by the one-clean-qubit computation induced by $R$.

The circuit~$R$ acts over ${(w + 1)}$~qubits,
which are divided into three quantum registers:
a single-qubit register~$\regO$,
a $k$-qubit register~$\regQ$,
and a ${(w - k)}$-qubit register~$\regR$.
It is supposed that the qubit in $\regO$ is initially in state~$\ket{0}$,
and all the qubits in $\regQ$ and $\regR$ are initially in the totally mixed state~${I/2}$.
Let $\regQ^{(1)}$ denote the single-qubit register consisting of the first qubit of $\regQ$,
which corresponds to the output qubit of $Q$.
First, the circuit~$R$ applies the Hadamard transformation~$H$ to $\regO$
if all the qubits in $\regQ$ are in state $\ket{0}$.
Next, $R$ applies $Q$ to ${(\regQ, \regR)}$,
where the qubits in $\regQ$
correspond to the $k$~clean qubits of the original $k$-clean-qubit computation induced by $Q$.
Now $R$ flips the phase
if and only if ${(\regO, \regQ^{(1)})}$ contains $11$
(i.e., $R$ applies the controlled-$Z$ transformation~${\controlled(Z)}$ to ${(\regO, \regQ^{(1)})}$).
$R$ then applies $\conjugate{Q}$ to ${(\regQ, \regR)}$,
and further applies $H$ to $\regO$ if all the qubits in $\regQ$ are in state $\ket{0}$.
Finally, the qubit in $\regO$ is measured in the computational basis,
and $R$ outputs the measurement result.

For the actual construction,
the first conditional application of $H$ to $\regO$ when all the qubits in $\regQ$ are in state~$\ket{0}$
is essentially realized by
first applying the $k$-controlled-$H$ transformation~${\controlled^k(H)}$ to ${(\regO, \regQ)}$
using the qubit in $\regO$ as the target,
and then applying the $\NOT$ transformation~$X$ to each of the qubits in $\regQ$.
Similarly, the second conditional application of $H$ to $\regO$
is essentially realized by
first applying $X$ to each of the qubits in $\regQ$,
and then applying ${\controlled^k(H)}$ to ${(\regO, \regQ)}$ using the qubit in $\regO$ as the target.
Figures~\ref{Figure: One-Clean-Qubit Simulation Procedure}~and~\ref{Figure: quantum circuit for One-Clean-Qubit Simulation Procedure}
summarize the construction of $R$.

\begin{figure}[t!]
  \begin{algorithm*}{\textsc{One-Clean-Qubit Simulation Procedure}}
    \begin{step}
    \item
      Prepare a single-qubit register~$\regO$
      a $k$-qubit register~$\regQ$,
      and a ${(w - k)}$-qubit register~$\regR$,
      where the qubit in $\regO$ is supposed to be initially in state~$\ket{0}$,
      while all the qubits in $\regQ$ and $\regR$ are supposed to be initially in the totally mixed state~${I/2}$.\\
      Apply the $k$-controlled-$H$ transformation~${\controlled^k(H)}$ to ${(\regO, \regQ)}$
      using the qubit in $\regO$ as the target,
      and then apply the $\NOT$ transformation~$X$ to each of the qubits in $\regQ$
      (applying these transformations has essentially the same effect as
      conditionally applying $H$ to $\regO$ if all the qubits in $\regQ$ are in state $\ket{0}$).
    \item
      Apply $Q$ to ${(\regQ, \regR)}$.
    \item
      Apply the phase-flip (i.e., multiply the phase by $-1$)
      if the content of ${(\regO, \regQ^{(1)})}$ is $11$,
      where $\regQ^{(1)}$ denotes the single-qubit register consisting of the first qubit of $\regQ$.
    \item
      Apply $\conjugate{Q}$ to ${(\regQ, \regR)}$.
    \item
      Apply $X$ to each of the qubits in $\regQ$,
      and then apply ${\controlled^k(H)}$ to ${(\regO, \regQ)}$
      using the qubit in $\regO$ as the target
      (applying these transformations has essentially the same effect as
      conditionally applying $H$ to $\regO$ if all the qubits in $\regQ$ are in state $\ket{0}$).
      Measure the qubit in $\regO$ in the computational basis.
      Accept if this results in $\ket{0}$, and reject otherwise.
    \end{step}
  \end{algorithm*}
  \caption{
    The \textsc{One-Clean-Qubit Simulation Procedure}
    induced by a quantum circuit~$Q$
    with the specification of the number~$k$ of clean qubits
    used in the computation of $Q$ to be simulated.
  }
  \label{Figure: One-Clean-Qubit Simulation Procedure}
\end{figure}

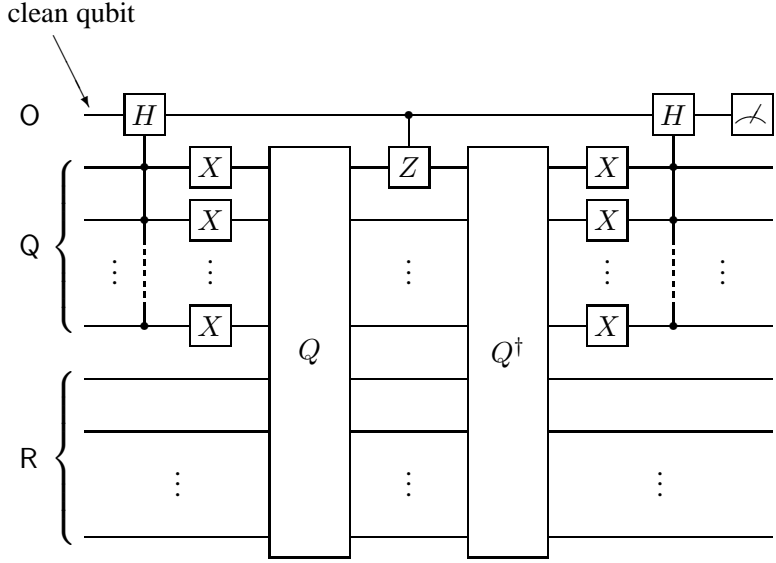
\begin{figure}[t!]
  \vspace{\baselineskip}
  \begin{center}
    \setlength{\unitlength}{1pt}
    \begin{picture}(290, 215)(  0,   0)
      \put(  0, 200){\makebox( 50,  15){clean qubit}}
      \put( 18, 200){\vector(1,-2){13.5}}
      \put(  5, 160){\makebox( 15,  20){$\regO$ \hfill}}
      \put(  5,  80){\makebox( 15,  80){$\regQ$ \hfill}}
      \put( 15,  85){\makebox( 15,  70){${\left\{ \rule{0mm}{35 \unitlength} \right.}$}}
      \put(  5,   0){\makebox( 15,  80){$\regR$ \hfill}}
      \put( 15,   5){\makebox( 15,  70){${\left\{ \rule{0mm}{35 \unitlength} \right.}$}}
      \put( 30, 170){\line(1,0){15}}
      \put( 30, 150){\line(1,0){22.5}}
      \put( 30, 130){\line(1,0){22.5}}
      \put( 30,  93.25){\makebox( 22.5,  40){$\vdots$}}
      \put( 30,  90){\line(1,0){22.5}}
      \put( 30,  70){\line(1,0){70}}
      \put( 30,  50){\line(1,0){70}}
      \put( 30,  13.25){\makebox( 70,  40){$\vdots$}}
      \put( 30,  10){\line(1,0){70}}
      \put( 45, 162.5){\framebox( 15,  15){$H$}}
      \put( 52.5, 121){\line(0,1){41.5}}
      \multiput( 52.5, 101)(0,4){5}{\line(0,1){2}}
      \put( 52.5,  90){\line(0,1){9}}
      \put( 52.5, 150){\circle*{3.2}}
      \put( 52.5, 130){\circle*{3.2}}
      \put( 52.5,  90){\circle*{3.2}}
      \put( 52.5, 150){\line(1,0){17.5}}
      \put( 52.5, 130){\line(1,0){17.5}}
      \put( 52.5,  90){\line(1,0){17.5}}
      \put( 60, 170){\line(1,0){92.5}}
      \put( 70, 142.5){\framebox( 15,  15){$X$}}
      \put( 70, 122.5){\framebox( 15,  15){$X$}}
      \put( 70, 100.75){\makebox( 15,  25){$\vdots$}}
      \put( 70,  82.5){\framebox( 15,  15){$X$}}
      \put( 85, 150){\line(1,0){15}}
      \put( 85, 130){\line(1,0){15}}
      \put( 85,  90){\line(1,0){15}}
      \put(100,   2.5){\framebox( 30,  155){$Q$}}
      \put(130, 150){\line(1,0){15}}
      \put(130, 130){\line(1,0){45}}
      \put(130,  93.25){\makebox( 45,  40){$\vdots$}}
      \put(130,  90){\line(1,0){45}}
      \put(130,  70){\line(1,0){45}}
      \put(130,  50){\line(1,0){45}}
      \put(130,  13.25){\makebox( 45,  40){$\vdots$}}
      \put(130,  10){\line(1,0){45}}
      \put(145, 142.5){\framebox( 15,  15){$Z$}}
      \put(152.5, 170){\circle*{3.2}}
      \put(152.5, 157.5){\line(0,1){12.5}}
      \put(152.5, 170){\line(1,0){92.5}}
      \put(160, 150){\line(1,0){15}}
      \put(175,   2.5){\framebox( 30,  155){$\conjugate{Q}$}}
      \put(205, 150){\line(1,0){15}}
      \put(205, 130){\line(1,0){15}}
      \put(205,  90){\line(1,0){15}}
      \put(205,  70){\line(1,0){85}}
      \put(205,  50){\line(1,0){85}}
      \put(205,  13.25){\makebox( 85,  40){$\vdots$}}
      \put(205,  10){\line(1,0){85}}
      \put(220, 142.5){\framebox( 15,  15){$X$}}
      \put(220, 122.5){\framebox( 15,  15){$X$}}
      \put(220, 100.75){\makebox( 15,  25){$\vdots$}}
      \put(220,  82.5){\framebox( 15,  15){$X$}}
      \put(235, 150){\line(1,0){17.5}}
      \put(235, 130){\line(1,0){17.5}}
      \put(235,  90){\line(1,0){17.5}}
      \put(245, 162.5){\framebox( 15,  15){$H$}}
      \put(252.5, 121){\line(0,1){41.5}}
      \multiput(252.5, 101)(0,4){5}{\line(0,1){2}}
      \put(252.5,  90){\line(0,1){9}}
      \put(252.5, 150){\circle*{3.2}}
      \put(252.5, 130){\circle*{3.2}}
      \put(252.5,  90){\circle*{3.2}}
      \put(252.5, 150){\line(1,0){37.5}}
      \put(252.5, 130){\line(1,0){37.5}}
      \put(252.5,  93.25){\makebox( 37.5,  40){$\vdots$}}
      \put(252.5,  90){\line(1,0){37.5}}
      \put(260, 170){\line(1,0){15}}
      \put(275, 162.5){\framebox( 15,  15){\metermark}}
    \end{picture}
  \end{center}
  \caption{
    The quantum circuit for the \textsc{One-Clean-Qubit Simulation Procedure}
    induced by a quantum circuit~$Q$ with a specification of the quantum register~$\regQ$,
    where the qubit in $\regO$ is supposed to be the only clean qubit at the beginning of the computation,
    and is also the output qubit of the computation.
    The combinations of the controlled-Hadamard transformation and the $\NOT$ transformations~$X$
    essentially correspond to applying the Hadamard transformation
    if all the qubits in $\regQ$ are in state~$\ket{0}$.
  }
  \label{Figure: quantum circuit for One-Clean-Qubit Simulation Procedure}
\end{figure}

\begin{proposition}
  For any quantum circuit~$Q$ and any positive integer~$k$,
  let $R$ be the quantum circuit
  corresponding to the \textsc{One-Clean-Qubit Simulation Procedure} induced by $Q$ and $k$.
  For the acceptance probability~${p_\acc(Q, k)}$
  of the $k$-clean-qubit computation induced by $Q$
  and the acceptance probability~${p_\acc(R, 1)}$
  of the one-clean-qubit computation induced by $R$,
  it holds that
  \[
    1 - 2^{-k} \bigl( 1 - (\op{p_\acc}(Q, k))^2 \bigr)
    \leq
    \op{p_\acc}(R, 1)
    \leq
    1 - 2^{-k} (1 - \op{p_\acc}(Q, k)).
  \]
  \label{Proposition: acceptance probability of One-Clean-Qubit Simulation Procedure}
\end{proposition}

\begin{proof}
  To analyze the acceptance probability~${p_\acc(R, 1)}$
  of the one-clean-qubit computation induced by $R$,
  suppose that the content of $\regQ$ is initially $q$ in $\Sigma^k$,
  and the content of $\regR$ is initially $r$ in $\Sigma^{(w - k)}$.

  First consider the case where $q$ is not all-zero.
  In this case, nothing is applied to the qubit in $\regO$ in Step~1,
  and thus the phase-flip is never performed in Step~3.
  Therefore, the application of $\conjugate{Q}$ in Step~4
  cancels out the application of $Q$ in Step~2,
  and thus the content of $\regQ$ remains $q$ in Step~5, which is not all-zero.
  Hence, nothing is applied to the qubit in $\regO$ in Step~5, either,
  and $R$ outputs $0$ and accepts.

  Now consider the case where $q$ is all-zero.
  In this case, by letting ${\ket{\psi_r} = \ket{0}^{\tensor k} \tensor \ket{r}}$
  the state in ${(\regO, \regQ, \regR)}$ just after Step~2 is given by
  \[
    \frac{1}{\sqrt{2}} (\ket{0} + \ket{1}) \tensor (Q \ket{\psi_r}).
  \]
  For each $j$ in $\Binary$,
  let $\Pi_j$ be the projection operator acting over $w$~qubits defined by
  ${\Pi_j = \ketbra{j} \tensor I^{\tensor (w - 1)}}$
  (notice that $\Pi_0$ is nothing but $\Pi_\acc$).
  Then the state in ${(\regO, \regQ, \regR)}$ just after Step~4 is given by
  \[
    \begin{split}
      \hspace{5mm}
      &
      \hspace{-5mm}
      \frac{1}{\sqrt{2}} \ket{0} \tensor \ket{\psi_r}
      +
      \frac{1}{\sqrt{2}} \ket{1} \tensor \bigl[ \conjugate{Q} (\Pi_0 - \Pi_1) Q \ket{\psi_r} \bigr]
      \\
      &
      =
      \frac{1}{\sqrt{2}} \ket{0} \tensor \ket{\psi_r}
      +
      \frac{1}{\sqrt{2}}
        \ket{1}
        \tensor
        \bigl[ \conjugate{Q} \bigl( 2 \Pi_0 - I^{\tensor w} \bigr) Q \ket{\psi_r} \bigr]
      \\
      &
      =
      \frac{1}{\sqrt{2}} (\ket{0} - \ket{1}) \tensor \ket{\psi_r}
      +
      \sqrt{2} \, \ket{1} \tensor \bigl( \conjugate{Q} \Pi_0 Q \ket{\psi_r} \bigr).
    \end{split}      
  \]
  Further define the projection operators~$\Delta_0$~and~$\Delta_1$ acting over $w$~qubits by
  ${\Delta_0 = (\ketbra{0})^{\tensor k} \tensor I^{\tensor (w - k)}}$
  and
  ${\Delta_1 = I^{\tensor w} - \Delta_0}$.
  Then the state in ${(\regO, \regQ, \regR)}$ just after Step~5 is given by
  \[
    \ket{1} \tensor \ket{\psi_r}
    +
    (\ket{0} - \ket{1}) \tensor \bigl( \Delta_0 \conjugate{Q} \Pi_0 Q \ket{\psi_r} \bigr)
    +
    \sqrt{2} \, \ket{1} \tensor \bigl( \Delta_1 \conjugate{Q} \Pi_0 Q \ket{\psi_r} \bigr).
  \]
  Hence, the probability that $R$ outputs $0$ is given by
  ${\bignorm{\Delta_0 \conjugate{Q} \Pi_0 Q \ket{\psi_r}}^2}$.

  It follows that the overall acceptance probability~${p_\acc(R, 1)}$
  of the one-clean-qubit computation induced by $R$ is given by
  \begin{equation}
    \op{p_\acc}(R, 1)
    =
    (1 - 2^{-k})
    +
    2^{-k}
    \cdot
    2^{-(w - k)} \sum_{r \in \Sigma^{(w - k)}} \bignorm{\Delta_0 \conjugate{Q} \Pi_0 Q \ket{\psi_r}}^2.
    \label{Equation: acceptance probability of One-Clean-Qubit Simulation Procedure}
  \end{equation}

  First notice that
  ${
    \Delta_0
    =
    (\ketbra{0})^{\tensor k} \tensor \sum_{r \in \Sigma^{(w - k)}} \ketbra{r}
    =
    \sum_{r \in \Sigma^{(w - k)}} \ketbra{\psi_r}
  }$,
  and thus, it holds that
  \[
    \bignorm{\Delta_0 \conjugate{Q} \Pi_0 Q \ket{\psi_r}}
    \geq
    \bignorm{\ketbra{\psi_r} \conjugate{Q} \Pi_0 Q \ket{\psi_r}}
    =
    \bigabs{\bra{\psi_r} \conjugate{Q} \Pi_0 Q \ket{\psi_r}}
    =
    \bignorm{\Pi_0 Q \ket{\psi_r}}^2,
  \]
  which is exactly the acceptance probability of the circuit~$Q$
  when the input state to it was ${\ket{\psi_r} = \ket{0}^{\tensor k} \tensor \ket{r}}$.
  As the acceptance probability~${p_\acc(Q, k)}$
  of the $k$-clean-qubit computation induced by $Q$
  is nothing but the expected value of ${\bignorm{\Pi_0 Q \ket{\psi_r}}^2}$ over $r$ in $\Sigma^{(w - k)}$,
  it follows that
  \[
    \begin{split}
      \hspace{5mm}
      &
      \hspace{-5mm}
      2^{-(w - k)} \sum_{r \in \Sigma^{(w - k)}} \bignorm{\Delta_0 \conjugate{Q} \Pi_0 Q \ket{\psi_r}}^2
      \\
      &
      \geq
      \biggl(
        2^{-(w - k)} \sum_{r \in \Sigma^{(w - k)}} \bignorm{\Delta_0 \conjugate{Q} \Pi_0 Q \ket{\psi_r}}
      \biggr)^2
      \geq
      \biggl(
        2^{-(w - k)} \sum_{r \in \Sigma^{(w - k)}} \bignorm{\Pi_0 Q \ket{\psi_r}}^2
      \biggr)^2
      =
      (\op{p_\acc}(Q, k))^2.
    \end{split}
  \]
  Combined with the equation~(\ref{Equation: acceptance probability of One-Clean-Qubit Simulation Procedure}),
  this implies that
  \[
    \op{p_\acc}(R, 1)
    \geq
    (1 - 2^{-k}) + 2^{-k} (\op{p_\acc}(Q, k))^2
    =
    1 - 2^{-k} \bigl( 1 - (\op{p_\acc}(Q, k))^2 \bigr),
  \]
  which provides the first inequality.

  Now notice that ${\bignorm{\Delta_0 \conjugate{Q} \Pi_0 Q \ket{\psi_r}}^2}$
  is at most ${\bignorm{\Pi_0 Q \ket{\psi_r}}^2}$,
  which is again exactly the acceptance probability of the circuit~$Q$
  when the input state to it was ${\ket{\psi_r} = \ket{0}^{\tensor k} \tensor \ket{r}}$.
  Again using the fact that
  ${p_\acc(Q, k)}$
  is nothing but the expected value of ${\bignorm{\Pi_0 Q \ket{\psi_r}}^2}$ over $r$ in $\Sigma^{(w - k)}$,
  it holds that
  \[
    2^{-(w - k)}
    \sum_{r \in \Sigma^{(w - k)}}
      \bignorm{\Delta_0 \conjugate{Q} \Pi_0 Q \ket{\psi_r}}^2
    \leq
    2^{-(w - k)}
    \sum_{r \in \Sigma^{(w - k)}}
      \bignorm{\Pi_0 Q \ket{\psi_r}}^2
    =
    \op{p_\acc}(Q, k).
  \]
  Combined with the equation~(\ref{Equation: acceptance probability of One-Clean-Qubit Simulation Procedure}),
  this implies that
  \[
    \op{p_\acc}(R, 1)
    \leq
    (1 - 2^{-k}) + 2^{-k} \op{p_\acc}(Q, k)
    =
    1 - 2^{-k} (1 - \op{p_\acc}(Q, k)),
  \]
  and the second inequality follows.
\end{proof}


\subsection{Amplifying Randomness of One-Clean-Qubit Computation}
\label{Subsection: Randomness Amplification Procedure}

This subsection presents a procedure,
called the \textsc{Randomness Amplification Procedure},
that constructs another quantum circuit~$R^{(N)}$ from $Q$ when a positive integer~$N$ is specified.
The circuit~$R^{(N)}$ is designed
so that the sequence~${\bigl\{ p_\acc \bigl( R^{(N)}, 1 \bigr) \bigr\}_{N \in \Natural}}$
of the acceptance probability of the one-clean-qubit computation induced by $R^{(N)}$
converges linearly to $1/2$
with a rate related to
the acceptance probability~${p_\acc(Q, 1)}$ of the one-clean-qubit computation induced by $Q$,
if ${0 < p_\acc(Q, 1) < 1}$.

For each $N$ in $\Natural$,
the circuit~$R^{(N)}$ acts over ${Nw  + 1}$~qubits,
which are divided into ${2N + 1}$~quantum registers:
a single-qubit register~$\regO$,
and a ${(w - 1)}$-qubit register~$\regR_j$ and a single-qubit register~$\regX_j$
for each $j$ in ${\{1, \dotsc, N\}}$.
It is supposed that the qubit in $\regO$ is initially in state~$\ket{0}$,
and all the qubits in $\regR_j$ and $\regX_j$ are initially in the totally mixed state~${I/2}$,
for all $j$ in ${\{1, \dotsc, N\}}$.
For ${j = 1}$ to $N$,
the circuit~$R^{(N)}$ repeats the process of
first applying $Q$ to ${(\regO, \regR_j)}$
and then applying the $\CNOT$ transformation to ${(\regO, \regX_j)}$
using the qubit in $\regO$ as the control,
where each application of the $\CNOT$ transformation
essentially has the same effect as measuring the qubit in $\regO$ in the computational basis
every time after $Q$ is applied.
Finally, the qubit in $\regO$ is measured in the computational basis,
and $R^{(N)}$ outputs the measurement result.
Figures~\ref{Figure: Randomness Amplification Procedure}~and~\ref{Figure: quantum circuit for Randomness Amplification Procedure}
summarize the construction of $R^{(N)}$.
(Strictly speaking, when ${j = N}$ at Step~2 of Figure~\ref{Figure: Randomness Amplification Procedure},
it is redundant to apply $\CNOT$ at Step~2.2,
as the qubit in $\regO$ is anyway measured at Step~3 in the computational basis.
This point is reflected in Figure~\ref{Figure: quantum circuit for Randomness Amplification Procedure}.)

\begin{figure}[t!]
  \begin{algorithm*}{\textsc{Randomness Amplification Procedure}}
    \begin{step}
    \item
      Prepare a single-qubit register~$\regO$,
      where the qubit in $\regO$ is supposed to be initially in state~$\ket{0}$.
      Prepare a ${(w - 1)}$-qubit register~$\regR_j$
      and a single-qubit register~$\regX_j$,
      for each $j$ in ${\{1, \dotsc, N\}}$,
      where all the qubits in $\regR_j$ and $\regX_j$ are supposed to be initially in the totally mixed state~${I/2}$.
    \item
      For ${j = 1}$ to $N$, perform the following:
      \begin{step}
      \item
        Apply $Q$ to ${(\regO, \regR_j)}$.
      \item
        Apply the $\CNOT$ transformation to ${(\regO, \regX_j)}$
        with the qubit in $\regO$ being the control.
      \end{step}
    \item
      Measure the qubit in $\regO$ in the computational basis.
      Accept if this results in $\ket{0}$, and reject otherwise.
    \end{step}
  \end{algorithm*}
  \caption{The \textsc{Randomness Amplification Procedure}.}
  \label{Figure: Randomness Amplification Procedure}
\end{figure}

\begin{figure}[t!]
  \vspace{\baselineskip}
  \begin{center}
    \setlength{\unitlength}{1pt}
    \begin{picture}(400, 450)(  0,   0)
      \put(  0, 305)
        {%
          \begin{picture}(400, 145)(  0,   0)
            \put(  0, 130){\makebox( 50,  15){clean qubit}}
            \put( 18, 130){\vector(1,-2){13.5}}
            \put(  0,  90){\makebox( 15,  20){$\regO$ \hfill}}
            \put(  0,  10){\makebox( 15,  80){$\regR_1$ \hfill}}
            \put( 15,  15){\makebox( 15,  70){${\left\{ \rule{0mm}{35 \unitlength} \right.}$}}
            \put(  0, -10){\makebox( 15,  20){$\regX_1$ \hfill}}
            \put( 30, 100){\line(1,0){15}}
            \put( 30,  80){\line(1,0){15}}
            \put( 30,  60){\line(1,0){15}}
            \put( 30,  23.25){\makebox( 15,  40){$\vdots$}}
            \put( 30,  20){\line(1,0){15}}
            \put( 30,   0){\line(1,0){59}}
            \put( 45,   0)
              {%
                \begin{picture}( 47.5, 107.5)(  0,   0)
                  \put(  0,  12.5){\framebox( 30,  95){$Q$}}
                  \put( 30, 100){\line(1,0){17.5}}
                  \put( 47.5, 100){\circle*{3}}
                  \put( 47.5, 100){\line(0,-1){96.5}}
                  {
                    \linethickness{0.265pt}
                    \put( 47.5,   0){\circle{7}}
                    \put( 44,   0){\line(1,0){7}}
                    \put( 47.5,   3.5){\line(0,-1){7}}
                  }
                \end{picture}
              }
            \put( 75,  80){\line(1,0){140}}
            \put( 75,  60){\line(1,0){140}}
            \put( 75,  23.25){\makebox(140,  40){$\vdots$}}
            \put( 75,  20){\line(1,0){140}}
            \put( 92.5, 100){\line(1,0){17.5}}
            \put( 96,   0){\line(1,0){119}}
            {%
              \begin{xy}
                \put(110, 100){\ar@{-} <20pt, -120pt>}
              \end{xy}
            }
            \multiput(217,  80)(4,0){27}{\line(1,0){2}}
            \multiput(217,  60)(4,0){27}{\line(1,0){2}}
            \put(215,  23.25){\makebox(110,  40){$\vdots$}}
            \multiput(217,  20)(4,0){27}{\line(1,0){2}}
            \multiput(217,   0)(4,0){27}{\line(1,0){2}}
            \put(325,  80){\line(1,0){75}}
            \put(325,  60){\line(1,0){75}}
            \put(325,  23.25){\makebox( 75,  40){$\vdots$}}
            \put(325,  20){\line(1,0){75}}
            \put(325,   0){\line(1,0){75}}
          \end{picture}
        }
      \put(  0, 185)
        {%
          \begin{picture}(400, 100)(  0,   0)
            \put(  0,  10){\makebox( 15,  80){$\regR_2$ \hfill}}
            \put( 15,  15){\makebox( 15,  70){${\left\{ \rule{0mm}{35 \unitlength} \right.}$}}
            \put(  0, -10){\makebox( 15,  20){$\regX_2$ \hfill}}
            \put( 30,  80){\line(1,0){110}}
            \put( 30,  60){\line(1,0){110}}
            \put( 30,  23.25){\makebox(110,  40){$\vdots$}}
            \put( 30,  20){\line(1,0){110}}
            \put( 30,   0){\line(1,0){154}}
            \put(130, 100){\line(1,0){10}}
            \put(140,   0)
              {%
                \begin{picture}( 47.5, 107.5)(  0,   0)
                  \put(  0,  12.5){\framebox( 30,  95){$Q$}}
                  \put( 30, 100){\line(1,0){17.5}}
                  \put( 47.5, 100){\circle*{3}}
                  \put( 47.5, 100){\line(0,-1){96.5}}
                  {
                    \linethickness{0.265pt}
                    \put( 47.5,   0){\circle{7}}
                    \put( 44,   0){\line(1,0){7}}
                    \put( 47.5,   3.5){\line(0,-1){7}}
                  }
                \end{picture}
              }
            \put(170,  80){\line(1,0){45}}
            \put(170,  60){\line(1,0){45}}
            \put(170,  23.25){\makebox( 45,  40){$\vdots$}}
            \put(170,  20){\line(1,0){45}}
            \put(187.5, 100){\line(1,0){17.5}}
            \put(191,   0){\line(1,0){24}}
            {%
              \begin{xy}
                \put(205, 100){\ar@{-} <10pt, -60pt>}
                \multiput(215.33,  38.02)(0.66,-3.96){15}{\ar@{-} <0.33pt, -1.98pt>}
              \end{xy}
            }
            \multiput(217,  80)(4,0){27}{\line(1,0){2}}
            \multiput(217,  60)(4,0){27}{\line(1,0){2}}
            \put(215,  23.25){\makebox(110,  40){$\vdots$}}
            \multiput(217,  20)(4,0){27}{\line(1,0){2}}
            \multiput(217,   0)(4,0){27}{\line(1,0){2}}
            \put(325,  80){\line(1,0){75}}
            \put(325,  60){\line(1,0){75}}
            \put(325,  23.25){\makebox(75,  40){$\vdots$}}
            \put(325,  20){\line(1,0){75}}
            \put(325,   0){\line(1,0){75}}
          \end{picture}
        }
      \put(  2,  43.25){\makebox( 15, 145){$\vdots$ \hfill}}
      \put( 30,  73.25){\makebox(185, 115){$\vdots$}}
      \put(215,  73.25){\makebox(110, 115){$\vdots$}}
      \put(325, 100.75){\makebox( 75,  87.5){$\vdots$}}
      \put(  0,  10)
        {%
          \begin{picture}(400,  80)(  0,   0)
            \put(  0,   0){\makebox( 18,  60){$\regR_N$ \hfill}}
            \put( 15,  -5){\makebox( 15,  70){${\left\{ \rule{0mm}{35 \unitlength} \right.}$}}
            \put( 30,  60){\line(1,0){185}}
            \put( 30,  40){\line(1,0){185}}
            \put( 30,   3.25){\makebox(185,  40){$\vdots$}}
            \put( 30,   0){\line(1,0){185}}
            \multiput(217,  60)(4,0){27}{\line(1,0){2}}
            \multiput(217,  40)(4,0){27}{\line(1,0){2}}
            \put(215,   3.25){\makebox(110,  40){$\vdots$}}
            \multiput(217,   0)(4,0){27}{\line(1,0){2}}
            {%
              \begin{xy}
                \multiput(326.02, 103.88)(0.66,-3.96){3}{\ar@{-} <0.33pt, -1.98pt>}
                \put(328,  92){\ar@{-} <2pt, -12pt>}
              \end{xy}
            }
            \put(325,  60){\line(1,0){15}}
            \put(325,  40){\line(1,0){15}}
            \put(325,   0){\line(1,0){15}}
            \put(330,  80){\line(1,0){10}}
            \put(340,  -7.5){\framebox( 30,  95){$Q$}}
            \put(370,  80){\line(1,0){15}}
            \put(370,  60){\line(1,0){30}}
            \put(370,  40){\line(1,0){30}}
            \put(370,   3.25){\makebox( 30,  40){$\vdots$}}
            \put(370,   0){\line(1,0){30}}
            \put(385,  72.5){\framebox( 15,  15){\metermark}}
          \end{picture}
        }
    \end{picture}
  \end{center}
  \caption{
    The quantum circuit for the \textsc{Randomness Amplification Procedure}
    induced by a quantum circuit~$Q$,
    where the qubit in $\regO$ is supposed to be the only clean qubit at the beginning of the computation,
    and is also the output qubit of the computation.
  }
  \label{Figure: quantum circuit for Randomness Amplification Procedure}
\end{figure}
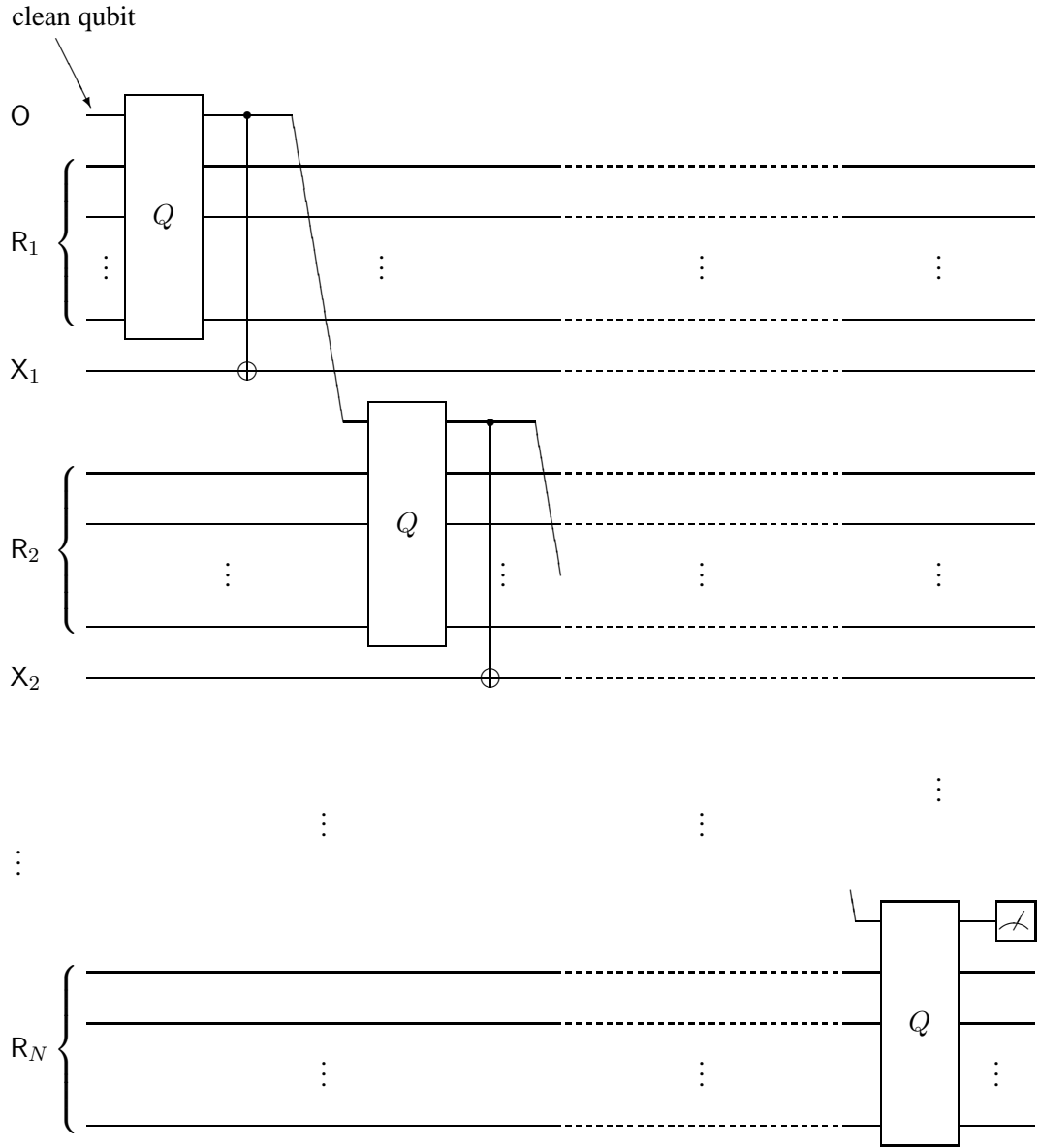

\begin{proposition}
  For any quantum circuit~$Q$ and any positive integer~$N$,
  let $R^{(N)}$ be the quantum circuit
  corresponding to the \textsc{Randomness Amplification Procedure} induced by $Q$ and $N$.
  For the acceptance probability~${p_\acc(Q, 1)}$
  of the one-clean-qubit computation induced by $Q$
  and the acceptance probability~${p_\acc \bigl( R^{(N)}, 1 \bigr)}$
  of the one-clean-qubit computation induced by $R^{(N)}$,
  it holds that
  \[
    \op{p_\acc} \bigl( R^{(N)}, 1 \bigr)
    =
    \frac{1}{2} + \frac{1}{2} (2 \op{p_\acc}(Q, 1) - 1)^N.
  \]
  \label{Proposition: acceptance probability of Randomness Amplification Procedure}
\end{proposition}

\begin{proof}
  For each $j$ in $\Binary$,
  let $\Pi_j$ be the projection operator acting over $w$~qubits defined by
  ${\Pi_j = \ketbra{j} \tensor I^{\tensor (w - 1)}}$,
  and let $\rho_j$ be the quantum state of $w$~qubits defined by
  ${\rho_j = \ketbra{j} \tensor \bigl( \frac{I}{2} \bigr)^{\tensor (w - 1)}}$.
  Note that ${\Pi_0 = \Pi_\acc}$ and ${\rho_0 = \rho_\init^{(w,1)}}$,
  and thus,
  the acceptance probability~${p_\acc(Q, 1)}$
  of the original one-clean-qubit computation induced by $Q$
  is given by
  \[
    \op{p_\acc}(Q, 1) = \tr \Pi_0 Q \rho_0 \conjugate{Q}.
  \]
  By noticing that
  ${\Pi_1 = I^{\tensor w} - \Pi_0}$
  and
  \[
    \tr \Pi_0 Q \rho_0 \conjugate{Q} + \tr \Pi_0 Q \rho_1 \conjugate{Q}
    =
    \frac{1}{2^{w - 1}} \tr \Pi_0 Q \bigl( I^{\tensor w} \bigr) \conjugate{Q}
    =
    \frac{1}{2^{w - 1}} \tr \Pi_0
    =
    1,
  \]
  it also holds that
  \[
    \tr \Pi_1 Q \rho_1 \conjugate{Q}
    =
    1 - \tr \Pi_0 Q \rho_1 \conjugate{Q}
    =
    1 - \bigl( 1 - \tr \Pi_0 Q \rho_0 \conjugate{Q} \bigr)
    =
    \tr \Pi_0 Q \rho_0 \conjugate{Q}
    =
    \op{p_\acc}(Q, 1).
  \]
  This implies that,
  after each repetition of Step~2 of the \textsc{Randomness Amplification Procedure},
  the content of $\regO$ remains unchanged with probability~${p_\acc(Q, 1)}$,
  and is flipped with probability~${1 - \op{p_\acc}(Q, 1)}$
  (by viewing that $\regO$ is ``measured'' in the computational basis
  as a result of the application of $\CNOT$ at Step~2.2).
  The content of $\regO$ is $0$ when entering Step~3
  if and only if the content of $\regO$ is flipped even number of times during Step~2.
  By Lemma~\ref{Lemma: coin-flipping lemma}, this happens with probability exactly
  ${\frac{1}{2} + \frac{1}{2} (2 \op{p_\acc}(Q, 1) - 1)^N}$,
  which gives the acceptance probability~${p_\acc \bigl( R^{(N)}, 1 \bigr)}$
  of the one-clean-qubit computation induced by $R^{(N)}$.
\end{proof}


\subsection{Checking Stability of One-Clean-Qubit Computation}
\label{Subsection: Stability Checking Procedures}

This subsection considers two procedures of checking
whether a given one-clean-qubit computation
has sufficiently high acceptance probability or not.
For a positive integer~$N$ specified,
the first procedure,
called the \textsc{One-Clean-Qubit Stability Checking Procedure},
constructs another quantum circuit~$R_1^{(N)}$ from $Q$
so that the acceptance probability~${p_\acc \bigl( R_1^{(N)}, 1 \bigr)}$
of the one-clean-qubit computation induced by $R_1^{(N)}$
is polynomially small with respect to $N$
when the acceptance probability was sufficiently close to $1/2$
in the original one-clean-qubit computation induced by $Q$.
The second procedure,
called the \textsc{Two-Clean-Qubit Stability Checking Procedure},
is a modification of the first procedure
so that, using two clean qubits,
the acceptance probability~${p_\acc \bigl( R_2^{(N)}, 2 \bigr)}$
of the two-clean-qubit computation induced by the circuit~$R_2^{(N)}$ constructed from $Q$
now becomes exponentially small with respect to $N$
when the acceptance probability was sufficiently close to $1/2$
in the original one-clean-qubit computation induced by $Q$.


\subsubsection{Stability Checking Using One Clean Qubit}
\label{Subsubsection: One-Clean-Qubit Stability Checking Procedure}

First consider a slightly simplified version of
the \textsc{One-Clean-Qubit Stability Checking Procedure}
described below.

For each $N$ in $\Natural$ satisfying ${N \geq 2}$,
the procedure to be constructed prepares ${2N + 1}$~quantum registers:
a single-qubit register~$\regQ$,
and a ${(w - 1)}$-qubit register~$\regR_j$ for each $j$ in ${\{1, \dotsc, 2N\}}$.
It is supposed
that all the qubits in $\regQ$ and $\regR_j$ are initially in the totally mixed state~${I/2}$,
for all $j$ in ${\{1, \dotsc, 2N\}}$.
The procedure also prepares an integer variable~$C$ that serves as a counter,
but the value of $C$ is not necessarily initialized to $0$ at the beginning of the computation,
and the initial value~$r$ of $C$ is chosen from the set~${\{0, \dotsc, N - 1\}}$ uniformly at random.
For ${j = 1}$ to ${2N}$,
the procedure repeats the process of
first applying $Q$ to ${(\regQ, \regR_j)}$
and then measuring the qubit in $\regQ$ in the computational basis.
Everytime this measurement results in $\ket{1}$,
the value of the counter~$C$ is increased by one.
After this repetition of ${2N}$~times,
the procedure checks the value of $C$,
and rejects if it is between $N$ and ${2N - 1}$.

For the actual construction of the above simplified procedure,
the positive integer~$N$ is chosen to be a power of two
so that the most significant bit is $0$ for any integer in the interval~${[0, N - 1]}$
and is $1$ for any integer in the interval~${[N, 2N - 1]}$,
when expressed as a binary string using ${\log N + 1}$~bits.
The actual procedure also introduces
a ${(\log N + 1)}$-qubit register~$\regC$
and a single-qubit register~$\regX_j$ for each $j$ in ${\{1, \dotsc, 2N\}}$.
Only the first qubit of $\regC$ is supposed to be initially in state~$\ket{0}$,
and all the other qubits used in the actual procedure
are supposed to be initially in the totally mixed state~${I/2}$.
The content of $\regC$ serves as a counter~$C$ of the simplified procedure.
The condition that only the first qubit in $\regC$ is initially in state~$\ket{0}$
and all the other qubits in $\regC$ are initially in the totally mixed state~${I/2}$
exactly corresponds to the process in the simplified procedure
of randomly picking the initial value of the counter~$C$ from the set~${\{0, \dotsc, N - 1\}}$.
The conditional increment of the counter value
is realized by the controlled-$\INCR{2N}$ transformation~${\controlled \bigl( \INCR{2N} \bigr)}$.
The final decision of acceptance and rejection of the simplified procedure
can be done by measuring the first qubit of $\regC$ in the computational basis.
Each register~$\regX_j$ is used
to simulate the measurement at each repetition round~$j$.
More precisely,
for each repetition round~$j$,
the measurement of the qubit in $\regQ$ in the computational basis
is replaced by the application of a $\CNOT$ transformation to ${(\regQ, \regX_j)}$.
Figures~\ref{Figure: One-Clean-Qubit Stability Checking Procedure}~and~\ref{Figure: quantum circuit for One-Clean-Qubit Stability Checking Procedure}
summarize the actual construction of the \textsc{One-Clean-Qubit Stability Checking Procedure}.

\begin{figure}[t!]
  \begin{algorithm*}{\textsc{One-Clean-Qubit Stability Checking Procedure}}
    \begin{step}
    \item
      Given a positive integer~$N$ that is a power of two,
      let ${l = \log N + 1}$.
      Prepare an $l$-qubit register~$\regC$,
      a single-qubit register~$\regQ$,
      and each ${(w - 1)}$-qubit register~$\regR_j$ for each $j$ in ${\{1, \dotsc, 2N\}}$.
      For each $j$ in ${\{1, \dotsc, l\}}$,
      let $\regC^{(j)}$ denote the single-qubit register
      corresponding to the $j$th qubit of $\regC$.
      The qubit in $\regC^{(1)}$ is supposed to be initially in state~$\ket{0}$,
      while the qubit in $\regC^{(j)}$ for each $j$ in ${\{2, \dotsc, l\}}$,
      the qubit in $\regQ$,
      and all the qubits in $\regR_{j'}$ for each $j'$ in ${\{1, \dotsc, 2N\}}$
      are supposed to be initially in the totally mixed state~${I/2}$.\\
      Prepare a single-qubit register~$\regX_j$ for each $j$ in ${\{1, \dotsc, 2N\}}$,
      where the qubit in each $\regX_j$ is supposed to be initially in the totally mixed state~${I/2}$.
    \item
      For ${j = 1}$ to ${2N}$, perform the following:
      \begin{step}
      \item
        Apply $Q$ to ${(\regQ, \regR_j)}$.
      \item
        Apply the $\CNOT$ transformation to ${(\regQ, \regX_j)}$
        with the qubit in $\regQ$ being the control.
        Apply the controlled-$\INCR{2N}$ transformation~${\controlled \bigl( \INCR{2N} \bigr)}$
        to ${(\regQ, \regC)}$
        with the qubit in $\regQ$ being the control.
      \end{step}
    \item
      Measure the qubit in $\regC^{(1)}$ in the computational basis.
      Accept if this results in $\ket{0}$, and reject otherwise.
    \end{step}
  \end{algorithm*}
  \caption{The \textsc{One-Clean-Qubit Stability Checking Procedure}.}
  \label{Figure: One-Clean-Qubit Stability Checking Procedure}
\end{figure}

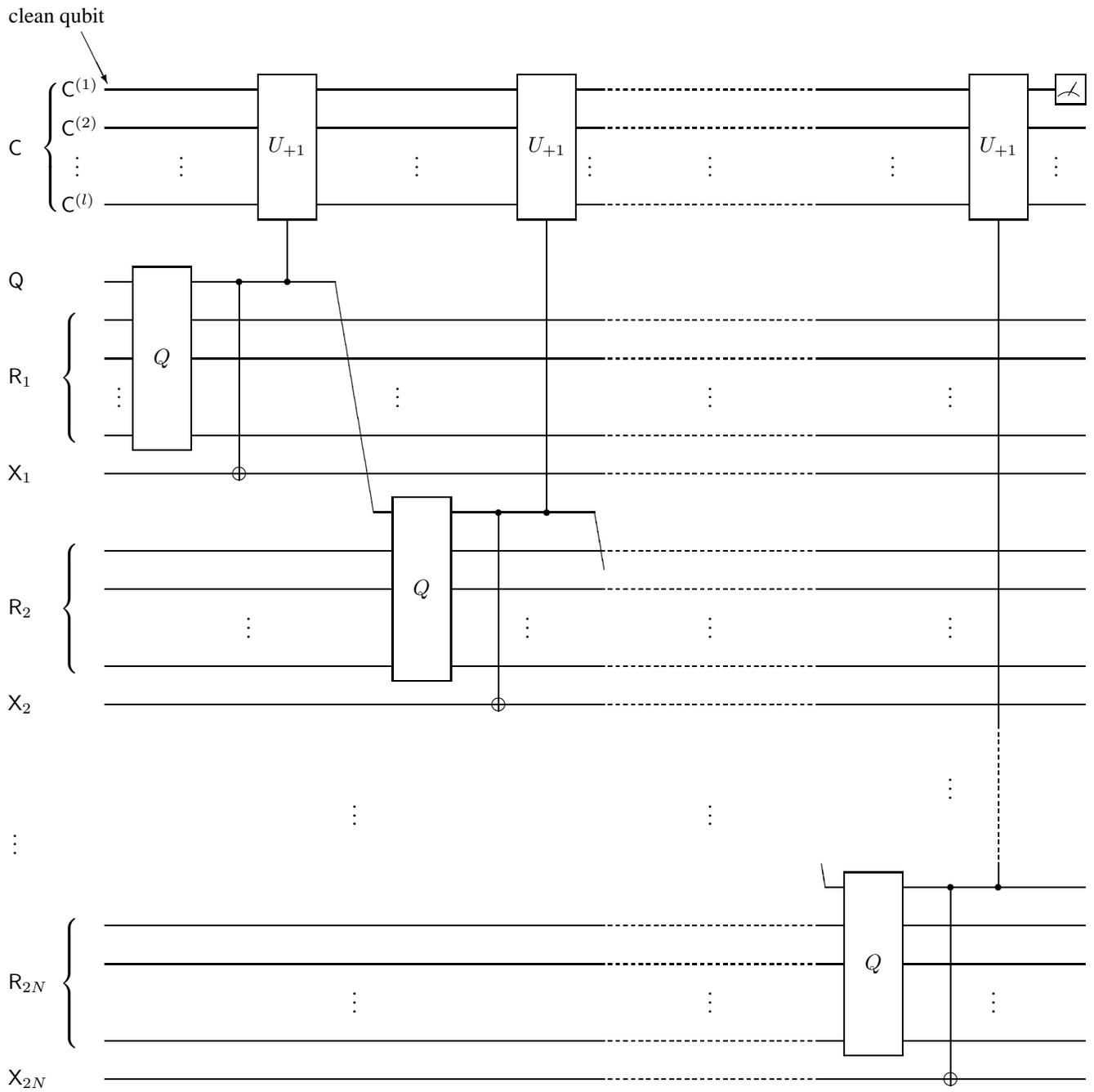
\begin{figure}[t!]
  \vspace{\baselineskip}
  \begin{center}
    \small
    \setlength{\unitlength}{0.89pt}
    \begin{picture}(560, 570)(  0,   0)
      \put(  0, 325)
        {%
          \begin{picture}(560, 245)(  0,   0)
            \put(  0, 230){\makebox( 50,  15){clean qubit}}
            \put( 38, 230){\vector(1,-2){13.5}}
            \put(  0, 130){\makebox( 15,  80){$\regC$ \hfill}}
            \put( 15, 135){\makebox( 15,  70){${\left\{ \rule{0mm}{36 \unitlength} \right.}$}}
            \put( 27.5, 191.25){\makebox( 25,  20){$\regC^{(1)}$ \hfill}}
            \put( 27.5, 171.25){\makebox( 25,  20){$\regC^{(2)}$ \hfill}}
            \put( 34.5, 143.25){\makebox( 25,  40){$\vdots$ \hfill}}
            \put( 27.5, 131.25){\makebox( 25,  20){$\regC^{(l)}$ \hfill}}
            \put(  0,  90){\makebox( 15,  20){$\regQ$ \hfill}}
            \put(  0,  10){\makebox( 15,  80){$\regR_1$ \hfill}}
            \put( 25,  15){\makebox( 15,  70){${\left\{ \rule{0mm}{36 \unitlength} \right.}$}}
            \put(  0, -10){\makebox( 15,  20){$\regX_1$ \hfill}}
            \put( 50, 200){\line(1,0){80}}
            \put( 50, 180){\line(1,0){80}}
            \put( 50, 143.25){\makebox( 80,  40){$\vdots$}}
            \put( 50, 140){\line(1,0){80}}
            \put( 50, 100){\line(1,0){15}}
            \put( 50,  80){\line(1,0){15}}
            \put( 50,  60){\line(1,0){15}}
            \put( 50,  23.25){\makebox( 15,  40){$\vdots$}}
            \put( 50,  20){\line(1,0){15}}
            \put( 50,   0){\line(1,0){66.5}}
            \put( 65,   0)
              {%
                \begin{picture}( 81.5, 187.5)(  0,   0)
                  \put(  0,  12.5){\framebox( 30,  95){$Q$}}
                  \put( 30, 100){\line(1,0){25}}
                  \put( 55, 100){\circle*{3}}
                  \put( 55, 100){\line(0,-1){96.5}}
                  {
                    \linethickness{0.265pt}
                    \put( 55,   0){\circle{7}}
                    \put( 51.5,   0){\line(1,0){7}}
                    \put( 55,   3.5){\line(0,-1){7}}
                  }
                  \put( 55, 100){\line(1,0){25}}
                  \put( 80, 100){\circle*{3}}
                  \put( 80, 100){\line(0,1){32.5}}
                  \put( 65, 132.5){\framebox( 30,  75){$U_{+1}$}}
                \end{picture}
              }
            \put( 95,  80){\line(1,0){215}}
            \put( 95,  60){\line(1,0){215}}
            \put( 95,  23.25){\makebox(215,  40){$\vdots$}}
            \put( 95,  20){\line(1,0){215}}
            \put(145, 100){\line(1,0){25}}
            \put(123.5,   0){\line(1,0){186.5}}
            \put(160, 200){\line(1,0){105}}
            \put(160, 180){\line(1,0){105}}
            \put(160, 143.25){\makebox(105,  40){$\vdots$}}
            \put(160, 140){\line(1,0){105}}
            {%
              \begin{xy}
                \put(170, 100){\ar@{-} <17.8pt, -106.8pt>}
              \end{xy}
            }
            \multiput(312,  80)(4,0){27}{\line(1,0){2}}
            \multiput(312,  60)(4,0){27}{\line(1,0){2}}
            \put(310,  23.25){\makebox(110,  40){$\vdots$}}
            \multiput(312,  20)(4,0){27}{\line(1,0){2}}
            \multiput(312,   0)(4,0){27}{\line(1,0){2}}
            \put(420,  80){\line(1,0){140}}
            \put(420,  60){\line(1,0){140}}
            \put(420,  23.25){\makebox(140,  40){$\vdots$}}
            \put(420,  20){\line(1,0){140}}
            \put(420,   0){\line(1,0){140}}
          \end{picture}
        }
      \put(  0, 205)
        {%
          \begin{picture}(560, 365)(  0,   0)
            \put(  0,  10){\makebox( 15,  80){$\regR_2$ \hfill}}
            \put( 25,  15){\makebox( 15,  70){${\left\{ \rule{0mm}{36 \unitlength} \right.}$}}
            \put(  0, -10){\makebox( 15,  20){$\regX_2$ \hfill}}
            \put( 50,  80){\line(1,0){150}}
            \put( 50,  60){\line(1,0){150}}
            \put( 50,  23.25){\makebox(150,  40){$\vdots$}}
            \put( 50,  20){\line(1,0){150}}
            \put( 50,   0){\line(1,0){201.5}}
            \put(190, 100){\line(1,0){10}}
            \put(200,   0)
              {%
                \begin{picture}( 81.5, 307.5)(  0,   0)
                  \put(  0,  12.5){\framebox( 30,  95){$Q$}}
                  \put( 30, 100){\line(1,0){25}}
                  \put( 55, 100){\circle*{3}}
                  \put( 55, 100){\line(0,-1){96.5}}
                  {
                    \linethickness{0.265pt}
                    \put( 55,   0){\circle{7}}
                    \put( 51.5,   0){\line(1,0){7}}
                    \put( 55,   3.5){\line(0,-1){7}}
                  }
                  \put( 55, 100){\line(1,0){25}}
                  \put( 80, 100){\circle*{3}}
                  \put( 80, 100){\line(0,1){152.5}}
                  \put( 65, 252.5){\framebox( 30,  75){$U_{+1}$}}
                \end{picture}
              }
            \put(230,  80){\line(1,0){80}}
            \put(230,  60){\line(1,0){80}}
            \put(230,  23.25){\makebox( 80,  40){$\vdots$}}
            \put(230,  20){\line(1,0){80}}
            \put(258.5,   0){\line(1,0){51.5}}
            \put(280, 100){\line(1,0){25}}
            \put(295, 320){\line(1,0){15}}
            \put(295, 300){\line(1,0){15}}
            \put(295, 263.25){\makebox( 15,  40){$\vdots$}}
            \put(295, 260){\line(1,0){15}}
            {%
              \begin{xy}
                \put(305, 100){\ar@{-} <4.45pt, -26.7pt>}
                \multiput(310.33,  68.02)(0.66,-3.96){23}{\ar@{-} <0.2937pt, -1.7622pt>}
              \end{xy}
            }
            \multiput(312,  80)(4,0){27}{\line(1,0){2}}
            \multiput(312,  60)(4,0){27}{\line(1,0){2}}
            \put(310,  23.25){\makebox(110,  40){$\vdots$}}
            \multiput(312,  20)(4,0){27}{\line(1,0){2}}
            \multiput(312,   0)(4,0){27}{\line(1,0){2}}
            \put(420,  80){\line(1,0){140}}
            \put(420,  60){\line(1,0){140}}
            \put(420,  23.25){\makebox(140,  40){$\vdots$}}
            \put(420,  20){\line(1,0){140}}
            \put(420,   0){\line(1,0){140}}
            \multiput(312, 320)(4,0){27}{\line(1,0){2}}
            \multiput(312, 300)(4,0){27}{\line(1,0){2}}
            \put(310, 263.25){\makebox(110,  40){$\vdots$}}
            \multiput(312, 260)(4,0){27}{\line(1,0){2}}
            \put(420, 320){\line(1,0){80}}
            \put(420, 300){\line(1,0){80}}
            \put(420, 263.25){\makebox(80,  40){$\vdots$}}
            \put(420, 260){\line(1,0){80}}
          \end{picture}
        }
      \put(  2,  63.25){\makebox( 15, 145){$\vdots$ \hfill}}
      \put( 50,  93.25){\makebox(260, 115){$\vdots$}}
      \put(310,  93.25){\makebox(110, 115){$\vdots$}}
      \put(420, 120.75){\makebox(140,  87.5){$\vdots$}}
      \put(  0,  10)
        {%
          \begin{picture}(560, 560)(  0,   0)
            \put(  0,  10){\makebox( 25,  80){$\regR_{2N}$ \hfill}}
            \put( 25,  15){\makebox( 15,  70){${\left\{ \rule{0mm}{36 \unitlength} \right.}$}}
            \put(  0, -10){\makebox( 25,  20){$\regX_{2N}$ \hfill}}
            \put( 50,  80){\line(1,0){260}}
            \put( 50,  60){\line(1,0){260}}
            \put( 50,  23.25){\makebox(260,  40){$\vdots$}}
            \put( 50,  20){\line(1,0){260}}
            \put( 50,   0){\line(1,0){260}}
            \multiput(312,  80)(4,0){27}{\line(1,0){2}}
            \multiput(312,  60)(4,0){27}{\line(1,0){2}}
            \put(310,  23.25){\makebox(110,  40){$\vdots$}}
            \multiput(312,  20)(4,0){27}{\line(1,0){2}}
            \multiput(312,   0)(4,0){27}{\line(1,0){2}}
            {%
              \begin{xy}
                \multiput(421.02, 123.88)(0.66,-3.96){3}{\ar@{-} <0.2937pt, -1.7622pt>}
                \put(423, 112){\ar@{-} <1.78pt, -10.68pt>}
              \end{xy}
            }
            \put(420,  80){\line(1,0){15}}
            \put(420,  60){\line(1,0){15}}
            \put(420,  20){\line(1,0){15}}
            \put(420,   0){\line(1,0){66.5}}
            \put(425, 100){\line(1,0){10}}
            \put(435,   0)
              {%
                \begin{picture}( 81.5, 307.5)(  0,   0)
                  \put(  0,  12.5){\framebox( 30,  95){$Q$}}
                  \put( 30, 100){\line(1,0){25}}
                  \put( 55, 100){\circle*{3}}
                  \put( 55, 100){\line(0,-1){96.5}}
                  {
                    \linethickness{0.265pt}
                    \put( 55,   0){\circle{7}}
                    \put( 51.5,   0){\line(1,0){7}}
                    \put( 55,   3.5){\line(0,-1){7}}
                  }
                  \put( 55, 100){\line(1,0){25}}
                  \put( 80, 100){\circle*{3}}
                  \put( 80, 100){\line(0,1){12.5}}
                  \multiput(80, 114.5)(0,4){17}{\line(0,1){2}}
                  \put( 80, 182.5){\line(0,1){265}}
                  \put( 65, 447.5){\framebox( 30,  75){$U_{+1}$}}
                \end{picture}
              }
            \put(465,  80){\line(1,0){95}}
            \put(465,  60){\line(1,0){95}}
            \put(465,  23.25){\makebox( 95,  40){$\vdots$}}
            \put(465,  20){\line(1,0){95}}
            \put(493.5,   0){\line(1,0){66.5}}
            \put(515, 100){\line(1,0){45}}
            \put(530, 515){\line(1,0){15}}
            \put(530, 495){\line(1,0){30}}
            \put(530, 458.25){\makebox( 30,  40){$\vdots$}}
            \put(530, 455){\line(1,0){30}}
            \put(545, 507.5){\framebox( 15,  15){\smallmetermark}}
          \end{picture}
        }
    \end{picture}
  \end{center}
  \caption{
    The quantum circuit for the \textsc{One-Clean-Qubit Stability Checking Procedure}
    induced by a quantum circuit~$Q$ and a positive integer~$N$,
    where $U_{+1}$ is a shorthand of the increment transformation~$\INCR{2N}$.
    The qubit in $\regC^{(1)}$ is supposed to be the only clean qubit at the beginning of the computation,
    and is also the output qubit of the computation.
    \vspace{5mm}
  }
  \label{Figure: quantum circuit for One-Clean-Qubit Stability Checking Procedure}
\end{figure}

First analyze the lower bound of the acceptance probability
of the one-clean-qubit computation induced by the quantum circuit
resulting from the \textsc{One-Clean-Qubit Stability Checking Procedure}.

\begin{proposition}
  For any quantum circuit~$Q$ and any positive integer~$N$ that is a power of two,
  let $R_1^{(N)}$ be the quantum circuit
  corresponding to the \textsc{One-Clean-Qubit Stability Checking Procedure} induced by $Q$ and~$N$.
  For the acceptance probability~${p_\acc(Q, 1)}$
  of the one-clean-qubit computation induced by $Q$
  and the acceptance probability~${p_\acc \bigl( R_1^{(N)}, 1 \bigr)}$
  of the one-clean-qubit computation induced by $R_1^{(N)}$,
  it holds that
  \[
    \op{p_\acc} \bigl( R_1^{(N)}, 1 \bigr)
    \geq
    (\op{p_\acc}(Q, 1))^{2N - 1}.
  \]
  \label{Proposition: lower bound of acceptance probability of One-Clean-Qubit Stability Checking Procedure}
\end{proposition}

\vspace{-\baselineskip} 
\begin{proof}
  As in Subsection~\ref{Subsection: Randomness Amplification Procedure},
  for each $j$ in $\Binary$,
  let $\Pi_j$ be the projection operator acting over $w$~qubits defined by
  ${\Pi_j = \ketbra{j} \tensor I^{\tensor (w - 1)}}$,
  and let $\rho_j$ be the quantum state of $w$~qubits defined by
  ${\rho_j = \ketbra{j} \tensor \bigl( \frac{I}{2} \bigr)^{\tensor (w - 1)}}$.
  The acceptance probability~${p_\acc(Q, 1)}$
  of the original one-clean-qubit computation induced by $Q$
  is given by
  \[
    \op{p_\acc}(Q, 1) = \tr \Pi_0 Q \rho_0 \conjugate{Q},
  \]
  and it holds that
  \[
    \tr \Pi_1 Q \rho_1 \conjugate{Q} = \op{p_\acc}(Q, 1).
  \]
  This implies that,
  for each repetition round during Step~2,
  the counter value stored in $\regC$ is increased by one with probability~${1 - p_\acc(Q, 1)}$
  if the content of $\regQ$ was $0$ when entering Step~2.1,
  while it is increased by one with probability~${p_\acc(Q, 1)}$
  if the content of $\regQ$ was $1$ when entering Step~2.1.

  Notice that, at the $j$th repetition round in Step~2 for ${j \geq 2}$,
  the content of $\regQ$ is $1$ when entering Step~2.1
  if and only if the previous repetition round has increased the counter value in $\regC$.
  Hence, taking into account that the content of $\regQ$ is initially $0$ or $1$ with equal probability,
  after all the repetition rounds of Step~2,
  the counter value in $\regC$ is never increased with probability
  \[
    \frac{1}{2} (\op{p_\acc}(Q, 1))^{2N}
    +
    \frac{1}{2} (1 - \op{p_\acc}(Q, 1)) \, (\op{p_\acc}(Q, 1))^{2N - 1}
    =
    \frac{1}{2} (\op{p_\acc}(Q, 1))^{2N - 1},
  \]
  while it is increased by ${2N}$ with probability
  \[
    \frac{1}{2} (1 - \op{p_\acc}(Q, 1)) \, (\op{p_\acc}(Q, 1))^{2N - 1}
    +
    \frac{1}{2} (\op{p_\acc}(Q, 1))^{2N}
    =
    \frac{1}{2} (\op{p_\acc}(Q, 1))^{2N - 1}.
  \]
  The content of $\regC$ stationarily remains its initial value~${r \leq N - 1}$ for the first case,
  and it comes back to the initial value~${r \leq N - 1}$ for the second case.
  As Step~3 results in acceptance at least in these two cases,
  it follows that
  \[
    \op{p_\acc} \bigl( R_1^{(N)}, 1 \bigr) 
    \geq
    \frac{1}{2} (\op{p_\acc}(Q, 1))^{2N - 1}
    +
    \frac{1}{2} (\op{p_\acc}(Q, 1))^{2N - 1}
    =
    (\op{p_\acc}(Q, 1))^{2N - 1},
  \]
  as claimed.
\end{proof}

Proposition~\ref{Proposition: lower bound of acceptance probability of One-Clean-Qubit Stability Checking Procedure}
in particular implies that,
if the original acceptance probability~${p_\acc(Q, 1)}$ is one,
the acceptance probability~${p_\acc \bigl( R_1^{(N)}, 1 \bigr)}$ is also one
for the circuit~$R_1^{(N)}$
corresponding to the \textsc{One-Clean-Qubit Stability Checking Procedure} induced by $Q$ and $N$.

The next proposition provides an upper bound of the acceptance probability
of the one-clean-qubit computation induced by the quantum circuit
resulting from the \textsc{One-Clean-Qubit Stability Checking Procedure},
assuming that the original acceptance probability~${p_\acc(Q, 1)}$ is close to ${1/2}$.

\begin{proposition}
  For any quantum circuit~$Q$ and any positive integer~$N$
  that is a power of two and at least ${2^6 = 64}$,
  let $R_1^{(N)}$ be the quantum circuit
  corresponding to the \textsc{One-Clean-Qubit Stability Checking Procedure} induced by $Q$ and $N$.
  If the acceptance probability~${p_\acc(Q, 1)}$
  of the one-clean-qubit computation induced by $Q$
  satisfies that
  ${\frac{1}{2} - \varepsilon \leq p_\acc(Q, 1) \leq \frac{1}{2} + \varepsilon}$
  for some $\varepsilon$ in ${\bigl[ 0, \frac{1}{8} \bigr]}$
  such that ${3 N^{- \frac{1}{3}} + 4 \varepsilon \leq 1}$,
  it holds for the acceptance probability~${p_\acc \bigl( R_1^{(N)}, 1 \bigr)}$
  of the one-clean-qubit computation induced by $R_1^{(N)}$
  that
  \[
    \op{p_\acc} \bigl( R_1^{(N)}, 1 \bigr)
    <
    3 N^{- \frac{1}{3}} + 4 \varepsilon.
  \]
  \label{Proposition: upper bound of acceptance probability of One-Clean-Qubit Stability Checking Procedure}
\end{proposition}

\begin{proof}
  As before,
  for each $j$ in $\Binary$,
  let $\Pi_j$ be the projection operator acting over $w$~qubits defined by
  ${\Pi_j = \ketbra{j} \tensor I^{\tensor (w - 1)}}$,
  and let $\rho_j$ be the quantum state of $w$~qubits defined by
  ${\rho_j = \ketbra{j} \tensor \bigl( \frac{I}{2} \bigr)^{\tensor (w - 1)}}$.
  The acceptance probability~${p_\acc(Q, 1)}$
  of the original one-clean-qubit computation induced by $Q$
  is given by
  \[
    \op{p_\acc}(Q, 1) = \tr \Pi_0 Q \rho_0 \conjugate{Q},
  \]
  and it holds that
  \[
    \tr \Pi_1 Q \rho_1 \conjugate{Q} = \op{p_\acc}(Q, 1).
  \]

  Notice that,
  for each repetition round during Step~2,
  the counter value in $\regC$ is increased by one with probability
  at least~${\min \bigl\{ p_\acc(Q, 1), 1 - p_\acc(Q, 1) \bigr\} \geq \frac{1}{2} - \varepsilon}$
  and at most~${\max \bigl\{ p_\acc(Q, 1), 1 - p_\acc(Q, 1) \bigr\} \leq \frac{1}{2} + \varepsilon}$
  regardless of the content of $\regQ$ being $0$ or $1$ when entering Step~2.1.
  Hence, from the Hoeffding bound (Lemma~\ref{Lemma: Hoeffding bounds}),
  the probability that the total increment of the counter value in $\regC$ is at most
  ${(1 - 2 \varepsilon - 2 \delta) N}$
  after all the ${2N}$~repetition rounds of Step~2
  is less than $e^{- 4 \delta^2 N}$,
  for any $\delta$ in ${\bigl[ 0, \frac{1}{2} - \varepsilon \bigr]}$.
  Similarly, the probability that the total increment of the counter value in $\regC$ is at least
  ${(1 + 2 \varepsilon + 2 \delta) N}$
  after all the ${2N}$~repetition rounds of Step~2
  is less than~$e^{- 4 \delta^2 N}$ also,
  for any $\delta$ in ${\bigl[ 0, \frac{1}{2} - \varepsilon \bigr]}$.
  It follows that,
  when $\delta$ is in ${\bigl( \frac{1}{2 \sqrt{N}}, \frac{1}{4} - \varepsilon + \frac{1}{4N} \bigr]}$
  and the initial counter value~$r$ satisfies that
  \[
    2 (\varepsilon + \delta) N - 1 \leq r \leq N - 2 (\varepsilon + \delta) N,
  \]
  after all the ${2N}$~repetition rounds of Step~2,
  the probability that the counter value in $\regC$ is in the interval~${[N, 2N - 1]}$
  is more than~${1 - 2 e^{- 4 \delta^2 N} > 1 - 2^{- 4 \delta^2 N + 1}}$,
  and thus, the acceptance probability at Step~3 is less than
  \[
    \Bigl[ 4 (\varepsilon + \delta) - \frac{2}{N} \Bigr] \cdot 1
    +
    \Bigl[ 1 - 4 (\varepsilon + \delta) + \frac{2}{N} \Bigr] \cdot 2^{- 4 \delta^2 N + 1}
    <
    4 (\varepsilon + \delta)
    +
    2^{- 4 \delta^2 N + 1}.
  \]
  By taking $\delta$ to be ${\frac{1}{2} N^{- \frac{1}{3}}}$
  (which is in ${\bigl( \frac{1}{2 \sqrt{N}}, \frac{1}{4} - \varepsilon + \frac{1}{4N} \bigr]}$
  as $\varepsilon$ is at most ${1/8}$ and $N$ is at least ${2^6 = 64}$),
  and using the fact that ${x \leq 2^{x - 1}}$ holds for any ${x \geq 2}$,
  it follows that
  \[
    \op{p_\acc} \bigl( R_1^{(N)}, 1 \bigr)
    <
    2 N^{- \frac{1}{3}} + 2^{- N^{\frac{1}{3}} + 1} + 4 \varepsilon
    \leq
    3 N^{- \frac{1}{3}} + 4 \varepsilon,
  \]
  which completes the proof.
\end{proof}


\subsubsection{Stability Checking Using Two Clean Qubits}
\label{Subsubsection: Two-Clean-Qubit Stability Checking Procedure}

As can be seen in the proof of Proposition~\ref{Proposition: upper bound of acceptance probability of One-Clean-Qubit Stability Checking Procedure},
one drawback of the \textsc{One-Clean-Qubit Stability Checking Procedure}
when analyzing the upper bound of its acceptance probability
is that there are some ``bad'' initial counter values
with which the procedure is forced to accept with unallowably high probability
even if the acceptance probability~${p_\acc(Q, 1)}$ of the underlying circuit~$Q$
was sufficiently close to ${1/2}$.
This is the essential reason why only a polynomially small upper bound can be proved
on the acceptance probability of the \textsc{One-Clean-Qubit Stability Checking Procedure}
in Proposition~\ref{Proposition: upper bound of acceptance probability of One-Clean-Qubit Stability Checking Procedure}.
As there are only polynomially many number of possible initial counter values,
even just one ``bad'' initial value is unacceptable
to go beyond polynomially small upper bounds.
The authors do not know how to get rid of this barrier with the use of just one clean qubit.

In contrast, if two clean qubits are available,
one can easily modify the procedure so that it has no ``bad'' initial counter value.
The idea is very simple.
This time, the register~$\regC$ uses ${\log N + 3}$~qubits
so that the counter takes values in ${\Integers_{8N} = \{0, \dotsc, 8N - 1\}}$.
The first two qubits of $\regC$ are supposed to be initially in state~$\ket{0}$,
which implies that the initial counter value is picked uniformly from the set~${\{0, \dotsc, 2N - 1\}}$.
The point is that one can increase the counter value by $N$ before the repetition starts
so that the actual initial value of the counter is in the set~${\{N, \dotsc, 3N - 1\}}$.
If the acceptance probability~${p_\acc(Q, 1)}$ of the underlying circuit~$Q$
was in the interval~${\bigl[ \frac{1}{2} - \varepsilon, \frac{1}{2} + \varepsilon \bigr]}$
for some sufficiently small $\varepsilon$,
the expected value of the counter after the repetition
must be in the interval~${[ (5 - 8 \varepsilon) N, (7 + 8 \varepsilon) N - 1 ]}$.
Hence, if $\varepsilon$ is in ${\bigl[ 0, \frac{1}{16} \bigr]}$, for instance,
the probability is exponentially small for the event
that the final counter value is not in the interval~${[4N, 8N - 1]}$.
This leads to the \textsc{Two-Clean-Qubit Stability Checking Procedure},
whose construction is summarized in Figure~\ref{Figure: Two-Clean-Qubit Stability Checking Procedure}.

\begin{figure}[t!]
  \begin{algorithm*}{\textsc{Two-Clean-Qubit Stability Checking Procedure}}
    \begin{step}
    \item
      Given a positive integer~$N$ that is a power of two,
      let ${l = \log N + 3}$.
      Prepare an $l$-qubit register~$\regC$,
      a single-qubit register~$\regQ$,
      and each ${(w - 1)}$-qubit register~$\regR_j$ for each $j$ in ${\{1, \dotsc, 8N\}}$.
      For each $j$ in ${\{1, \dotsc, l\}}$,
      let $\regC^{(j)}$ denote the single-qubit quantum register
      corresponding to the $j$th qubit of $\regC$.
      The qubits in $\regC^{(1)}$ and $\regC^{(2)}$ are supposed to be initially in state~$\ket{0}$,
      while the qubit in $\regC^{(j)}$ for each $j$ in ${\{3, \dotsc, l\}}$,
      the qubit in $\regQ$,
      and all the qubits in $\regR_{j'}$ for each $j'$ in ${\{1, \dotsc, 8N\}}$
      are supposed to be initially in the totally mixed state~${I/2}$.\\
      Prepare a single-qubit register~$\regX_j$ for each $j$ in ${\{1, \dotsc, 8N\}}$,
      where the qubit in each $\regX_j$ is supposed to be initially in the totally mixed state~${I/2}$.
    \item
      For ${j = 1}$ to $N$,
      apply the unitary transformation~$\INCR{8N}$ to $\regC$.
    \item
      For ${j = 1}$ to ${8N}$, perform the following:
      \begin{step}
      \item
        Apply $Q$ to ${(\regQ, \regR_j)}$.
      \item
        Apply the $\CNOT$ transformation to ${(\regQ, \regX_j)}$
        with the qubit in $\regQ$ being the control.
        Apply the controlled-$\INCR{8N}$ transformation~${\controlled \bigl( \INCR{8N} \bigr)}$
        to ${(\regQ, \regC)}$
        with the qubit in $\regQ$ being the control.
      \end{step}
    \item
      Measure the qubit in $\regC^{(1)}$ in the computational basis.
      Accept if this results in $\ket{0}$, and reject otherwise.
    \end{step}
  \end{algorithm*}
  \caption{The \textsc{Two-Clean-Qubit Stability Checking Procedure}.}
  \label{Figure: Two-Clean-Qubit Stability Checking Procedure}
\end{figure}

First, the following lower bound holds for the acceptance probability of the two-clean-qubit computation
induced by the quantum circuit resulting from the \textsc{Two-Clean-Qubit Stability Checking Procedure}.
In particular,
if the original acceptance probability~${p_\acc(Q, 1)}$ is one,
the acceptance probability~${p_\acc \bigl( R_2^{(N)}, 2 \bigr)}$ is also one
for the circuit~$R_2^{(N)}$
corresponding to the \textsc{Two-Clean-Qubit Stability Checking Procedure} induced by $Q$ and $N$.

\begin{proposition}
  For any quantum circuit~$Q$ and any positive integer~$N$ that is a power of two,
  let $R_2^{(N)}$ be the quantum circuit
  corresponding to the \textsc{Two-Clean-Qubit Stability Checking Procedure} induced by $Q$ and $N$.
  For the acceptance probability~${p_\acc(Q, 1)}$
  of the one-clean-qubit computation induced by $Q$
  and the acceptance probability~${p_\acc \bigl( R_2^{(N)}, 2 \bigr)}$
  of the two-clean-qubit computation induced by $R_2^{(N)}$,
  it holds that
  \[
    \op{p_\acc} \bigl( R_2^{(N)}, 2 \bigr)
    \geq
    (\op{p_\acc}(Q, 1))^{8N - 1}.
  \]
  \label{Proposition: lower bound of acceptance probability of Two-Clean-Qubit Stability Checking Procedure}
\end{proposition}

The proof of Proposition~\ref{Proposition: lower bound of acceptance probability of Two-Clean-Qubit Stability Checking Procedure}
is essentially the same as that of Proposition~\ref{Proposition: lower bound of acceptance probability of One-Clean-Qubit Stability Checking Procedure},
and is omitted.

The next proposition provides an upper bound of the acceptance probability
of the two-clean-qubit computation induced by the quantum circuit
resulting from the \textsc{Two-Clean-Qubit Stability Checking Procedure},
assuming that the original acceptance probability~${p_\acc(Q, 1)}$ is close to ${1/2}$.

\begin{proposition}
  For any quantum circuit~$Q$ and any positive integer~$N$
  that is a power of two and at least ${2^4 = 16}$,
  let $R_2^{(N)}$ be the quantum circuit
  corresponding to the \textsc{Two-Clean-Qubit Stability Checking Procedure} induced by $Q$ and $N$.
  If the acceptance probability~${p_\acc(Q, 1)}$
  of the one-clean-qubit computation induced by $Q$
  satisfies that
  ${\frac{1}{2} - \varepsilon \leq p_\acc(Q, 1) \leq \frac{1}{2} + \varepsilon}$
  for some $\varepsilon$ in ${\bigl[ 0, \frac{1}{16} \bigr]}$,
  it holds for the acceptance probability~${p_\acc \bigl( R_2^{(N)}, 2 \bigr)}$
  of the two-clean-qubit computation induced by $R_2^{(N)}$
  that
  \[
    \op{p_\acc} \bigl( R_2^{(N)}, 2 \bigr)
    <
    2^{- \frac{N}{16} + 1}.
  \]
  \label{Proposition: upper bound of acceptance probability of Two-Clean-Qubit Stability Checking Procedure}
\end{proposition}

\begin{proof}
  The proof is similar to that of Proposition~\ref{Proposition: upper bound of acceptance probability of One-Clean-Qubit Stability Checking Procedure}.
  As before,
  for each $j$ in $\Binary$,
  let $\Pi_j$ be the projection operator acting over $w$~qubits defined by
  ${\Pi_j = \ketbra{j} \tensor I^{\tensor (w - 1)}}$,
  and let $\rho_j$ be the quantum state of $w$~qubits defined by
  ${\rho_j = \ketbra{j} \tensor \bigl( \frac{I}{2} \bigr)^{\tensor (w - 1)}}$.
  The acceptance probability~${p_\acc(Q, 1)}$
  of the original one-clean-qubit computation induced by $Q$
  is given by
  \[
    \op{p_\acc}(Q, 1) = \tr \Pi_0 Q \rho_0 \conjugate{Q},
  \]
  and it holds that
  \[
    \tr \Pi_1 Q \rho_1 \conjugate{Q} = \op{p_\acc}(Q, 1).
  \]

  Notice that,
  for each repetition round during Step~3,
  the counter value in $\regC$ is increased by one with probability
  at least~${\min \bigl\{ p_\acc(Q, 1), 1 - p_\acc(Q, 1) \bigr\} \geq \frac{1}{2} - \varepsilon}$
  and at most~${\max \bigl\{ p_\acc(Q, 1), 1 - p_\acc(Q, 1) \bigr\} \leq \frac{1}{2} + \varepsilon}$
  regardless of the content of $\regQ$ being $0$ or $1$ when entering Step~3.1.
  Hence, from the Hoeffding bound (Lemma~\ref{Lemma: Hoeffding bounds}),
  the probability that the total increment of the counter value in $\regC$ is at most
  ${(4 - 8 \varepsilon - 8 \delta) N}$
  after all the ${8N}$~repetition rounds of Step~3
  is less than $e^{- 16 \delta^2 N}$,
  for any $\delta$ in ${\bigl[ 0, \frac{1}{2} - \varepsilon \bigr]}$.
  Similarly, the probability that the total increment of the counter value in $\regC$ is at least
  ${(4 + 8 \varepsilon + 8 \delta) N}$
  after all the ${8N}$~repetition rounds of Step~3
  is less than~$e^{- 16 \delta^2 N}$ also,
  for any $\delta$ in ${\bigl[ 0, \frac{1}{2} - \varepsilon \bigr]}$.
  As the counter value~$r$ in $\regC$ at the beginning of Step~3 satisfies that
  \[
    N \leq r \leq 3N - 1,
  \]
  when $\delta$ is in ${\bigl( \frac{1}{4 \sqrt{N}}, \frac{1}{8} - \varepsilon + \frac{1}{8N} \bigr]}$,
  it holds that,
  after all the ${8N}$~repetition rounds of Step~3,
  the probability that the counter value in $\regC$ is in the interval~${[4N, 8N - 1]}$
  is more than~${1 - 2 e^{- 16 \delta^2 N} > 1 - 2^{- 16 \delta^2 N + 1}}$,
  and thus, the acceptance probability at Step~4 is less than
  \[
    1 - \bigl( 1 - 2^{- 16 \delta^2 N + 1} \bigr)
    =
    2^{- 16 \delta^2 N + 1}.
  \]
  By taking $\delta$ to be $\frac{1}{16}$,
  it follows that
  \[
    \op{p_\acc} \bigl( R_2^{(N)}, 2 \bigr)
    <
    2^{- \frac{N}{16} + 1},
  \]
  which completes the proof.
\end{proof}


\section{Error Reduction for One-Sided-Error Cases}
\label{Section: error reduction for one-sided-error cases}

This section proves the error reduction results in the cases of one-sided bounded error.
First, Subsection~\ref{Subsection: error reduction for perfect-completeness cases}
treats the cases with perfect completeness,
namely,
Theorems~\ref{Theorem: error reduction for perfect-completeness case}~and~\ref{Theorem: co-RQlogP is in co-RQ1P}.
The cases with perfect soundness (Corollary~\ref{Corollary: error reduction for perfect-soundness cases})
then easily follow
from Theorems~\ref{Theorem: error reduction for perfect-completeness case}~and~\ref{Theorem: co-RQlogP is in co-RQ1P},
as will be found in Subsection~\ref{Subsection: error reduction for perfect-soundness cases}.


\subsection{Cases with Perfect Completeness}
\label{Subsection: error reduction for perfect-completeness cases}

\subsubsection{One-Clean-Qubit Case}

This subsection proves Theorem~\ref{Theorem: co-RQlogP is in co-RQ1P},
stating that any problem computable using logarithmically many clean qubits
with one-sided bounded error of perfect completeness
and soundness bounded away from one by an inverse-polynomial
is necessarily computable using only one clean qubit
with perfect completeness and polynomially small soundness error.

\begin{proof}[Proof of Theorem~\ref{Theorem: co-RQlogP is in co-RQ1P}]
  For any polynomial-time computable function~$\function{s}{\Nonnegative}{[0,1]}$
  satisfying ${1 - s \geq \frac{1}{q}}$ for some polynomially bounded function~$\function{q}{\Nonnegative}{\Natural}$,
  let ${A = (A_\yes, A_\no)}$ be a problem in ${\QlogP(1, s)}$.
  Then $A$ has a polynomial-time uniformly generated family~${\{Q_x\}_{x \in \Sigma^\ast}}$
  of quantum circuits
  such that, for every input~$x$,
  ${\op{p_\acc}(Q_x, \op{k}(\abs{x})) = 1}$ if $x$ is in $A_\yes$
  and
  ${\op{p_\acc}(Q_x, \op{k}(\abs{x})) \leq \op{s}(\abs{x})}$ if $x$ is in $A_\no$
  for some logarithmically bounded function~$\function{k}{\Nonnegative}{\Natural}$.
  For any polynomially bounded function~$\function{p}{\Nonnegative}{\Natural}$,
  the proof constructs a polynomial-time uniformly generated family of quantum circuits
  that puts $A$ in ${\QoneP \bigl( 1, \frac{1}{p} \bigr)}$.

  Fix an input~$x$.

  From the circuit~$Q_x$,
  one first constructs a quantum circuit~$R_x$
  according to the \textsc{One-Clean-Qubit Simulation Procedure}.
  By Proposition~\ref{Proposition: acceptance probability of One-Clean-Qubit Simulation Procedure},
  it holds that
  ${\op{p_\acc}(R_x, 1) = 1}$ if $x$ is in $A_\yes$,
  and
  \[
    1 - 2^{-\op{k}(\abs{x})}
    \leq
    \op{p_\acc}(R_x, 1)
    \leq
    1 - 2^{-\op{k}(\abs{x})} (1 - \op{s}(\abs{x}))
    \leq
    1 - \frac{1}{\op{r}_1(\abs{x})}
  \]
  if $x$ is in $A_\no$,
  where $\function{r_1}{\Nonnegative}{\Natural}$ is the polynomially bounded function
  defined by ${r_1 = 2^k q}$.

  From the circuit~$R_x$,
  one constructs a quantum circuit~$R'_x$
  according to the \textsc{Randomness Amplification Procedure}
  with the integer~${N = \op{r}_1(\abs{x}) \op{r}_2(\abs{x})}$
  for a polynomially bounded function~$\function{r_2}{\Nonnegative}{\Natural}$
  satisfying ${r_2 \geq p + 1}$.
  By Proposition~\ref{Proposition: acceptance probability of Randomness Amplification Procedure},
  it holds that
  ${\op{p_\acc}(R'_x, 1) = 1}$ if $x$ is in $A_\yes$,
  and
  \[
    \frac{1}{2}
    \leq
    \op{p_\acc}(R'_x, 1)
    \leq
    \frac{1}{2}
    +
    \frac{1}{2} \biggl( 1 - \frac{2}{\op{r}_1(\abs{x})} \biggr)^{\op{r}_1(\abs{x}) \op{r}_2(\abs{x})}
    <
    \frac{1}{2}
    +
    \frac{1}{2} e^{- 2 \op{r}_2(\abs{x})}
    <
    \frac{1}{2} + 2^{- 2 \op{r}_2(\abs{x})}
  \]
  if $x$ is in $A_\no$.
 
  Finally, from the circuit~$R'_x$,
  one constructs a quantum circuit~$R''_x$
  according to the \textsc{One-Clean-Qubit Stability Checking Procedure}
  with $N$ to be the smallest integer that is a power of two and at least ${(6 \op{p}(\abs{x}))^3}$
  (i.e., ${N = 2^{\ceilL{3 \log (6 \op{p}(\abs{x}))}} > 2^6 = 64}$).
  By Propositions~\ref{Proposition: lower bound of acceptance probability of One-Clean-Qubit Stability Checking Procedure}~and~\ref{Proposition: upper bound of acceptance probability of One-Clean-Qubit Stability Checking Procedure},
  it holds that
  ${\op{p_\acc}(R''_x, 1) = 1}$ if $x$ is in $A_\yes$,
  and
  \[
    \op{p_\acc}(R''_x, 1)
    \leq
    \frac{1}{2 \op{p}(\abs{x})}
    +
    4 \cdot 2^{- 2 \op{r}_2(\abs{x})}
    <
    \frac{1}{\op{p}(\abs{x})}
  \]
  if $x$ is in $A_\no$,
  where the last inequality follows from
  the facts that ${r_2 \geq p + 1}$
  and that the inequality~${2^n > n}$ holds.
  
  The claim follows
  with the polynomial-time uniformly generated family~${\{R''_x\}_{x \in \Sigma^\ast}}$ of quantum circuits.
\end{proof}

\subsubsection{Two-Clean-Qubit Case}

This subsection proves Theorem~\ref{Theorem: error reduction for perfect-completeness case},
stating that, using two clean qubits rather than one,
soundness can be made exponentially small.

\begin{proof}[Proof of Theorem~\ref{Theorem: error reduction for perfect-completeness case}]
  The proof is essentially the same as that of Theorem~\ref{Theorem: co-RQlogP is in co-RQ1P},
  except that the circuit~$R''_x$ is constructed
  according to the \textsc{Two-Clean-Qubit Stability Checking Procedure}
  rather than the \textsc{One-Clean-Qubit Stability Checking Procedure}.

  For any polynomial-time computable function~$\function{s}{\Nonnegative}{[0,1]}$
  satisfying ${1 - s \geq \frac{1}{q}}$ for some polynomially bounded function~$\function{q}{\Nonnegative}{\Natural}$,
  let ${A = (A_\yes, A_\no)}$ be a problem in ${\QlogP(1, s)}$.
  Then $A$ has a polynomial-time uniformly generated family~${\{Q_x\}_{x \in \Sigma^\ast}}$
  of quantum circuits
  such that, for every input~$x$,
  ${\op{p_\acc}(Q_x, \op{k}(\abs{x})) = 1}$ if $x$ is in $A_\yes$
  and
  ${\op{p_\acc}(Q_x, \op{k}(\abs{x})) \leq \op{s}(\abs{x})}$ if $x$ is in $A_\no$
  for some logarithmically bounded function~$\function{k}{\Nonnegative}{\Natural}$.
  For any polynomially bounded function~$\function{p}{\Nonnegative}{\Natural}$,
  the proof constructs a polynomial-time uniformly generated family of quantum circuits
  that puts $A$ in ${\QtwoP(1, 2^{-p})}$.

  Fix an input~$x$.

  From the circuit~$Q_x$,
  one first constructs a quantum circuit~$R_x$
  according to the \textsc{One-Clean-Qubit Simulation Procedure}.
  By Proposition~\ref{Proposition: acceptance probability of One-Clean-Qubit Simulation Procedure},
  it holds that
  ${\op{p_\acc}(R_x, 1) = 1}$ if $x$ is in $A_\yes$,
  and
  \[
    1 - 2^{-\op{k}(\abs{x})}
    \leq
    \op{p_\acc}(R_x, 1)
    \leq
    1 - 2^{-\op{k}(\abs{x})} (1 - \op{s}(\abs{x}))
    \leq
    1 - \frac{1}{\op{r}(\abs{x})}
  \]
  if $x$ is in $A_\no$,
  where $\function{r}{\Nonnegative}{\Natural}$ is the polynomially bounded function
  defined by ${r = 2^k q}$.

  From the circuit~$R_x$,
  one constructs a quantum circuit~$R'_x$
  according to the \textsc{Randomness Amplification Procedure}
  with the integer~${N \geq \alpha \op{r}(\abs{x})}$, where ${\alpha = \frac{3}{2} \ln 2 < 1.04}$.
  By Proposition~\ref{Proposition: acceptance probability of Randomness Amplification Procedure},
  it holds that
  ${\op{p_\acc}(R'_x, 1) = 1}$ if $x$ is in $A_\yes$,
  and
  \[
    \frac{1}{2}
    \leq
    \op{p_\acc}(R'_x, 1)
    \leq
    \frac{1}{2}
    +
    \frac{1}{2} \biggl( 1 - \frac{2}{\op{r}(\abs{x})} \biggr)^{\alpha \op{r}(\abs{x})}
    <
    \frac{1}{2}
    +
    \frac{1}{2} e^{- 3 \ln 2}
    =
    \frac{1}{2} + \frac{1}{16}
  \]
  if $x$ is in $A_\no$.
 
  Finally, from the circuit~$R'_x$,
  one constructs a quantum circuit~$R''_x$
  according to the \textsc{Two-Clean-Qubit Stability Checking Procedure}
  with $N$ to be the smallest integer that is a power of two and at least ${16 (\op{p}(\abs{x}) + 1)}$
  (i.e., ${N = 2^{\ceilL{\log (\op{p}(\abs{x}) + 1)} + 4}}$).
  If $x$ is in $A_\yes$,
  Proposition~\ref{Proposition: lower bound of acceptance probability of Two-Clean-Qubit Stability Checking Procedure}
  ensures that
  ${\op{p_\acc}(R''_x, 2) = 1}$.
  On the other hand,
  if $x$ is in $A_\no$,
  Proposition~\ref{Proposition: upper bound of acceptance probability of Two-Clean-Qubit Stability Checking Procedure}
  ensures that
  \[
    \op{p_\acc}(R''_x, 2)
    \leq
    2^{- \op{p}(\abs{x})}
  \]
  (notice that the bounds for the probability~${\op{p_\acc}(R'_x, 1)}$ above
  are such that
  Proposition~\ref{Proposition: upper bound of acceptance probability of Two-Clean-Qubit Stability Checking Procedure}
  can be used to show the upper bound for the probability~${\op{p_\acc}(R''_x, 2)}$).

  The claim follows
  with the polynomial-time uniformly generated family~${\{R''_x\}_{x \in \Sigma^\ast}}$ of quantum circuits.
\end{proof}


\subsection{Cases with Perfect Soundness}
\label{Subsection: error reduction for perfect-soundness cases}

Now Corollary~\ref{Corollary: error reduction for perfect-soundness cases}
is immediate from Theorems~\ref{Theorem: error reduction for perfect-completeness case}~and~\ref{Theorem: co-RQlogP is in co-RQ1P}
by considering complement problems.

\begin{proof}[Proof of Corollary~\ref{Corollary: error reduction for perfect-soundness cases}]
  For any polynomial-time computable function~$\function{c}{\Nonnegative}{[0,1]}$
  satisfying ${c \geq \frac{1}{q}}$ for some polynomially bounded function~$\function{q}{\Nonnegative}{\Natural}$,
  let ${A = (A_\yes, A_\no)}$ be a problem in ${\QlogP(c, 0)}$.
  Consider the complement problem~${\overline{A} = (A_\no, A_\yes)}$ of $A$.
  As $\overline{A}$ is in ${\QlogP(1, \delta)}$
  for the constant~${\delta = 1 - c}$
  satisfying ${1 - \delta  = c \geq \frac{1}{q}}$,
  it follows from Theorems~\ref{Theorem: error reduction for perfect-completeness case}~and~\ref{Theorem: co-RQlogP is in co-RQ1P}
  that $\overline{A}$ is in ${\QtwoP(1, 2^{-p})}$
  and also in ${\QoneP \bigl( 1, \frac{1}{p} \bigr)}$,
  for any polynomially bounded function~$\function{p}{\Nonnegative}{\Natural}$.
  This implies that $A$ is in ${\QtwoP(1 - 2^{-p}, 0)}$
  and also in ${\QoneP \bigl( 1 - \frac{1}{p}, 0 \bigr)}$,
  for any polynomially bounded function~$\function{p}{\Nonnegative}{\Natural}$,
  as desired.
\end{proof}


\section{Error Reduction for Two-Sided-Error Cases}
\label{Section: error reduction for two-sided-error cases}

This section considers the cases with two-sided bounded error.
First, Subsection~\ref{Subsection: OR-type repetition procedure}
provides one more technical tool,
called the \textsc{OR-Type Repetition Procedure},
which aims to simulate the standard error-reduction method of a repetition with an OR-type decision.
By combining this procedure with the three procedures presented in Section~\ref{Section: building blocks},
Theorems~\ref{Theorem: error reduction for two-sided-error case}~and~\ref{Theorem: BQlogP with constant gap is in BQ1P}
are proved in Subsection~\ref{Subsection: Proofs of Theorems for Error-Reduction for Two-Sided-Error Cases}.


\subsection{OR-Type Repetition Procedure}
\label{Subsection: OR-type repetition procedure}

Again consider any quantum circuit~$Q$ acting over $w$~qubits,
which is supposed to be applied to the $k$-clean-qubit initial state~%
${\rho_\init^{(w, k)} = (\ketbra{0})^{\tensor k} \tensor \bigl( \frac{I}{2} \bigr)^{\tensor (w - k)}}$.
This section presents a procedure,
called the \textsc{OR-Type Repetition Procedure},
that constructs another quantum circuit~$R^{(N)}$ from $Q$ when a positive integer~$N$ is specified.
Intuitively,
the circuit~$R^{(N)}$ is designed so that,
when some additional clean qubits are available
other than the $k$~clean qubits
that are used for the $k$-clean-qubit computation induced by $Q$,
one can perform $N$~attempts of the $k$-clean-qubit computation induced by $Q$
and accept if and only if at least one of these attempts results in acceptance
in the original computation of $Q$.

If ${Nk}$~clean qubits were allowed to use (in addition to one clean qubit as the output qubit),
the above procedure is easily constructed
by preparing $N$~copies of the initial state~$\rho_\init^{(w,k)}$,
applying $Q$ to each of these $N$~initial states,
and accepting if and only if at least one of the $N$~attempts results in acceptance.
The only problem in this construction is
that the resulting computation requires unallowably many clean qubits in its initial state,
as $N$ is desirably much larger than $k$
for the purpose of error reduction of the $k$-clean-qubit computations.

To realize this OR-type repetition idea,
consider the procedure described in Figure~\ref{Figure: simplified description of OR-Type Repetition Procedure}.
Instead of preparing $N$~copies of the initial state,
now the $k$~qubits in the register~$\regQ$ that were initially clean
are reused for each attempt.
In order to reuse the $k$~clean qubits,
the procedure tries to initialize these qubits by applying the inverse of $Q$ in Step~2.3
(and measuring each qubit in $\regQ$ in the computational basis in Step~2.4).
Of course, the initialization may not necessarily succeed every time,
and the failure of the initialization is also counted as an ``acceptance'' in that attempt.
The underlying idea of this procedure is essentially the same
as that in the method in Ref.~\cite{Wat01JCSS}
used for reducing the computation error of
one-sided bounded error logarithmic-space quantum Turing machines,
where the phase-shift transformation is applied
when the initialization attempt fails,
instead of counting it as an ``acceptance''.
It is unclear whether the phase-shift trick in Ref.~\cite{Wat01JCSS} works in the present case,
as one cannot perfectly judge in the present case whether the initialization succeeds or not,
due to the existence of polynomially many qubits that are initially in the totally mixed states.

\begin{figure}[t!]
  \begin{algorithm*}{\textsc{OR-Type Repetition Procedure} --- Simplified Description}
    \begin{step}
    \item
      Prepare a $k$-qubit register~$\regQ$
      where all the qubits in $\regQ$ are supposed to be initially in state~$\ket{0}$.
      For each $j$ in ${\{1, \dotsc, k\}}$,
      let $\regQ^{(j)}$ denote the single-qubit quantum register
      corresponding to the $j$th qubit of $\regQ$. 
      Prepare ${(w-k)}$-qubit registers~$\regR_j$,
      for $j$ in ${\{1, \dotsc, N\}}$,
      where all the qubits in $\regR_j$ are supposed to be initially in the totally mixed state~${I/2}$.
      Initialize a counter~$C$ to $0$.
    \item
      For ${j = 1}$ to $N$, perform the following:
      \begin{step}
      \item
        Apply $Q$ to ${(\regQ, \regR_j)}$.
      \item
        Measure the qubit in $\regQ^{(1)}$ in the computational basis.
        If this results in $\ket{0}$, increase the counter~$C$ by one.
      \item
        Apply $\conjugate{Q}$ to ${(\regQ, \regR_j)}$.
      \item
        Measure each qubit in $\regQ$ in the computational basis.
        If any of these results in $\ket{1}$, increase the counter~$C$ by one.
      \end{step}
    \item
      Reject if ${C = 0}$, and accept otherwise.
    \end{step}
  \end{algorithm*}
  \caption{The \textsc{OR-Type Repetition Procedure} (a simplified description).}
  \label{Figure: simplified description of OR-Type Repetition Procedure}
\end{figure}

The actual construction of the above simplified procedure uses
a single-qubit register~$\regO$ and an $l$-qubit register~$\regC$ for ${l = \ceil{\log N} + 1}$,
and further introduces a single-qubit quantum register~$\regX_j$
and a $k$-qubit register~$\regY_j$
for each $j$ in ${\{1, \dotsc, N\}}$.
The qubits in $\regO$, $\regC$, and $\regQ$ are supposed to be initially in state~$\ket{0}$,
and all the other qubits used in the actual procedure
are supposed to be initially in the totally mixed state~${I/2}$.
The qubit in $\regO$ serves as the output qubit,
and the content of $\regC$ serves as a counter~$C$ of the simplified procedure.
Each conditional increment of the counter value is realized
either by using the controlled-$\INCR{2^l}$ transformation~${\controlled \bigl( \INCR{2^l} \bigr)}$
or by combining the $\INCR{2^l}$~transformation
and the $k$-controlled-$\INCR{2^l}$ transformation~${\controlled^k \bigl( \INCR{2^l} \bigr)}$.
The decision of acceptance and rejection at Step~3 of the simplified procedure
can be simulated by combining $\NOT$ transformations and a generalized Toffoli transformation.
Each register~$\regX_j$ is used to simulate the measurement at Step~2.2
while each register~$\regY_j$ is used to simulate the measurement at Step~2.4,
for each repetition round~$j$ of Step~2.
More precisely,
for each repetition round~$j$ of Step~2,
the measurement of the qubit in $\regQ^{(1)}$ in the computational basis at Step~2.2
is replaced by the application of a $\CNOT$ transformation to ${(\regQ^{(1)}, \regX_j)}$,
while the measurements of the qubits in $\regQ$ in the computational basis at Step~2.4
is replaced by the applications of $\CNOT$ transformations to ${(\regQ, \regY_j)}$.
Figure~\ref{Figure: OR-Type Repetition Procedure}
summarizes the actual construction of the \textsc{OR-Type Repetition Procedure}.

\begin{figure}[t!]
  \begin{algorithm*}{\textsc{OR-Type Repetition Procedure}}
    \begin{step}
    \item
      Given a positive integer~$N$,
      let ${l = \ceil{\log N} + 1}$.
      Prepare a single-qubit register~$\regO$,
      an $l$-qubit register~$\regC$,
      and a $k$-qubit register~$\regQ$,
      where all the qubits in $\regO$, $\regC$, and $\regQ$
      are supposed to be initially in state~$\ket{0}$.
      For each $i$ in ${\{1, \dotsc, k\}}$,
      let $\regQ^{(i)}$ denote the single-qubit quantum register
      corresponding to the $i$th qubit of $\regQ$.
      Prepare a ${(w - k)}$-qubit register~$\regR_j$,
      a single-qubit register~$\regX_j$,
      and a $k$-qubit register~$\regY_j$,
      for each $j$ in ${\{1, \dotsc, N\}}$,
      where all the qubits in $\regR_j$, $\regX_j$, and $\regY_j$
      are supposed to be initially in the totally mixed state~${I/2}$.
      For each $i$ in ${\{1, \dotsc, k\}}$
      and $j$ in ${\{1, \dotsc, N\}}$,
      let $\regY_j^{(i)}$ denote the single-qubit quantum register
      corresponding to the $i$th qubit of $\regY_j$.
    \item
      For ${j = 1}$ to $N$, perform the following:
      \begin{step}
      \item
        Apply $Q_x$ to ${(\regQ, \regR_j)}$.
      \item
        Apply the $\CNOT$ transformation to ${(\regQ^{(1)}, \regX_j)}$
        with the qubit in $\regQ^{(1)}$ being the control.\\
        Apply $\INCR{2^l}$ to $\regC$ if the content of $\regQ^{(1)}$ is $0$
        (this can be realized by first applying $X$ to $\regQ^{(1)}$,
        then applying the controlled-$\INCR{2^l}$ transformation~${\controlled \bigl( \INCR{2^l} \bigr)}$ to ${(\regQ^{(1)}, \regC)}$
        with the qubit in $\regQ^{(1)}$ being the control,
        and further applying $X$ to $\regQ^{(1)}$).
      \item
        Apply $\conjugate{Q_x}$ to ${(\regQ, \regR_j)}$.
      \item
        For each $i$ in ${\{1, \dotsc, k\}}$,
        apply the $\CNOT$ transformation to ${\bigl( \regQ^{(i)}, \regY_j^{(i)} \bigr)}$
        with the qubit in $\regQ^{(i)}$ being the control.\\
        Apply $\INCR{2^l}$ to $\regC$ if any of the qubits in $\regQ$ contains $1$
        (this can be realized by first applying $\INCR{2^l}$ to $\regC$,
        next applying $X$ to each of the qubits in $\regQ$,
        then applying the $k$-controlled-$\conjugate{\INCR{2^l}}$ transformation~${\controlled^k \bigl( \conjugate{\INCR{2^l}} \bigr)}$
        to ${(\regQ, \regC)}$
        with the qubits in $\regQ$ being the control,
        and further applying $X$ to each of the qubits in $\regQ$).
      \end{step}
    \item
      Apply $X$ to $\regO$ if any of the qubits in $\regC$ contains $1$
      (this can be realized by first applying $X$ to each of the qubits in $\regO$ and $\regC$,
      and then applying the generalized Toffoli transformation~${\controlled^l(X)}$ to ${(\regC, \regO)}$
      with the qubit in $\regO$ being the target
      --- and further applying $X$ to each of the qubits in $\regC$
      to precisely realize the transformation in this step,
      which is unnecessary for the purpose of this procedure).\\
      Measure the qubit in $\regO$ in the computational basis.
      Accept if this results in $\ket{0}$, and reject otherwise.
    \end{step}
  \end{algorithm*}
  \caption{The \textsc{OR-Type Repetition Procedure}.}
  \label{Figure: OR-Type Repetition Procedure}
\end{figure}

\newpage 

\begin{proposition}
  For any quantum circuit~$Q$ and any positive integer~$N$,
  let $R^{(N)}$ be the quantum circuit
  corresponding to the \textsc{OR-Type Repetition Procedure} induced by $Q$ and $N$.
  For the acceptance probability~${p_\acc(Q, k)}$
  of the $k$-clean-qubit computation induced by $Q$
  and the acceptance probability~${p_\acc \bigl( R^{(N)}, k + \ceil{\log N} + 2 \bigr)}$
  of the ${(k + \ceil{\log N} + 2)}$-clean-qubit computation induced by $R^{(N)}$,
  it holds that
  \[
    1 - (1 - \op{p_\acc}(Q, k))^N
    \leq
    \op{p_\acc} \bigl( R^{(N)}, k + \ceil{\log N} + 2 \bigr)
    \leq
    1 - (1 - \op{p_\acc}(Q, k))^{2N}.
  \]
  \label{Proposition: acceptance probability of OR-Type Repetition Procedure}
\end{proposition}

\begin{proof}
  For ease of explanations,
  the proof analyzes the simplified version of the \textsc{OR-Type Repetition Procedure}
  in Figure~\ref{Figure: simplified description of OR-Type Repetition Procedure},
  which is sufficient for the claim,
  as the actual construction of the \textsc{OR-Type Repetition Procedure}
  exactly simulates the simplified version.

  First notice that the counter value remains zero when entering Step~3
  if and only if
  the simulation of $Q$ results in rejection
  (i.e., the measurement in Step~2.2 results in $\ket{1}$)
  \emph{and}
  the initialization of the qubits in $\regQ$ succeeds
  (i.e., none of the measurements in Step~2.4 results in $\ket{1}$)
  for all the $N$~attempts during Step~2.
  Suppose that the qubits in ${(\regQ, \regR_j)}$ form the state~%
  ${
    \ket{\psi_r} = \ket{0}^{\tensor k} \tensor \ket{r}
  }$
  when entering Step~2.1,
  for some $r$ in $\Sigma^{(w-k)}$,
  and let $p_r$ be the probability defined by
  \[
    p_r = \bignorm{\Pi_1 Q \ket{\psi_r}}^2,
  \]
  where $\Pi_1$ is the projection operator acting over $w$~qubits defined by
  ${\Pi_1 = \ketbra{1} \tensor I^{\tensor (w - 1)}}$.
  Then, it is clear that the measurement at Step~2.2 results in $\ket{1}$ with probability exactly $p_r$.
  Hence,
  by letting $\Delta_0$ be the projection operator acting over $w$~qubits defined by
  ${\Delta_0 = (\ketbra{0})^{\tensor k} \tensor I^{\tensor (w - k)}}$,
  the joint probability that
  the measurement at Step~2.2 results in $\ket{1}$
  \emph{and}
  none of the measurements in Step~2.4 results in $\ket{1}$
  is given by
  \[
    q_r = \bignorm{\Delta_0 \conjugate{Q} \Pi_1 Q \ket{\psi_r}}^2.
  \]

  To prove the first inequality of the claim,
  notice that ${q_r = \bignorm{\Delta_0 \conjugate{Q} \Pi_1 Q \ket{\psi_r}}^2}$
  is at most ${p_r = \bignorm{\Pi_1 Q \ket{\psi_r}}^2}$.
  As the acceptance probability~${p_\acc(Q, k)}$
  of the $k$-clean-qubit computation induced by $Q$
  is nothing but the expected value of ${1 - p_r}$ over $r$ in $\Sigma^{(w - k)}$,
  it holds that
  \[
    2^{-(w - k)} \sum_{r \in \Sigma^{(w - k)}} q_r
    \
    \leq
    \
    2^{-(w - k)} \sum_{r \in \Sigma^{(w - k)}} p_r
    \
    =
    \
    1 - \op{p_\acc}(Q, k),
  \]
  and thus, the probability is at most ${1 - \op{p_\acc}(Q, k)}$
  for the event that the counter value remains zero after one iteration of Step~2,
  conditioned that the counter value was zero when starting that iteration.
  Overall, the probability that the counter value remains zero when entering Step~3 is at most
  ${
    (1 - \op{p_\acc}(Q, k))^N
  }$,
  and thus, the procedure results in acceptance with probability at least
  ${
    1 - (1 - \op{p_\acc}(Q, k))^N
  }$,
  which shows the first inequality.

  For the second inequality of the claim,
  note that
  ${
    \Delta_0
    =
    (\ketbra{0})^{\tensor k} \tensor \sum_{r \in \Sigma^{(w - k)}} \ketbra{r}
    =
    \sum_{r \in \Sigma^{(w - k)}} \ketbra{\psi_r}
  }$,
  and thus, it holds that
  \[
    q_r
    \geq
    \bignorm{\ketbra{\psi_r} \conjugate{Q} \Pi_1 Q \ket{\psi_r}}^2
    =
    \bigabs{\bra{\psi_r} \conjugate{Q} \Pi_1 Q \ket{\psi_r}}^2
    =
    \bignorm{\Pi_1 Q \ket{\psi_r}}^4
    =
    p_r^2.
  \]
  Again using the fact that
  ${p_\acc(Q, k)}$
  is nothing but the expected value of ${1 - p_r}$ over $r$ in $\Sigma^{(w - k)}$,
  it follows that
  \[
    \begin{split}
      2^{-(w - k)} \sum_{r \in \Sigma^{(w - k)}} q_r
      \
      &
      \geq
      \
      2^{-(w - k)} \sum_{r \in \Sigma^{(w - k)}} p_r^2
      \\
      &
      \geq
      \biggl(
      2^{-(w - k)} \sum_{r \in \Sigma^{(w - k)}} p_r
      \biggr)^2
      =
      \biggl(
      1
      -
      2^{-(w - k)} \sum_{r \in \Sigma^{(w - k)}} (1 - p_r)
      \biggr)^2
      =
      (1 - \op{p_\acc}(Q, k))^2,
    \end{split}
  \]
  and thus, the probability is at least ${(1 - \op{p_\acc}(Q, k))^2}$
  for the event that the counter value remains zero after one iteration of Step~2,
  conditioned that the counter value was zero when starting that iteration.
  Overall, the probability that the counter value remains zero when entering Step~3 is at least
  ${
    \bigl[ (1 - \op{p_\acc}(Q, k))^2 \bigr]^N
    =
    (1 - \op{p_\acc}(Q, k))^{2N}
  }$.
  Accordingly, the procedure results in acceptance with probability at most
  ${
    1 - (1 - \op{p_\acc}(Q, k))^{2N}
  }$,
  and the second inequality follows.
\end{proof}


\subsection{Proofs of Theorems~\ref{Theorem: error reduction for two-sided-error case}~and~\ref{Theorem: BQlogP with constant gap is in BQ1P}}
\label{Subsection: Proofs of Theorems for Error-Reduction for Two-Sided-Error Cases}

Now we are ready to prove Theorems~\ref{Theorem: error reduction for two-sided-error case}~and~\ref{Theorem: BQlogP with constant gap is in BQ1P}.
First, we prove two lemmas.

\begin{lemma}
  For any constants~$c$~and~$s$ in $\Real$ satisfying ${0 < s < c < 1}$,
  \[
    \QlogP(c, s)
    \subseteq
    \QoneP \Bigl( \frac{3}{4}, \, \frac{5}{8} \Bigr).
  \]
  \label{Lemma: BQlogP with constant gap is in BQ1P with constant gap}
\end{lemma}

\begin{proof}
  For any constants~$c$~and~$s$ in $\Real$ satisfying ${0 < s < c < 1}$,
  let ${A = (A_\yes, A_\no)}$ be a problem in ${\QlogP(c, s)}$.
  Then $A$ has a polynomial-time uniformly generated family~${\{Q_x\}_{x \in \Sigma^\ast}}$
  of quantum circuits
  such that, for every input~$x$,
  ${\op{p_\acc}(Q_x, \op{k}_1(\abs{x})) \geq c}$ if $x$ is in $A_\yes$
  and
  ${\op{p_\acc}(Q_x, \op{k}_1(\abs{x})) \leq s}$ if $x$ is in $A_\no$
  for some logarithmically bounded function~$\function{k_1}{\Nonnegative}{\Natural}$.
  The proof constructs a polynomial-time uniformly generated family of quantum circuits
  that puts $A$ in ${\QoneP \bigl( \frac{3}{4}, \:\! \frac{5}{8} \bigr)}$.

  Fix an input~$x$.

  From the circuit~$Q_x$,
  one first constructs a quantum circuit~$R_x$
  that runs $N$~attempts of $Q_x$ in parallel,
  and accepts if and only if at least $\frac{c+s}{2}$~fraction of the $N$~attempts results in acceptance.
  Notice that this can be easily implementable if $N$ is a power of two,
  by combining the increment transformation
  with the threshold-check transformation discussed in Subsection~\ref{Subsection: quantum circuits}.
  By taking a sufficiently large constant~$N$ that is a power of two,
  it holds that
  ${\op{p_\acc}(R_x, \op{k}_2(\abs{x})) \geq \frac{15}{16}}$ if $x$ is in $A_\yes$,
  and
  ${\op{p_\acc}(R_x, \op{k}_2(\abs{x})) \leq \frac{1}{16}}$ if $x$ is in $A_\no$,
  where $\function{k_2}{\Nonnegative}{\Natural}$ is the logarithmically bounded function defined by
  ${k_2 = N k_1 + \log N + 1}$.

  From the circuit~$R_x$,
  one constructs a quantum circuit~$R'_x$
  according to the \textsc{One-Clean-Qubit Simulation Procedure}
  with ${k = k_2}$.
  By Proposition~\ref{Proposition: acceptance probability of One-Clean-Qubit Simulation Procedure},
  it holds that
  ${
    \op{p_\acc}(R'_x, 1)
    \geq
    1 - \frac{31}{256 \op{q}(\abs{x})}
    >
    1 - \frac{1}{8 \op{q}(\abs{x})}
  }$
  if $x$ is in $A_\yes$,
  and
  ${
    \op{p_\acc}(R'_x, 1)
    \leq
    1 - \frac{15}{16 \op{q}(\abs{x})}
  }$
  if $x$ is in $A_\no$,
  where $\function{q}{\Nonnegative}{\Natural}$ is the polynomially bounded function
  defined by ${q = 2^{k_2}}$.

  Finally, from the circuit~$R'_x$,
  one constructs a quantum circuit~$R''_x$
  according to the \textsc{Randomness Amplification Procedure}
  with the integer~${N = 2 \op{q}(\abs{x})}$.
  By Proposition~\ref{Proposition: acceptance probability of Randomness Amplification Procedure},
  it holds that
  \[
    \op{p_\acc}(R''_x, 1)
    \geq
    \frac{1}{2}
    +
    \frac{1}{2} \biggl( 1 - \frac{1}{4 \op{q}(\abs{x})} \biggr)^{2 \op{q}(\abs{x})}
    >
    \frac{1}{2}
    +
    \frac{1}{2} \biggl( 1 - 2 \op{q}(\abs{x}) \cdot \frac{1}{4 \op{q}(\abs{x})} \biggr)
    =
    \frac{3}{4}
  \]
  if $x$ is in $A_\yes$,
  and
  \[
    \op{p_\acc}(R''_x, 1)
    \leq
    \frac{1}{2}
    +
    \frac{1}{2} \biggl( 1 - \frac{15}{8 \op{q}(\abs{x})} \biggr)^{2 \op{q}(\abs{x})}
    <
    \frac{1}{2}
    +
    \frac{1}{2} \biggl( 1 - \frac{15}{8 \op{q}(\abs{x})} \biggr)^{\frac{16}{15} \op{q}(\abs{x})}
    <
    \frac{1}{2}
    +
    \frac{1}{2 e^2}
    <
    \frac{5}{8}
  \]
  if $x$ is in $A_\no$.

  The claim follows
  with the polynomial-time uniformly generated family~${\{R''_x\}_{x \in \Sigma^\ast}}$ of quantum circuits.
\end{proof}

\begin{lemma}
  For any constants~$c$~and~$s$ in $\Real$ satisfying ${0 < s < c < 1}$
  and for any polynomially bounded function~$\function{p}{\Nonnegative}{\Natural}$,
  \[
    \QoneP(c, s)
    \subseteq
    \QlogP \Bigl( 1 - 2^{-p}, \, \frac{1}{2} \Bigr).
  \]
  \label{Lemma: BQ1P with constant gap is in BQlogP with almost perfect completeness}
\end{lemma}

\begin{proof}
  For any constants~$c$~and~$s$ in $\Real$ satisfying ${0 < s < c < 1}$,
  let ${A = (A_\yes, A_\no)}$ be a problem in ${\QoneP(c, s)}$.
  Then $A$ has a polynomial-time uniformly generated family~${\{Q_x\}_{x \in \Sigma^\ast}}$
  of quantum circuits
  such that, for every input~$x$,
  ${\op{p_\acc}(Q_x, 1) \geq c}$ if $x$ is in $A_\yes$
  and
  ${\op{p_\acc}(Q_x, 1) \leq s}$ if $x$ is in $A_\no$.
  For any polynomially bounded function~$\function{p}{\Nonnegative}{\Natural}$,
  the proof constructs a polynomial-time uniformly generated family of quantum circuits
  that puts $A$ in ${\QlogP \bigl( 1 - 2^{-p}, \frac{1}{2} \bigr)}$.

  Fix an input~$x$.

  From the circuit~$Q_x$,
  one first constructs a quantum circuit~$R_x$
  that runs ${\op{l}(\abs{x})}$~attempts of $Q_x$ in parallel,
  and accepts if and only if at least $\frac{c+s}{2}$~fraction of the ${\op{l}(\abs{x})}$~attempts results in acceptance,
  where $\function{l}{\Nonnegative}{\Natural}$ is some logarithmically bounded function
  such that ${\op{l}(n)}$ is a power of two for every $n$ in $\Nonnegative$.
  Again notice that this can be easily implementable
  by combining the increment transformation
  with the threshold-check transformation discussed in Subsection~\ref{Subsection: quantum circuits}.
  By taking such a function~$l$ that is at least ${\frac{2}{(c - s)^2} (\log p + 2)}$, 
  it holds that
  ${\op{p_\acc}(R_x, \op{k}(\abs{x})) \geq 1 - \frac{1}{4 \op{p}(\abs{x})}}$ if $x$ is in $A_\yes$,
  and
  ${\op{p_\acc}(R_x, \op{k}(\abs{x})) \leq \frac{1}{4 \op{p}(\abs{x})}}$ if $x$ is in $A_\no$,
  where $\function{k}{\Nonnegative}{\Natural}$ is the logarithmically bounded function defined by
  ${k = l + \log l + 1}$.

  From the circuit~$R_x$,
  one further constructs a quantum circuit~$R'_x$
  according to the \textsc{OR-Type Repetition Procedure}
  with the integer~${N = \op{p}(\abs{x})}$.
  By Proposition~\ref{Proposition: acceptance probability of OR-Type Repetition Procedure},
  it holds that
  \[
    \op{p_\acc} \bigl( R'_x, \op{k}(\abs{x}) + \ceil{\log \op{p}(\abs{x})} + 2 \bigr)
    \geq
    1 - \biggl( \frac{1}{4 \op{p}(\abs{x})} \biggr)^{\op{p}(\abs{x})}
    >
    1 - 2^{- \op{p}(\abs{x})}
  \]
  if $x$ is in $A_\yes$,
  and
  \[
    \op{p_\acc} \bigl( R'_x, \op{k}(\abs{x}) + \ceil{\log \op{p}(\abs{x})} + 2 \bigr)
    <
    1 - \biggl( 1 - \frac{1}{4 \op{p}(\abs{x})} \biggr)^{2 \op{p}(\abs{x})}
    <
    1 - \biggl( 1 - 2 \op{p}(\abs{x}) \cdot \frac{1}{4 \op{p}(\abs{x})} \biggr)
    =
    \frac{1}{2}
  \]
  if $x$ is in $A_\no$.

  The claim follows
  with the polynomial-time uniformly generated family~${\{R'_x\}_{x \in \Sigma^\ast}}$ of quantum circuits.
\end{proof}

\subsubsection{One-Clean-Qubit Case}

By using Lemmas~\ref{Lemma: BQlogP with constant gap is in BQ1P with constant gap}~and~\ref{Lemma: BQ1P with constant gap is in BQlogP with almost perfect completeness},
Theorem~\ref{Theorem: BQlogP with constant gap is in BQ1P}
is proved as follows.

\begin{proof}[Proof of Theorem~\ref{Theorem: BQlogP with constant gap is in BQ1P}]
  For any constants~$c$~and~$s$ in $\Real$ satisfying ${0 < s < c < 1}$,
  let ${A = (A_\yes, A_\no)}$ be a problem in ${\QlogP(c, s)}$.
  It is sufficient to show
  that the containment~${\QlogP(c, s) \subseteq \QoneP \bigl( 1 - 2^{-p}, \frac{1}{p} \bigr)}$ holds
  for any polynomially bounded function~$\function{p}{\Nonnegative}{\Natural}$.
  Indeed, the containment~${\QlogP(c, s) \subseteq \QoneP \bigl( 1 - \frac{1}{p}, 2^{-p} \bigr)}$
  is then proved by considering the complement problem~${\overline{A} = (A_\no, A_\yes)}$ of $A$:
  As $\overline{A}$ is in ${\QlogP(\varepsilon, \delta)}$
  for the constants~${\varepsilon = 1 - s}$ and ${\delta = 1 - c}$
  satisfying ${0 < \delta < \varepsilon < 1}$,
  the containment above to be proved implies
  that $\overline{A}$ is in ${\QoneP \bigl( 1 - 2^{-p}, \frac{1}{p} \bigr)}$,
  and thus, that $A$ is in ${\QoneP \bigl( 1 - \frac{1}{p}, 2^{-p} \bigr)}$,
  for any polynomially bounded function~$\function{p}{\Nonnegative}{\Natural}$.

  From Lemma~\ref{Lemma: BQlogP with constant gap is in BQ1P with constant gap},
  $A$ is in ${\QoneP \bigl( \frac{3}{4}, \:\! \frac{5}{8} \bigr)}$.
  Therefore, from Lemma~\ref{Lemma: BQ1P with constant gap is in BQlogP with almost perfect completeness},
  $A$ is in ${\QlogP \bigl( 1 - 2^{- q}, \frac{1}{2} \bigr)}$,
  where $\function{q}{\Nonnegative}{\Natural}$ is the polynomially bounded function
  defined by ${q = p + 4 \ceil{\log p} + 13}$.
  Hence, $A$ has a polynomial-time uniformly generated family~${\{Q_x\}_{x \in \Sigma^\ast}}$
  of quantum circuits
  such that, for every input~$x$,
  ${\op{p_\acc}(Q_x, \op{k}(\abs{x})) \geq 1 - 2^{- \op{q}(\abs{x})}}$ if $x$ is in $A_\yes$
  and
  ${\op{p_\acc}(Q_x, \op{k}(\abs{x})) \leq \frac{1}{2}}$ if $x$ is in $A_\no$,
  for some logarithmically bounded function~$\function{k}{\Nonnegative}{\Natural}$.
  The rest of the proof is essentially the same as the proof of Theorem~\ref{Theorem: co-RQlogP is in co-RQ1P}
  in Section~\ref{Section: error reduction for one-sided-error cases}.

  Fix an input~$x$.

  From the circuit~$Q_x$,
  one first constructs a quantum circuit~$R_x$
  according to the \textsc{One-Clean-Qubit Simulation Procedure}.
  By Proposition~\ref{Proposition: acceptance probability of One-Clean-Qubit Simulation Procedure},
  it holds that
  \[
    \op{p_\acc}(R_x, 1)
    \geq
    1 - 2^{- \op{k}(\abs{x})} \Bigl[ 1 - \bigl( 1 - 2^{- \op{q}(\abs{x})} \bigr)^2 \Bigr]
    >
    1 - 2^{- \op{q}(\abs{x}) - \op{k}(\abs{x}) + 1}
    =
    1 - \frac{2^{- \op{q}(\abs{x}) + 2}}{\op{r}_1(\abs{x})}
  \]
  if $x$ is in $A_\yes$,
  and
  \[
    1 - \frac{2}{\op{r}(\abs{x})}
    =
    1 - 2^{-\op{k}(\abs{x})}
    \leq
    \op{p_\acc}(R_x, 1)
    \leq
    1 - 2^{- \op{k}(\abs{x}) - 1}
    =
    1 - \frac{1}{\op{r}_1(\abs{x})}
  \]
  if $x$ is in $A_\no$,
  where $\function{r_1}{\Nonnegative}{\Natural}$ is the polynomially bounded function
  defined by ${r_1 = 2^{k + 1}}$.

  From the circuit~$R_x$,
  one constructs a quantum circuit~$R'_x$
  according to the \textsc{Randomness Amplification Procedure}
  with the integer~${N = \op{r}_1(\abs{x}) \op{r}_2(\abs{x})}$
  for a polynomially bounded function~$\function{r_2}{\Nonnegative}{\Natural}$
  such that ${r_2 = p + 1}$.
  By Proposition~\ref{Proposition: acceptance probability of Randomness Amplification Procedure},
  it holds that
  \[
    \begin{split}
      \hspace{5mm}
      &
      \hspace{-5mm}
      \op{p_\acc}(R'_x, 1)
      \geq
      \frac{1}{2}
      +
      \frac{1}{2} \biggl( 1 - \frac{2^{- \op{q}(\abs{x}) + 3}}{\op{r}_1(\abs{x})} \biggr)^{\op{r}_1(\abs{x}) \op{r}_2(\abs{x})}
      \\
      &
      >
      \frac{1}{2}
      +
      \frac{1}{2} \biggl( 1 - \op{r}_1(\abs{x}) \op{r}_2(\abs{x}) \cdot \frac{2^{- \op{q}(\abs{x}) + 3}}{\op{r}_1(\abs{x})} \biggr)
      =
      1
      -
      \op{r}_2(\abs{x}) \cdot 2^{- \op{q}(\abs{x}) + 2}
    \end{split}
  \]
  if $x$ is in $A_\yes$,
  while
  \[
    \frac{1}{2}
    \leq
    \op{p_\acc}(R'_x, 1)
    \leq
    \frac{1}{2}
    +
    \frac{1}{2} \biggl( 1 - \frac{2}{\op{r}_1(\abs{x})} \biggr)^{\op{r}_1(\abs{x}) \op{r}_2(\abs{x})}
    <
    \frac{1}{2}
    +
    \frac{1}{2} e^{- 2 \op{r}_2(\abs{x})}
    <
    \frac{1}{2} + 2^{- 2 \op{r}_2(\abs{x})}
  \]
  if $x$ is in $A_\no$.
 
  Finally, from the circuit~$R'_x$,
  one constructs a quantum circuit~$R''_x$
  according to the \textsc{One-Clean-Qubit Stability Checking Procedure}
  with $N$ to be the smallest integer that is a power of two and at least ${(6 \op{p}(\abs{x}))^3}$
  (i.e., ${N = 2^{\ceilL{3 \log (6 \op{p}(\abs{x}))}} > 2^6 = 64}$).
  If $x$ is in $A_\yes$,
  Proposition~\ref{Proposition: lower bound of acceptance probability of One-Clean-Qubit Stability Checking Procedure}
  ensures that
  \[
    \begin{split}
      \hspace{5mm}
      &
      \hspace{-5mm}
      \op{p_\acc}(R''_x, 1)
      \geq
      \Bigl( 1 - \op{r}_2(\abs{x}) \cdot 2^{- \op{q}(\abs{x}) + 2} \Bigr)^{4 \cdot (6 \op{p}(\abs{x}))^3 - 1}
      \\
      &
      >
      \Bigl( 1 - \op{r}_2(\abs{x}) \cdot 2^{- \op{q}(\abs{x}) + 2} \Bigr)^{2 \cdot (8 \op{p}(\abs{x}))^3}
      >
      1 - 2 \cdot (8 \op{p}(\abs{x}))^3 \cdot \op{r}_2(\abs{x}) \cdot 2^{- \op{q}(\abs{x}) + 2}
      \geq
      1 - 2^{- \op{p}(\abs{x})},
    \end{split}
  \]
  where the first inequality follows from the fact
  that $N$ is at most ${2 \cdot (6 \op{p}(\abs{x}))^3}$,
  the second inequality follows from the fact that ${2 \cdot (6 \op{p}(\abs{x}))^3 < (8 \op{p}(\abs{x}))^3}$,
  and the last inequality follows from the fact
  that ${\op{q}(\abs{x}) = \op{p}(\abs{x}) + 4 \ceil{\log \op{p}(\abs{x})} + 13}$,
  that ${(8 \op{p}(\abs{x}))^3 = 2^{3 \log (8 \op{p}(\abs{x}))} = 2^{3 \log \op{p}(\abs{x}) + 9}}$,
  that ${\op{r}_2(\abs{x}) = \op{p}(\abs{x}) + 1 = 2^{\log (\op{p}(\abs{x}) + 1)}}$,
  and that the inequality~${\log (n + 1) \leq \log n + 1}$ holds for any ${n \geq 1}$.
  On the other hand,
  if $x$ is in $A_\no$,
  Proposition~\ref{Proposition: upper bound of acceptance probability of One-Clean-Qubit Stability Checking Procedure},
  ensures that
  \[
    \op{p_\acc}(R''_x, 1)
    \leq
    \frac{1}{2 \op{p}(\abs{x})}
    +
    4 \cdot 2^{- 2 \op{r}_2(\abs{x})}
    <
    \frac{1}{\op{p}(\abs{x})},
  \]
  where the last inequality follows from
  the facts that ${r_2 \geq p + 1}$
  and that the inequality~${2^n > n}$ holds.
  
  The claim follows
  with the polynomial-time uniformly generated family~${\{R''_x\}_{x \in \Sigma^\ast}}$ of quantum circuits.
\end{proof}

\subsubsection{Two-Clean-Qubit Case}

Again by using Lemmas~\ref{Lemma: BQlogP with constant gap is in BQ1P with constant gap}~and~\ref{Lemma: BQ1P with constant gap is in BQlogP with almost perfect completeness},
Theorem~\ref{Theorem: error reduction for two-sided-error case}
is proved as follows.

\begin{proof}[Proof of Theorem~\ref{Theorem: error reduction for two-sided-error case}]
  The proof is essentially the same as that of Theorem~\ref{Theorem: BQlogP with constant gap is in BQ1P},
  except that the circuit~$R''_x$ is constructed
  according to the \textsc{Two-Clean-Qubit Stability Checking Procedure}
  rather than the \textsc{One-Clean-Qubit Stability Checking Procedure}.

  For any constants~$c$~and~$s$ in $\Real$ satisfying ${0 < s < c < 1}$,
  let ${A = (A_\yes, A_\no)}$ be a problem in ${\QlogP(c, s)}$.
  From Lemma~\ref{Lemma: BQlogP with constant gap is in BQ1P with constant gap},
  $A$ is in ${\QoneP \bigl( \frac{3}{4}, \:\! \frac{5}{8} \bigr)}$.
  Therefore, from Lemma~\ref{Lemma: BQ1P with constant gap is in BQlogP with almost perfect completeness},
  $A$ is in ${\QlogP \bigl( 1 - 2^{- q}, \frac{1}{2} \bigr)}$,
  where $\function{q}{\Nonnegative}{\Natural}$ is the polynomially bounded function
  defined by ${q = p + \ceil{\log p} + 12}$.
  Hence, $A$ has a polynomial-time uniformly generated family~${\{Q_x\}_{x \in \Sigma^\ast}}$
  of quantum circuits
  such that, for every input~$x$,
  ${\op{p_\acc}(Q_x, \op{k}(\abs{x})) \geq 1 - 2^{- \op{q}(\abs{x})}}$ if $x$ is in $A_\yes$
  and
  ${\op{p_\acc}(Q_x, \op{k}(\abs{x})) \leq \frac{1}{2}}$ if $x$ is in $A_\no$,
  for some logarithmically bounded function~$\function{k}{\Nonnegative}{\Natural}$.
  The rest of the proof is essentially the same as the proof of Theorem~\ref{Theorem: error reduction for perfect-completeness case}
  in Section~\ref{Section: error reduction for one-sided-error cases}.

  Fix an input~$x$.

  From the circuit~$Q_x$,
  one first constructs a quantum circuit~$R_x$
  according to the \textsc{One-Clean-Qubit Simulation Procedure}.
  As in the proof of Theorem~\ref{Theorem: BQlogP with constant gap is in BQ1P},
  by Proposition~\ref{Proposition: acceptance probability of One-Clean-Qubit Simulation Procedure},
  it holds that
  \[
    \op{p_\acc}(R_x, 1)
    >
    1 - \frac{2^{- \op{q}(\abs{x}) + 2}}{\op{r}(\abs{x})}
  \]
  if $x$ is in $A_\yes$,
  and
  \[
    1 - \frac{2}{\op{r}(\abs{x})}
    \leq
    \op{p_\acc}(R_x, 1)
    \leq
    1 - \frac{1}{\op{r}(\abs{x})}
  \]
  if $x$ is in $A_\no$,
  where $\function{r}{\Nonnegative}{\Natural}$ is the polynomially bounded function
  defined by ${r = 2^{k + 1}}$.

  From the circuit~$R_x$,
  one constructs a quantum circuit~$R'_x$
  according to the \textsc{Randomness Amplification Procedure}
  with the integer~${N \geq \alpha \op{r}(\abs{x})}$, where ${\alpha = \frac{3}{2} \ln 2 < 1.04}$.
  By Proposition~\ref{Proposition: acceptance probability of Randomness Amplification Procedure},
  it holds that
  \[
    \begin{split}
      \hspace{5mm}
      &
      \hspace{-5mm}
      \op{p_\acc}(R'_x, 1)
      \geq
      \frac{1}{2}
      +
      \frac{1}{2} \biggl( 1 - \frac{2^{- \op{q}(\abs{x}) + 3}}{\op{r}(\abs{x})} \biggr)^{\alpha \op{r}(\abs{x})}
      \\
      &
      >
      \frac{1}{2}
      +
      \frac{1}{2} \biggl( 1 - \alpha \op{r}(\abs{x}) \cdot \frac{2^{- \op{q}(\abs{x}) + 3}}{\op{r}(\abs{x})} \biggr)
      =
      1
      -
      \alpha \cdot 2^{- \op{q}(\abs{x}) + 2}
    \end{split}
  \]
  if $x$ is in $A_\yes$,
  while
  \[
    \frac{1}{2}
    \leq
    \op{p_\acc}(R'_x, 1)
    \leq
    \frac{1}{2}
    +
    \frac{1}{2} \biggl( 1 - \frac{2}{\op{r}(\abs{x})} \biggr)^{\alpha \op{r}(\abs{x})}
    <
    \frac{1}{2}
    +
    \frac{1}{2} e^{- 3 \ln 2}
    =
    \frac{1}{2} + \frac{1}{16}
  \]
  if $x$ is in $A_\no$.
 
  Finally, from the circuit~$R'_x$,
  one constructs a quantum circuit~$R''_x$
  according to the \textsc{Two-Clean-Qubit Stability Checking Procedure}
  with $N$ to be the smallest integer that is a power of two and at least ${16 (\op{p}(\abs{x}) + 1)}$
  (i.e., ${N = 2^{\ceilL{\log (\op{p}(\abs{x}) + 1)} + 4}}$).
  If $x$ is in $A_\yes$,
  Proposition~\ref{Proposition: lower bound of acceptance probability of Two-Clean-Qubit Stability Checking Procedure}
  ensures that
  \[
    \op{p_\acc}(R''_x, 2)
    \geq
    \bigl( 1 - \alpha \cdot 2^{- \op{q}(\abs{x}) + 2} \bigr)^{256 (\op{p}(\abs{x}) + 1)}
    >
    1 - 256 (\op{p}(\abs{x}) + 1) \cdot \alpha \cdot 2^{- \op{q}(\abs{x}) + 2}
    >
    1 - 2^{- \op{p}(\abs{x})},
  \]
  where the first inequality follows from the fact
  that $N$ is at most ${2 \cdot 16 (\op{p}(\abs{x}) + 1) = 32 (\op{p}(\abs{x}) + 1)}$,
  and the last inequality follows from the fact
  that ${\op{q}(\abs{x}) = \op{p}(\abs{x}) + \ceil{\log \op{p}(\abs{x})} + 12}$,
  that ${\alpha = \frac{3}{2} \ln 2 < 1.04 < 2}$,
  that ${\op{p}(\abs{x}) + 1 = 2^{\log (\op{p}(\abs{x}) + 1)}}$,
  and that the inequality~${\log (n + 1) \leq \log n + 1}$ holds for any ${n \geq 1}$.

  On the other hand,
  if $x$ is in $A_\no$,
  Proposition~\ref{Proposition: upper bound of acceptance probability of Two-Clean-Qubit Stability Checking Procedure}
  ensures that
  \[
    \op{p_\acc}(R''_x, 2)
    \leq
    2^{- \op{p}(\abs{x})}
  \]
  (notice that the bounds for the probability~${\op{p_\acc}(R'_x, 1)}$ above
  are such that
  Proposition~\ref{Proposition: upper bound of acceptance probability of Two-Clean-Qubit Stability Checking Procedure}
  can be used to show the upper bound for the probability~${\op{p_\acc}(R''_x, 2)}$).

  The claim follows
  with the polynomial-time uniformly generated family~${\{R''_x\}_{x \in \Sigma^\ast}}$ of quantum circuits.
\end{proof}


\section{Completeness Results for \problemfont{Trace Estimation} Problem}
\label{Section: completeness results}

This section proves Theorem~\ref{Theorem: completeness of the trace estimation problem},
which states that the \problemfont{Trace Estimation} problem ${\TrEst(a, b)}$
with any parameters~$a$~and~$b$ satisfying ${0 < b < a < 1}$
is complete for both $\BQlogP$ and $\BQoneP$
under polynomial-time many-one reduction.

As presented in Ref.~\cite{ShoJor08QIC},
the following two properties are known to hold.

\begin{lemma}
  For any constants~$a$~and~$b$ in $\Real$ satisfying ${-1 \leq b < a \leq 1}$,
  ${\TrEst(a, b)}$ is in ${\QoneP \bigl( \frac{1}{2} + \frac{a}{2}, \frac{1}{2} + \frac{b}{2} \bigr)}$.
  \label{Lemma: TrEst(a,b) is in QoneP}
\end{lemma}

\begin{lemma}
  For any constants~$a$~and~$b$ in $\Real$ satisfying ${0 \leq b < a \leq 1}$,
  ${\TrEst(a, b)}$ is hard for
  ${\QoneP \bigl( \frac{a}{4}, \frac{b}{4} \bigr)}$
  under polynomial-time many-one reduction.
  \label{Lemma: QoneP-hardness of TrEst(a,b)}
\end{lemma}

Theorem~\ref{Theorem: completeness of the trace estimation problem} is then easily proved as follows,
by combining Theorem~\ref{Theorem: BQlogP with constant gap is in BQ1P} and these two lemmas.

\begin{proof}[Proof of Theorem~\ref{Theorem: completeness of the trace estimation problem}]
  We show the $\BQoneP$-completeness of ${\TrEst(a, b)}$.
  The $\BQlogP$-completeness is then trivial,
  as ${\BQlogP = \BQoneP}$ due to Theorem~\ref{Theorem: BQlogP with constant gap is in BQ1P}.
  In what follows,
  fix the parameters~$a$~and~$b$, ${0 < b < a < 1}$,
  of ${\TrEst(a, b)}$.

  For the membership that ${\TrEst(a, b)}$ is in $\BQoneP$,
  notice that the class~${\QoneP \bigl( \frac{1}{2} + \frac{a}{2}, \frac{1}{2} + \frac{b}{2} \bigr)}$
  is included in the class~${\QoneP(c, s)}$
  for any constants~$c$~and~$s$ satisfying that
  ${c \geq \max \bigl\{ \frac{2}{3}, \frac{1}{2} + \frac{a}{2} \bigr\}}$
  and
  ${s \leq \min \bigl\{ \frac{1}{3}, \frac{1}{2} + \frac{b}{2} \bigr\}}$,
  due to Theorem~\ref{Theorem: BQlogP with constant gap is in BQ1P}.
  Therefore, the class~${\QoneP \bigl( \frac{1}{2} + \frac{a}{2}, \frac{1}{2} + \frac{b}{2} \bigr)}$
  is included in $\BQoneP$ also.
  Hence, the membership is immediate,
  as Lemma~\ref{Lemma: TrEst(a,b) is in QoneP} ensures that
  ${\TrEst(a, b)}$ is in ${\QoneP \bigl( \frac{1}{2} + \frac{a}{2}, \frac{1}{2} + \frac{b}{2} \bigr)}$.

  Now for the $\BQoneP$-hardness of ${\TrEst(a, b)}$,
  notice that $\BQoneP$ is included in the class~${\QoneP(c, s)}$
  for any constants~$c$~and~$s$ satisfying that
  ${c \geq \max \bigl\{ \frac{2}{3}, \frac{a}{4} \bigr\}}$
  and
  ${s \leq \min \bigl\{ \frac{1}{3}, \frac{b}{4} \bigr\}}$,
  due to Theorem~\ref{Theorem: BQlogP with constant gap is in BQ1P}.
  Therefore, $\BQoneP$ is included
  in the class~${\QoneP \bigl( \frac{a}{4}, \frac{b}{4} \bigr)}$ also.
  From Lemma~\ref{Lemma: QoneP-hardness of TrEst(a,b)},
  ${\TrEst(a, b)}$ is hard for
  ${\QoneP \bigl( \frac{a}{4}, \frac{b}{4} \bigr)}$
  under polynomial-time many-one reduction,
  and thus,
  ${\TrEst(a, b)}$ is hard for $\BQoneP$ also,
  under polynomial-time many-one reduction.
\end{proof}

\begin{remark}
  As mentioned in Subsection~\ref{Subsection: remarks on uniformity},
  the completeness results of Theorem~\ref{Theorem: completeness of the trace estimation problem}
  hold even under logarithmic-space many-one reduction,
  if the classes~$\BQlogP$ and $\BQoneP$ are defined
  with logarithmic-space uniformly generated family of quantum circuits.
\end{remark}


\section{Hardness of Weak Classical Simulations of $\boldsymbol{\DQC{1}}$ Computation}
\label{Section: hardness of weak simulation}

This section deals with the hardness of weakly simulating a $\DQC{1}$ computation.
First, Subsection~\ref{Subsection: weak simulatability}
reviews the notions of weak simulatability that are discussed in this paper.
Subsection~\ref{Subsection: proofs of hardness of weak simulatability} then proves
Theorems~\ref{Theorem: NQP = NQ[1]P and SBQP = SBQ[1]P}~and~\ref{Theorem: hardness of classically simulating DQC1, informal}.


\subsection{Weak Simulatability}
\label{Subsection: weak simulatability}

Following conventions, this paper uses the following notions of simulatability.

Consider any family~${\{ Q_x \}_{x \in \Sigma^\ast}}$ of quantum circuits.
For each circuit~$Q_x$, suppose that $m$~output qubits are measured in the computational basis
after the application of $Q_x$ to a certain prescribed initial state
(which will be clear from the context).
Let $\function{P_x}{\Sigma^m}{[0,1]}$ be the probability distribution
derived from the output of $Q_x$
(i.e., ${P_x(y_1, \dotsc, y_m)}$
is the probability of obtaining the measurement result~${(y_1, \dotsc, y_m)}$ in $\Sigma^m$
when $Q_x$ is applied to the prescribed initial state).

The family~${\{Q_x\}_{x \in \Sigma^\ast}}$ is
\emph{weakly simulatable with multiplicative error~${c \geq 1}$}
if there exists a family~${\{P'_x\}_{x \in \Sigma^\ast}}$ of probability distributions
that can be sampled classically in polynomial time
such that, for any $x$ in $\Sigma^\ast$ and any ${(y_1, \dotsc, y_m)}$ in $\Sigma^m$,
\begin{equation}
  \frac{1}{c} \op{P}_x(y_1, \dotsc, y_m)
  \leq
  \op{P}'_x(y_1, \dotsc, y_m)
  \leq
  c \op{P}_x(y_1, \ldots, y_m).
  \label{Equation: multiplicative approximation}
\end{equation}

Similarly, the family~${\{Q_x\}_{x \in \Sigma^\ast}}$ is
\emph{weakly simulatable with exponentially small additive error}
if, for any polynomially bounded function~$q$,
there exists a family~${\{P'_x\}_{x \in \Sigma^\ast}}$ of probability distributions
that can be sampled classically in polynomial time
such that, for any $x$ in $\Sigma^\ast$ and any ${(y_1, \dotsc, y_m)}$ in $\Sigma^m$,
\[
  \bigabs{\op{P}_x(y_1, \dotsc, y_m) - \op{P}'_x(y_1, \dotsc, y_m)} \leq 2^{-\op{q}(\abs{x})}.
\]

A few remarks are in order regarding the notions of weak simulatablity above.

First, the notion of weak simulatablity with multiplicative error
was first defined in Ref.~\cite{TerDiV04QIC} in a slightly different form.
The definition taken in this paper is found in Refs.~\cite{BreJozShe11RSPA, MorFujFit14PRL},
for instance.
The version in Ref.~\cite{TerDiV04QIC} uses the bound~%
${
  \bigabs{\op{P}_x(y_1, \dotsc, y_m) - \op{P}'_x(y_1, \dotsc, y_m)}
  \leq
  \varepsilon \op{P}_x(y_1, \dotsc, y_m)
}$
instead of the bounds~(\ref{Equation: multiplicative approximation}),
and these two versions are essentially equivalent.
The results in this paper hold for any $\varepsilon$ in ${[0,1)}$
when using the version in Ref.~\cite{TerDiV04QIC}.

The notion of weak simulatability with exponentially small additive error
was introduced in Ref.~\cite{TakYamTan14QIC},
and was used also in Ref.~\cite{TakTanYamTan15COCOON}.

As the notion of weak simulatablity with multiplicative error
is often used when discussing the classical simulatability of quantum models~\cite{TerDiV04QIC, BreJozShe11RSPA, AarArk13ToC, NiVan13QIC, JozVan14QIC, MorFujFit14PRL, Bro15PRA},
the hardness result on the $\DQC{1}$ model under this notion
certainly makes it possible to discuss the power of the $\DQC{1}$ model
along the line of these existing studies.
As discussed in Refs.~\cite{BreJozShe11RSPA, AarArk13ToC}, however,
a much more reasonable notion is the weak simulatability with \emph{polynomially small additive error}
in total variation distance.
Proving or disproving classical simulatability under this notion
is one of the most important open problems in most of quantum computation models
including the $\DQC{1}$ model.


\subsection{Proofs of Theorems~\ref{Theorem: NQP = NQ[1]P and SBQP = SBQ[1]P}~and~\ref{Theorem: hardness of classically simulating DQC1, informal}}
\label{Subsection: proofs of hardness of weak simulatability}

First, 
Theorem~\ref{Theorem: NQP = NQ[1]P and SBQP = SBQ[1]P},
stating that the restriction to the $\DQC{1}$ computation
does not change the complexity classes~$\NQP$ and $\SBQP$,
can be easily proved by using the \textsc{One-Clean-Qubit Simulation Procedure}
presented in Subsection~\ref{Subsection: One-Clean-Qubit Simulation Procedure}.

\begin{proof}[Proof of Theorem~\ref{Theorem: NQP = NQ[1]P and SBQP = SBQ[1]P}]
  It suffices to show that ${\co\NQP \subseteq \co\NQoneP}$ and ${\co\SBQP \subseteq \co\SBQoneP}$.

  We first show that ${\co\NQP \subseteq \co\NQoneP}$.

  Consider any problem ${A = (A_\yes, A_\no)}$ is in ${\co\NQP}$,
  and let ${\{Q_x\}_{x \in \Sigma^\ast}}$
  be a polynomial-time uniformly generated family of quantum circuits that witnesses this fact.
  For each $x$ in $\Sigma^\ast$,
  the circuit~$Q_x$ acts over ${\op{w}(\abs{x})}$~qubits,
  for some polynomially bounded function~$\function{w}{\Nonnegative}{\Natural}$.
  By the definition of $\NQP$,
  for every input~$x$,
  the acceptance probability~${p_\acc \bigl( Q_x, \op{w}(\abs{x}) \bigr)}$ of the circuit~$Q_x$
  is one if $x$ is in $A_\yes$,
  while it is less than one if $x$ is in $A_\no$.

  Fix an input~$x$,
  and consider the quantum circuit~$R_x$
  corresponding to the \textsc{One-Clean-Qubit Simulation Procedure} induced by $Q_x$ and ${\op{w}(\abs{x})}$.
  By the properties of $Q_x$,
  Proposition~\ref{Proposition: acceptance probability of One-Clean-Qubit Simulation Procedure}
  ensures that
  the acceptance probability~${p_\acc(R_x, 1)}$ of the one-clean-qubit computation induced by $R_x$
  is one if $x$ is in $A_\yes$,
  while it is less than one if $x$ is in $A_\no$,
  which implies that $A$ is in ${\co\NQoneP}$.

  Now we show that ${\co\SBQP \subseteq \co\SBQoneP}$.

  Consider any problem ${A = (A_\yes, A_\no)}$ is in ${\co\SBQP}$.
  Then, by the amplification property of $\SBQP$ presented in Ref.~\cite{Kup15ToC},
  there exists a polynomial-time uniformly generated family~${\{Q_x\}_{x \in \Sigma^\ast}}$
  of quantum circuits
  such that, for every input~$x$,
  $Q_x$ acts over ${\op{w}(\abs{x})}$~qubits,
  for some polynomially bounded function~$\function{w}{\Nonnegative}{\Natural}$,
  and the acceptance probability~${p_\acc \bigl( Q_x, \op{w}(\abs{x}) \bigr)}$ of the circuit~$Q_x$
  is at least ${1 - 2^{- \op{p}(\abs{x}) - 2}}$ if $x$ is in $A_\yes$,
  while it is at most ${1 - 2^{- \op{p}(\abs{x})}}$ if $x$ is in $A_\no$,
  for some polynomially bounded function~$\function{p}{\Nonnegative}{\Natural}$.
  
  Fix an input~$x$,
  and consider the quantum circuit~$R_x$
  corresponding to the \textsc{One-Clean-Qubit Simulation Procedure} induced by $Q_x$ and ${\op{w}(\abs{x})}$.
  By the properties of $Q_x$,
  Proposition~\ref{Proposition: acceptance probability of One-Clean-Qubit Simulation Procedure}
  ensures that
  the acceptance probability~${p_\acc(R_x, 1)}$ of the one-clean-qubit computation induced by $R_x$
  is at least
  \[
    1 - 2^{- \op{w}(\abs{x})} \Bigl[ 1 - \bigl( 1 - 2^{- \op{p}(\abs{x}) - 2} \bigr)^2 \Bigr]
    >
    1 - 2^{- \op{w}(\abs{x})} \bigl[ 1 - \bigl( 1 - 2^{- \op{p}(\abs{x}) - 1} \bigr) \bigr]
    =
    1 - 2^{- \op{w}(\abs{x}) - \op{p}(\abs{x}) - 1},
  \]
  if $x$ is in $A_\yes$,
  while it is at most
  \[
    1 - 2^{- \op{w}(\abs{x})} \bigl[ 1 - \bigl( 1 - 2^{- \op{p}(\abs{x})} \bigr) \bigr]
    =
    1 - 2^{- \op{w}(\abs{x}) - \op{p}(\abs{x})},
  \]
 if $x$ is in $A_\no$,
 which ensures that $A$ is in ${\co\SBQoneP}$.
\end{proof}

In fact, the argument used to show that ${\SBQP = \SBQoneP}$
in the proof of Theorem~\ref{Theorem: NQP = NQ[1]P and SBQP = SBQ[1]P}
can be extended to prove the following amplification property of $\SBQoneP$,
which is analogous to the cases of $\SBQP$ and $\SBP$.

\begin{theorem}
  For any polynomially bounded function~$\function{p}{\Nonnegative}{\Natural}$,
  there exists a polynomially bounded function~$\function{q}{\Nonnegative}{\Natural}$
  such that
  \[
    \SBQoneP \subseteq \QoneP \bigl( 2^{-q} \cdot (1 - 2^{-p}), \, 2^{-q} \cdot 2^{-p} \bigr).
  \]
  \label{Theorem: amplification in SBQ[1]P}
\end{theorem}

\begin{proof}
  It suffices to show that
  for any polynomially bounded function~$\function{p}{\Nonnegative}{\Natural}$,
  there exists a polynomially bounded function~$\function{q}{\Nonnegative}{\Natural}$
  such that
  ${\co\SBQP \subseteq \QoneP \bigl(1 - 2^{-q} \cdot 2^{-p}, \, 1 - 2^{-q} \cdot (1 - 2^{-p}) \bigr)}$
  holds,
  due to Theorem~\ref{Theorem: NQP = NQ[1]P and SBQP = SBQ[1]P}.

  Fix a polynomially bounded function~$\function{p}{\Nonnegative}{\Natural}$,
  and consider any problem ${A = (A_\yes, A_\no)}$ is in ${\co\SBQP}$.
  By the amplification property of $\SBQP$ presented in Ref.~\cite{Kup15ToC},
  there exists a polynomial-time uniformly generated family~${\{Q_x\}_{x \in \Sigma^\ast}}$
  of quantum circuits
  such that, for every input~$x$,
  $Q_x$ acts over ${\op{w}(\abs{x})}$~qubits,
  for some polynomially bounded function~$\function{w}{\Nonnegative}{\Natural}$,
  and the acceptance probability~${p_\acc \bigl( Q_x, \op{w}(\abs{x}) \bigr)}$ of the circuit~$Q_x$
  is at least ${1 - 2^{-r} \cdot 2^{- p - 1}}$ if $x$ is in $A_\yes$,
  while it is at most ${1 - 2^{-r} \cdot (1 - 2^{- p - 1})}$ if $x$ is in $A_\no$,
  for some polynomially bounded function~$\function{r}{\Nonnegative}{\Natural}$.
  
  Fix an input~$x$,
  and consider the quantum circuit~$R_x$
  corresponding to the \textsc{One-Clean-Qubit Simulation Procedure} induced by $Q_x$ and ${\op{w}(\abs{x})}$.
  By the properties of $Q_x$,
  Proposition~\ref{Proposition: acceptance probability of One-Clean-Qubit Simulation Procedure}
  ensures that
  the acceptance probability~${p_\acc(R_x, 1)}$ of the one-clean-qubit computation induced by $R_x$
  is at least
  \[
    \begin{split}
      \hspace{5mm}
      &
      \hspace{-5mm}
      1
      -
      2^{- \op{w}(\abs{x})}
      \Bigl[ 1 - \bigl( 1 - 2^{- \op{r}(\abs{x})} \cdot 2^{- \op{p}(\abs{x}) - 1} \bigr)^2 \Bigr]
      \\
      &
      >
      1
      -
      2^{- \op{w}(\abs{x})}
      \bigl[ 1 - \bigl( 1 - 2^{- \op{r}(\abs{x})} \cdot 2^{- \op{p}(\abs{x})} \bigr) \bigr]
      =
      1 - 2^{- \op{w}(\abs{x}) - \op{r}(\abs{x})} \cdot 2^{- \op{p}(\abs{x})}
    \end{split}
  \]
  if $x$ is in $A_\yes$,
  while it is at most
  \[
    \begin{split}
      \hspace{5mm}
      &
      \hspace{-5mm}
      1
      -
      2^{- \op{w}(\abs{x})}
      \Bigl[ 1 - \Bigl( 1 - 2^{- \op{r}(\abs{x})} \cdot \bigl( 1 - 2^{- \op{p}(\abs{x}) - 1} \bigr) \Bigr) \Bigr]
      \\
      &
      =
      1 - 2^{- \op{w}(\abs{x}) - \op{r}(\abs{x})} \cdot \bigl( 1 - 2^{- \op{p}(\abs{x}) - 1} \bigr)
      <
      1 - 2^{- \op{w}(\abs{x}) - \op{r}(\abs{x})} \cdot \bigl( 1 - 2^{- \op{p}(\abs{x})} \bigr)
    \end{split}
  \]
 if $x$ is in $A_\no$,
 which implies that $A$ is in
 ${\QoneP \bigl(1 - 2^{-q} \cdot 2^{-p}, \, 1 - 2^{-q} \cdot (1 - 2^{-p}) \bigr)}$
 for ${q = w + r}$,
 and the claim follows.
\end{proof}

Now we are ready to prove Theorem~\ref{Theorem: hardness of classically simulating DQC1, informal}.
More formally, we prove the following statement.

\begin{theorem}
  Suppose that any polynomial-time uniformly generated family of quantum circuits,
  when used in $\DQC{\mathit{1}}$ computations,
  is weakly simulatable with multiplicative error~${c \geq 1}$
  or exponentially small additive error.
  Then ${\PH = \AM}$.
  \label{Theorem: hardness of classically simulating DQC1}
\end{theorem}

Theorem~\ref{Theorem: hardness of classically simulating DQC1} follows directly
from Lemmas~\ref{Lemma: consequences of weak simulation of DQC1},~\ref{Lemma: PH is in AM if NQP is in NP},~and~\ref{Lemma: PH is in AM if NQP is in SBP} below,
combined with the fact that ${\AM \subseteq \PH}$.

\begin{lemma}
  Suppose that any polynomial-time uniformly generated family of quantum circuits,
  when used in $\DQC{\mathit{1}}$ computations,
  is weakly simulatable with multiplicative error~${c \geq 1}$
  (resp., exponentially small additive error).
  Then ${\NQP \subseteq \NP}$
  (resp., ${\NQP \subseteq \SBP}$).
  \label{Lemma: consequences of weak simulation of DQC1}
\end{lemma}

\begin{proof}
  Fix any problem~${A = (A_\yes, A_\no)}$ in $\NQP$.
  By Theorem~\ref{Theorem: NQP = NQ[1]P and SBQP = SBQ[1]P},
  there exists a polynomial-time uniformly generated family~${\{R_x\}_{x \in \Sigma^\ast}}$
  of quantum circuits
  such that, for every ${x \in \Sigma^\ast}$,
  the acceptance probability~${p_\acc(R_x, 1)}$ of the $\DQC{1}$ computation induced by $R_x$
  is nonzero if $x$ is in $A_\yes$,
  while it is zero if $x$ is in $A_\no$.
  From the assumption of this lemma,
  there exists a polynomial-time uniformly generated family~${\{R'_x\}_{x \in \Sigma^\ast}}$
  of randomized circuits
  that weakly simulates ${\{R_x\}_{x \in \Sigma^\ast}}$ with multiplicative error~${c \geq 1}$.
  Now the definition of weak simulatability with multiplicative error ensures that
  the probability that the circuit~$R'_x$ outputs $0$ (corresponding to acceptance) is nonzero
  if and only if ${\op{p}_\acc(R_x, 1)}$ is nonzero,
  which happens only when $x$ is in $A_\yes$.
  This implies that $A$ is in $\NP$.

  In the case where ${\{R'_x\}_{x \in \Sigma^\ast}}$ weakly simulates ${\{R_x\}_{x \in \Sigma^\ast}}$
  with exponentially small additive error,
  the proof uses the fact that the membership of $A$ in $\NQP$
  is witnessed by a polynomial-time uniformly generated family~${\{Q_x\}_{x \in \Sigma^\ast}}$
  of quantum circuits,
  where each $Q_x$ is composed only of Hadamard, Toffoli, and $\NOT$ gates
  (more precisely, each $Q_x$ is composed only of Hadamard, $T$, and $\CNOT$ gates,
  satisfying the definition of $\NQP$ in the present paper,
  but may be assumed to be composed in such a way that
  $T$ and $\CNOT$ gates are used only for applying Toffoli and $\NOT$ transformations).
  Notice that, for such $Q_x$,
  the acceptance probability, if it is nonzero, must be at least $2^{- \op{p}(\abs{x})}$
  for some polynomially bounded function~$\function{p}{\Nonnegative}{\Natural}$.
  Hence, from Theorem~\ref{Theorem: NQP = NQ[1]P and SBQP = SBQ[1]P} and its proof,
  the family~${\{R_x\}_{x \in \Sigma^\ast}}$ of quantum circuits
  that witnesses the membership of $A$ in ${\NQoneP = \NQP}$
  may be assumed to be such that, for each $R_x$,
  the acceptance probability~${p_\acc(R_x, 1)}$, if it is nonzero, must be at least $2^{- \op{q}(\abs{x})}$
  for some polynomially bounded function~$\function{q}{\Nonnegative}{\Natural}$.
  Furthermore, from the definition of weak simulatability with exponentially small additive error,
  the family~${\{R'_x\}_{x \in \Sigma^\ast}}$
  may be assumed to simulate ${\{R_x\}_{x \in \Sigma^\ast}}$
  with additive error at most~$2^{-\op{q}(\abs{x}) - 2}$.
  This implies that the probability that $R'_x$ outputs $0$ is at least
  ${3 \cdot 2^{- \op{q}(\abs{x}) - 2}}$
  if $x$ is in $A_\yes$,
  while it is at most
  $2^{-\op{q}(\abs{x}) - 2}$
  if $x$ is in $A_\no$,
  which ensures that $A$ is in $\SBP$.
\end{proof}

\begin{lemma}
  If ${\NQP \subseteq \NP}$, then ${\PH \subseteq \AM}$.
  \label{Lemma: PH is in AM if NQP is in NP}
\end{lemma}

\begin{lemma}
  If ${\NQP \subseteq \SBP}$, then ${\PH \subseteq \AM}$.
  \label{Lemma: PH is in AM if NQP is in SBP}
\end{lemma}

The proof of Lemma~\ref{Lemma: PH is in AM if NQP is in NP}
requires the notion of the $\BP$~operator,
whereas the proof of Lemma~\ref{Lemma: PH is in AM if NQP is in SBP}
requires the notion of the $\widehat{\BP}$~operator,
a variant of the $\BP$~operator introduced in Ref.~\cite{TodOgi92SIComp}.

For any complexity class~$\calC$ of promise problems,
a promise problem~${A = (A_\yes, A_\no)}$ is in ${\BP \cdot \calC}$
iff there exist a promise problem~${B = (B_\yes, B_\no)}$ in $\calC$
and a polynomially bounded function~$\function{r}{\Nonnegative}{\Natural}$
such that, for every $x$ in $\Sigma^\ast$, it holds that
\begin{alignat*}{3}
  &
  x \in A_\yes
  &
  &
  \Longrightarrow
  \bigabs{
    \set{z \in \Sigma^{\op{r}(\abs{x})}}{\lrangle{x, z} \in B_\yes}
  }
  &
  &
  \geq
  \frac{2}{3} \cdot 2^{\op{r}(\abs{x})}
  \quad
  \text{and}
  \\
  &
  x \in A_\no
  &
  &
  \Longrightarrow
  \bigabs{
    \set{z \in \Sigma^{\op{r}(\abs{x})}}{\lrangle{x, z} \in B_\no}
  }
  &
  &
  \geq
  \frac{2}{3} \cdot 2^{\op{r}(\abs{x})}.
\end{alignat*}

Similarly,
for any complexity class~$\calC$ of promise problems,
a promise problem~${A = (A_\yes, A_\no)}$ is in ${\widehat{\BP} \cdot \calC}$
iff for any polynomially bounded function~$\function{q}{\Nonnegative}{\Natural}$,
there exist a promise problem~${B = (B_\yes, B_\no)}$ in $\calC$
and a polynomially bounded function~$\function{r}{\Nonnegative}{\Natural}$
such that, for every $x$ in $\Sigma^\ast$, it holds that
\begin{alignat*}{3}
  &
  x \in A_\yes
  &
  &
  \Longrightarrow
  \bigabs{
    \set{z \in \Sigma^{\op{r}(\abs{x})}}{\lrangle{x, z} \in B_\yes}
  }
  &
  &
  \geq
  \bigl( 1 - 2^{- \op{q}(\abs{x})} \bigr) \cdot 2^{\op{r}(\abs{x})}
  \quad
  \text{and}
  \\
  &
  x \in A_\no
  &
  &
  \Longrightarrow
  \bigabs{
    \set{z \in \Sigma^{\op{r}(\abs{x})}}{\lrangle{x, z} \in B_\no}
  }
  &
  &
  \geq
  \bigl( 1 - 2^{- \op{q}(\abs{x})} \bigr) \cdot 2^{\op{r}(\abs{x})}.
\end{alignat*}

It is easy to see that ${\AM = \BP \cdot \NP}$.
By the standard error reduction of $\AM$,
one can also see that ${\AM = \BP \cdot \NP = \widehat{\BP} \cdot \NP}$.

Now Lemma~\ref{Lemma: PH is in AM if NQP is in NP} is proved as follows.

\begin{proof}[Proof of Lemma~\ref{Lemma: PH is in AM if NQP is in NP}]
  The claim follows from the following sequence of containments: 
  \[
    \PH \subseteq \BP \cdot \coCequalP = \BP \cdot \NQP \subseteq \BP \cdot \NP = \AM,
  \]
  where the first inclusion is by Corollary~2.5 of Ref.~\cite{TodOgi92SIComp} or Corollary~5.2 of Ref.~\cite{Tar93TCS},
  and we have used the fact that ${\NQP = \coCequalP}$~\cite{FenGreHomPru99RSPA},
  the assumption ${\NQP \subseteq \NP}$ of this lemma,
  and the fact that ${\AM = \BP \cdot \NP}$, respectively.
\end{proof}

Similarly, Lemma~\ref{Lemma: PH is in AM if NQP is in SBP} is proved as follows.

\begin{proof}[Proof of Lemma~\ref{Lemma: PH is in AM if NQP is in SBP}]
  The claim follows from the following sequence of containments:
  \[
    \PH
    \subseteq
    \widehat{\BP} \cdot \coCequalP
    =
    \widehat{\BP} \cdot \NQP
    \subseteq
    \widehat{\BP} \cdot \SBP
    \subseteq
    \widehat{\BP} \cdot \AM
    =
    \widehat{\BP} \cdot \widehat{\BP} \cdot \NP
    =
    \widehat{\BP} \cdot \NP
    =
    \AM,
  \]
  where the first inclusion is by Corollary 2.5 of Ref.~\cite{TodOgi92SIComp},
  the next equality uses the fact that ${\NQP = \coCequalP}$~\cite{FenGreHomPru99RSPA},
  the next two inclusions are by the assumption~${\NQP \subseteq \SBP}$ of this lemma
  and by the fact that ${\SBP \subseteq \AM}$~\cite{BohGlaMei06JCSS},
  and the last three equalities use the characterization of $\AM$ by ${\widehat{\BP} \cdot \NP}$
  as well as Lemma~2.8 of Ref.~\cite{TodOgi92SIComp} on the removability of a duplicate $\widehat{\BP}$ operator.
\end{proof}

\begin{remark}
  If any polynomial-time uniformly generated family of quantum circuits,
  when used in $\DQC{1}$ computations,
  is weakly simulatable either with multiplicative error~${c \geq 1}$
  or with exponentially small additive error,
  then the inclusion~${\SBQP \subseteq \SBP}$ can also be proved,
  by using the fact that ${\SBQP = \SBQoneP}$ proved in Theorem~\ref{Theorem: NQP = NQ[1]P and SBQP = SBQ[1]P}
  as well as the amplification property of $\SBQP$
  (or the amplification property of $\SBQoneP$ proved in Theorem~\ref{Theorem: amplification in SBQ[1]P}).
  This can also prove Theorem~\ref{Theorem: hardness of classically simulating DQC1, informal}
  (more formally, Theorem~\ref{Theorem: hardness of classically simulating DQC1}),
  the hardness result on such weak simulatability of $\DQC{1}$ computations,
  as the collapse of $\PH$ to $\AM$ also follows from the assumption~${\SBQP \subseteq \SBP}$,
  due to the inclusion~${\NQP \subseteq \SBQP}$.
\end{remark}

Finally,
the argument based on $\NQP$ and $\SBQP$
used to prove Theorem~\ref{Theorem: hardness of classically simulating DQC1, informal}
(more formally, Theorem~\ref{Theorem: hardness of classically simulating DQC1})
in this section
can also be used to show the hardness of weak classical simulations of other quantum computing models.
In particular, it can replace the existing argument based on $\postBQP$,
which was developed in Ref.~\cite{BreJozShe11RSPA}
and has appeared frequently in the literature~\cite{AarArk13ToC, JozVan14QIC, MorFujFit14PRL, TakYamTan14QIC, Bro15PRA, TakTanYamTan15COCOON}.
This also weakens
the complexity assumption necessary to prove the hardness results for such models,
including the IQP model~\cite{BreJozShe11RSPA} and the Boson sampling~\cite{AarArk13ToC}
(the polynomial-time hierarchy now collapses to the second level,
rather than the third level when using $\postBQP$).
Moreover, the hardness results for such models now hold for any constant multiplicative error~${c \geq 1}$,
rather than only for $c$ satisfying~${1 \leq c < \sqrt{2}}$
as in Refs.~\cite{BreJozShe11RSPA, MorFujFit14PRL}.


\section{Conclusion}
\label{Section: conclusion}

This paper has developed several error-reduction methods for quantum computation with few clean qubits,
which simultaneously reduce the number of necessary clean qubits to just one or two.
Using such possibilities of error-reduction,
this paper has shown that the \problemfont{Trace Estimation} problem is complete
for $\BQlogP$ and $\BQoneP$.
One of the technical tools has also been used to show
the hardness of weak classical simulations of $\DQC{1}$ computations.
A few open problems are listed below concerning the power of quantum computation with few clean qubits:

\begin{itemize}
\item
  In the case of one-sided bounded error,
  can any quantum computation with logarithmically many clean qubits
  be made to have exponentially small one-sided error
  by just using one clean qubit, rather than two?
  A similar question may be asked even in the case of two-sided bounded error
  whether both completeness and soundness errors can be made exponentially small simultaneously
  by just using one clean qubit.
\item
  In the two-sided error case,
  is error-reduction possible
  even when a starting quantum computation with few clean qubits
  has only an inverse-polynomial gap between completeness and soundness?
\item
  Are $\DQC{1}$ computations provable to be hard
  to classically simulate under a more desirable notion of simulatability,
  like those discussed in Refs.~\cite{BreJozShe11RSPA, AarArk13ToC}?
  Such results are known for other computation models like Boson sampling
  by assuming hardness of some computational problems~\cite{AarArk13ToC, BreMonShe15arXiv, FefUma15arXiv}.
\end{itemize}


\subsection*{Acknowledgements}

The authors are grateful to Richard Cleve, Fran\c{c}ois Le Gall, Keiji Matsumoto, 
and Yasuhiro Takahashi for very useful discussions.
Keisuke Fujii is supported by
the Grant-in-Aid for Research Activity Start-up~No.~25887034
of the Japan Society for the Promotion of Science. 
Hirotada Kobayashi and Harumichi Nishimura are supported by
the Grant-in-Aid for Scientific Research~(A)~No.~24240001
of the Japan Society for the Promotion of Science.
Tomoyuki Morimae is supported by
the Program to Disseminate Tenure Tracking System
of the Ministry of Education, Culture, Sports, Science and Technology in Japan,
the Grant-in-Aid for Scientific Research on Innovative Areas~No.~15H00850
of the Ministry of Education, Culture, Sports, Science and Technology in Japan,
and
the Grant-in-Aid for Young Scientists~(B)~No.~26730003
of the Japan Society for the Promotion of Science.
Harumichi Nishimura is also supported by
the Grant-in-Aid for Scientific Research on Innovative Areas~No.~24106009
of the Ministry of Education, Culture, Sports, Science and Technology in Japan,
which Hirotada Kobayashi and Seiichiro Tani are also grateful to.
Harumichi Nishimura further acknowledges support from
the Grant-in-Aid for Scientific Research~(C)~No.~25330012
of the Japan Society for the Promotion of Science.
Part of the work of Shuhei Tamate was done
while this author was at the RIKEN Center for Emergent Matter Science, Wako, Saitama, Japan.


\newcommand{\etalchar}[1]{$^{#1}$}

\end{document}